\newif\ifextended 
\newif\ifccblock 
\newif\ifauthor 
\newif\iforcid 
\newif\ifcolor 
\newif\ifthesis 
\newif\iftprint 

\extendedtrue\ccblockfalse\authortrue\orcidtrue\colortrue\thesisfalse\tprintfalse

\documentclass[runningheads]{llncs}

\newcommand{\gfivewidth}{169mm}
\newcommand{\gfiveheight}{239mm}
\newlength{\evenmargin}
\ifthesis
  \usepackage[paperwidth=\gfivewidth, paperheight=\gfiveheight, textwidth=\textwidth, textheight=\textheight \iftprint\else, inner=\evenmargin, outer=\evenmargin\fi]{geometry}
\fi

\usepackage{graphicx}
\makeatletter
\ifcolor
\RequirePackage[bookmarks,unicode,colorlinks=true]{hyperref}%
   \def\@citecolor{blue}%
   \def\@urlcolor{blue}%
   \def\@linkcolor{blue}%
\fi

\iforcid
\def\orcidID#1{\smash{\href{http://orcid.org/#1}{\protect\raisebox{-1.25pt}{\protect\includegraphics{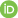}}}}}
\else
\def\orcidID#1{}
\fi
\makeatother

\usepackage{subcaption}

\usepackage{enumitem}

\usepackage[utf8]{inputenc}

\usepackage{pgfplots}
\pgfplotsset{compat=1.15}
\usepgfplotslibrary{statistics}

\usepackage{algorithm}

\newtheorem{condition}{Condition}

\usepackage{listings}
\usepackage{glossaries}
\makeglossaries%
\newacronym{smc}{SMC}{sequential Monte Carlo}
\newacronym{mcmc}{MCMC}{Markov chain Monte Carlo}
\newacronym{ppl}{PPL}{probabilistic programming language}
\newacronym{bpf}{BPF}{bootstrap particle filter}
\glsdisablehyper

\newcommand{\llangle}{\langle\!\langle}
\newcommand{\rrangle}{\rangle\!\rangle}

\newcommand{\ttt}[1]{{\normalfont\texttt{#1}}}
\newcommand{\x}[1]{{\normalfont\textit{#1}}}

\newcommand{\true}{\mathit{true}}
\newcommand{\false}{\mathit{false}}

\newcommand{\termgeo}{{\term_{\mathit{geo}}}}
\newcommand{\termobs}{{\term_{\mathit{obs}}}}
\newcommand{\termair}{{\term_{\mathit{air}}}}
\newcommand{\termseq}{{\term_{\mathit{seq}}}}
\newcommand{\termloop}{{\term_{\mathit{loop}}}}
\newcommand{\termunit}{{\term_{\mathit{unit}}}}

\newcommand{\id}{{\normalfont\text{id}}}
\newcommand{\term}{\mathbf{t}}
\newcommand{\val}{\mathbf{v}}
\newcommand{\stopt}{\mathbf{t}_{\x{stop}}}
\newcommand{\redex}{\mathbf{r}}
\newcommand{\econt}{\mathbf{E}}

\newcommand\T{\mathbb{T}}
\newcommand\R{\mathbb{R}}
\newcommand\X{\mathbb{X}}
\newcommand\Tr{\mathbb{S}}

\newcommand\Tcal{\mathcal{T}}
\newcommand\Bor{\mathcal{B}}
\newcommand\Xcal{\mathcal{X}}
\newcommand\Trcal{\mathcal{S}}

\newcommand{\dr}{{d_\R}}
\newcommand{\drpos}{{d_{\R_+}}}
\newcommand{\drn}[1]{{d_{\R^{#1}}}}
\newcommand{\ds}{{d_\Tr}}
\newcommand{\dt}{{d_\T}}
\newcommand{\de}{{d_\mathbb{E}}}

\newcommand{\trace}[1]{(#1)_\Tr}

\newcommand{\repeatlemma}[2]{%
  \begingroup
  \def\thelemma{\ref{#1}}
  \begin{lemma}
    #2
  \end{lemma}
  \addtocounter{lemma}{-1}
  \endgroup
}

\newcommand{\theoremmain}{%
  If there is $N \in \mathbb{N}$ such that $f_{\term,n} = f_\term$ whenever $n > N$, then
  $
    \llangle \term \rrangle_n = \llangle \term
    \rrangle
  $
  for all $n > N$.
}

\newcommand{\theoremdct}{%
  Assume that $\lim_{n \to \infty} f_{\term,n} = f_\term$ holds pointwise $\mu_\Tr$-ae.
  Furthermore, assume that there exists a measurable function $g : (\Tr,\Trcal) \to (\R_+,\Bor_+)$ such that $f_{\term,n} \leq g$ $\mu_\Tr$-ae for all $n$, and $\int_\Tr g(s) d\mu_\Tr(s) < \infty$.
  Then
  $
    \lim_{n\to\infty} \llangle \term \rrangle_n = \llangle \term \rrangle
  $
  pointwise.
}

\newcommand{\lemmatracemeas}{%
  $(\Tr,\Trcal)$ is a measurable space.%
}

\newcommand{\lemmatracemeasspace}{%
  $(\Tr,\Trcal,\mu_\Tr)$ is a measure space.
  Furthermore, $\mu_\Tr$ is $\sigma$-finite.%
}

\newcommand{\lemmatermmeas}{%
  $(\T, \Tcal)$ is a measurable space.%
}

\newcommand{\lemmarfmeas}{%
  $r_\term : (\Tr, \Trcal) \to (\T,\Tcal)$
  and
  $f_\term : (\Tr, \Trcal) \to (\R_+,\Bor_+)$ 
  are measurable.%
}

\newcommand{\lemmarfnmeas}{%
  $r_{\term,n} : (\Tr, \Trcal) \to (\T,\Tcal)$
  and
  $f_{\term,n} : (\Tr, \Trcal) \to (\R_+,\Bor_+)$
  are measurable.%
}

\newcommand{\lemmaprependmeas}{%
  The functions $\x{prepend}_s : (\Tr,\Trcal) \to (\Tr,\Trcal)$ are
  measurable.%
}

\newcommand{\lemmafinitekernel}{%
  The functions $k_{\term,n} : \Tr \times \Trcal \to \R_+$ are sub-probability kernels.%
}

\newcommand{\lemmadecompose}{%
  Let $n > 0$. If $f_{\term,n}(s) > 0$, then $f_{\term,n}(s) = f_{\term,n-1}(\underline{s})f_{r_{\term,n-1}(\underline{s}),1}(\overline{s})$ for exactly one decomposition $\underline{s} \ast \overline{s} = s$.
  If $f_{\term,n}(s) = 0$, then $f_{\term,n-1}(\underline{s})f_{r_{\term,n-1}(\underline{s}),1}(\overline{s}) = 0$ for all decompositions $\underline{s} \ast \overline{s} = s$.
  As a consequence, if $f_{\term,n}(s) > 0$, then $p_{r_{\term,n-1}(\underline{s}),1}(\overline{s}) = 1$.%
}

\newcommand{\lemmapropdensity}{%
  For $n \in \mathbb{N}$,
  $
    \langle \term \rangle_n(S) =
    \int_S
    f_{\term,n-1}(\underline{s})
    p_{r_{\term,n-1}(\underline{s}),1}(\overline{s})
    d\mu_\Tr(s),
  $
  where $\underline{s} \ast \overline{s} = s$ is the unique decomposition from Lemma~\ref{lemma:decompose}.%
}

\newcommand{\lemmapropfin}{%
  $\langle \term \rangle_0$ is a sub-probability measure.
  Also, if $\llangle \term \rrangle_{n-1}$ is finite, then $\langle \term \rangle_n$ is finite.%
}

\newcommand{\lemmaweights}{%
  $\displaystyle
    w_n(s)
    = \frac{f_{\llangle \term \rrangle_n}(s)}{f_{\langle \term \rangle_n}(s)}
    = \begin{cases}
      f_{r_{\term,n-1}(\underline{s}),1}(\overline{s}) & \text{if } n > 0 \\
      f_{\term,0}(s) & \text{if } n = 0 \\
    \end{cases}
  $
  when $f_{\langle \term \rangle_n}(s) > 0$.
  Here, $\underline{s} \ast \overline{s} = s$ is the unique decomposition from Lemma~\ref{lemma:decompose}.%
}

\usepackage[absolute]{textpos}

\usepackage{tikz}
\usetikzlibrary{positioning,shapes,calc,patterns}

\usepackage{mathtools}
\usepackage{stmaryrd}
\usepackage{amsfonts}
\usepackage{amssymb}

\ifcolor\else\usepackage[hidelinks]{hyperref}\fi

\usepackage{marvosym}

\begin{document}


\ifauthor
\begin{textblock*}{1.2\textwidth}(1in + \oddsidemargin - 0.1\textwidth,10pt)
  \noindent
  \scriptsize
  This is an author-prepared version\ifextended, extended with appendices and minor additions to the main text\fi.
  © The Author(s) 2021.
  This version of the contribution\ifextended, except the extensions,\fi\ has been accepted for publication at ESOP 2021, after peer review, but is not the Version of Record (the version published by Springer) and does not reflect post-acceptance improvements, or any corrections.
  The Version of Record is available online at: \url{https://doi.org/10.1007/978-3-030-72019-3_15}.
\end{textblock*}
\fi

\title{%
  Correctness of Sequential Monte Carlo Inference for Probabilistic Programming
  Languages%
  \thanks{%
    This project is financially supported by the Swedish Foundation for
    Strategic Research (ASSEMBLE RIT15-0012) and the Swedish Research Council
    (grant 2013-4853).%
  }
}
\titlerunning{Correctness of Sequential Monte Carlo for Probabilistic Programming}
%
\author{%
  Daniel Lundén(\Letter)\inst{1}\orcidID{0000-0003-3127-5640} \and
  Johannes Borgström\inst{2}\orcidID{0000-0001-5990-5742} \and
  David Broman\inst{1}\orcidID{0000-0001-8457-4105}
}
\authorrunning{D. Lundén et al.}
%
\institute{Digital Futures and EECS,\\KTH Royal Institute of Technology, Stockholm, Sweden\\
  \email{\{dlunde,dbro\}@kth.se}
  \and
  Uppsala University, Uppsala, Sweden\\
  \email{johannes.borgstrom@it.uu.se}
}

\maketitle              

\begin{abstract}
  Probabilistic programming is an approach to reasoning under uncertainty by encoding inference problems as programs.
In order to solve these inference problems, \glspl{ppl} employ different inference algorithms, such as \gls{smc}, \gls{mcmc}, or variational methods.
Existing research on such algorithms mainly concerns their implementation and efficiency, rather than the correctness of the algorithms themselves when applied in the context of expressive \glspl{ppl}.
To remedy this, we give a correctness proof for \gls{smc} methods in the context of an expressive \gls{ppl} calculus, representative of popular \glspl{ppl} such as WebPPL, Anglican, and Birch.
Previous work have studied correctness of \gls{mcmc} using an operational semantics, and correctness of \gls{smc} and \gls{mcmc} in a denotational setting without term recursion.
However, for \gls{smc} inference---one of the most commonly used algorithms in \glspl{ppl} as of today---no formal correctness proof exists in an operational setting.
In particular, an open question is if the resample locations in a probabilistic program affects the correctness of \gls{smc}.
We solve this fundamental problem, and make four novel contributions:
(i) we extend an untyped \gls{ppl} lambda calculus and operational semantics
to include explicit resample terms, expressing synchronization points in \gls{smc} inference;
(ii) we prove, for the first time, that subject to mild restrictions, any placement of the explicit resample terms is valid for a generic form of \gls{smc} inference;
(iii) as a result of (ii), our calculus benefits from classic results from the \gls{smc} literature: a law of large numbers and an unbiased estimate of the model evidence; and
(iv) we formalize the bootstrap particle filter for the calculus and discuss how our results can be further extended to other \gls{smc} algorithms.

\keywords{%
  Probabilistic Programming \and
  Sequential Monte Carlo \and
  Operational Semantics \and
  Functional Programming \and
  Measure Theory
}
\end{abstract}
%
%
%

\glsresetall%

\section{Introduction}\label{sec:intro}
Probabilistic programming is a programming paradigm for probabilistic models, encompassing a wide range of programming languages, libraries, and platforms~\cite{carpenter2017stan,goodman2008church,goodman2014design,murray2018automated,scibior2018functional,design2016tolpin,tran2016edward}.
Such probabilistic models are typically created to express \emph{inference problems}, which are ubiquitous and highly significant in, for instance, machine learning \cite{bishop2006pattern}, artificial intelligence \cite{russell2009artificial}, phylogenetics \cite{ronquist2003mrbayes,ronquist2020probabilistic}, and topic modeling \cite{blei2003latent}.

In order to solve such inference problems, an \emph{inference algorithm} is required.
Common general-purpose algorithm choices for inference problems include \emph{\gls{smc}} methods~\cite{doucet2001sequential}, \emph{\Gls{mcmc}} methods~\cite{gilks1995markov}, and \emph{variational} methods~\cite{wainwright2008graphical}.
In traditional settings, correctness results for such algorithms often come in the form of laws of large numbers, central limit theorems, or optimality arguments.
However, for general-purpose \glspl{ppl}, the emphasis has predominantly been on algorithm implementations and their efficiency~\cite{goodman2014design,murray2018automated,design2016tolpin}, rather than the correctness of the algorithms themselves.
In particular, explicit connections between traditional theoretical \gls{smc} results and \gls{ppl} semantics have been limited.
In this paper, we bridge this gap by formally connecting fundamental \gls{smc} results to the context of an expressive \gls{ppl} calculus.

Essentially, \gls{smc} works by simulating many executions of a probabilistic program concurrently, occasionally \emph{resampling} the different executions.
In this resampling step, \gls{smc} discards less likely executions, and replicates more likely executions, while remembering the average likelihood at each resampling step in order to estimate the overall likelihood.
In expressive \glspl{ppl}, there is freedom in choosing where in a program this resampling occurs.
For example, most \gls{smc} implementations, such as WebPPL~\cite{goodman2014design}, Anglican~\cite{wood2014new}, and Birch~\cite{murray2018automated}, always resample when all executions have reached a call to the \emph{weighting} construct in the language.
At possible resampling locations, Anglican takes a conservative approach by dynamically checking during runtime if all executions have either stopped at a weighting construct, or all have finished. If none of these two cases apply, report a runtime error.
In contrast, WebPPL does not perform any checks and simply includes the executions that have finished in the resampling step.  There are also heuristic approaches~\cite{lunden2018automatic} that automatically \emph{align} resampling locations in programs, ensuring that all executions finish after encountering the same number of them.
The motivations for using the above approaches are all based on experimental validation.
As such, an open research problem is whether there are any inherent restrictions when selecting resampling locations, or if the correctness of \gls{smc} is independent of this selection.
This is not only important theoretically to guarantee the correctness of inference results, but also for inference performance, both since inference performance is affected by the locations of resampling locations~\cite{lunden2018automatic} and since dynamic checks result in direct runtime overhead.
We address this research problem in this paper.


In the following, we give an overview of the paper and our contributions.
In Section~\ref{sec:motivation}, we begin by giving a motivating example from phylogenetics, illustrating the usefulness of our results.
Next, in Section~\ref{sec:calculus}, we define the syntax and operational semantics of an expressive functional \gls{ppl} calculus based on the operational formalization in Borgström et al.~\cite{borgstrom2016lambda}, representative of common \glspl{ppl}.
The operational semantics assign to each pair of term $\term$ and initial random \emph{trace} (sequences of random samples) a non-negative weight.
This weight is accumulated during evaluation through a \ttt{weight} construct, which, in current calculi and implementations of \gls{smc}, is (implicitly) always followed by a resampling.
To decouple resampling from weighting, we present our first contribution.
\begin{enumerate}[align=left]
  \item[(i)]
    We extend the calculus from Borgström et al.~\cite{borgstrom2016lambda} to include explicit \texttt{resample} terms, expressing explicit synchronization points for performing resampling in \gls{smc}.
    With this extension, we also define a semantics which limits the number of evaluated resample terms, laying the foundation for the remaining contributions.
\end{enumerate}

In Section~\ref{sec:targetmeas}, we define the probabilistic semantics of the calculus.
The weight from the operational semantics is used to define unnormalized distributions $\llangle \term \rrangle$ over traces and $\llbracket \term \rrbracket$ over result terms.
The measure $\llbracket \term \rrbracket$ is called the \emph{target measure}, and finding a representation of this is the main objective of inference algorithms. 

We give a formal definition of \gls{smc} inference based on Chopin~\cite{chopin2004central} in Section~\ref{sec:formalsmc}.
This includes both a generic \gls{smc} algorithm, and two standard correctness results from the \gls{smc} literature: a law of large numbers~\cite{chopin2004central}, and the unbiasedness of the likelihood estimate~\cite{naesseth2019elements}.

In Section~\ref{sec:smcppl}, we proceed to present the main contributions.
\begin{enumerate}[align=left]
  \item[(ii)]
    From the \gls{smc} formulation by Chopin~\cite{chopin2004central}, we formalize a sequence of distributions $\llangle\term\rrangle_n$, indexed by $n$, such that $\llangle\term\rrangle_n$ allows for evaluating at most $n$ \ttt{resample}s.
    This sequence is determined by the placement of \ttt{resample}s in $\term$.
    Our first result is Theorem~\ref{theorem:main}, showing that $\llangle \term \rrangle_n$ eventually equals $\llangle \term \rrangle$ if the number of calls to \texttt{resample} is upper bounded.
    Because of the explicit \ttt{resample} construct, this also implies that, for \emph{all} \ttt{resample} placements such that the number of calls to \ttt{resample} is upper bounded, $\llangle\term\rrangle_n$ eventually equals $\llangle\term\rrangle$.
    We further relax the finite upper bound restriction and investigate under which conditions $\lim_{n \to \infty} \llangle \term \rrangle_n = \llangle \term \rrangle$ pointwise.
    In particular, we relate this equality to the dominated convergence theorem in Theorem~\ref{theorem:dct}, which states that the limit converges as long as there exists a function dominating the weights encountered during evaluation.
    This gives an alternative set of conditions under which $\llangle \term \rrangle_n$ converges to $\llangle\term\rrangle$ (now asymptotically, in the number of resamplings $n$).
\end{enumerate}
The contribution is fundamental, in that it provides us with a sequence of approximating distributions $\llangle\term\rrangle_n$ of $\llangle\term\rrangle$ that can be targeted by the \gls{smc} algorithm of Section~\ref{sec:formalsmc}.
As a consequence, we can extend the standard correctness results of that section to our calculus.
This is our next contribution.
\begin{enumerate}[align=left]
  \item[(iii)]
    Given a suitable sequence of transition kernels (ways of moving between the $\llangle \term \rrangle_n$), we can \emph{correctly} approximate $\llangle \term \rrangle_n$ with the \gls{smc} algorithm from Section~\ref{sec:formalsmc}.
    The approximation is correct in the sense of Section~\ref{sec:formalsmc}: the law of large numbers and the unbiasedness of the likelihood estimate holds.
    As a consequence of (ii), \gls{smc} also correctly approximates $\llangle \term \rrangle$, and in turn the target measure $\llbracket \term \rrbracket$.
    Crucially, this also means estimating the model evidence (likelihood), which allows for compositionality~\cite{gordon2013model} and comparisons between different models~\cite{ronquist2020probabilistic}.
    This contribution is summarized in Theorem~\ref{theorem:consistency}.
\end{enumerate}
Related to the above contributions, \'{S}cibior et al.~\cite{scibior2017denotational} formalizes \gls{smc} and \gls{mcmc} inference as transformations over monadic inference representations using a denotational approach (in contrast to our operational approach).
They prove that their \gls{smc} transformations preserve the measure of the initial representation of the program (i.e., the target measure).
Furthermore, their formalization is based on a simply-typed lambda calculus with primitive recursion, while our formalization is based on an untyped lambda calculus which naturally supports full term recursion.
Our approach is also rather more elementary, only requiring basic measure theory compared to the relatively heavy mathematics (category theory and synthetic measure theory) used by them.
Regarding generalizability, their approach is both general and compositional in the different inference transformations, while we abstract over parts of the SMC algorithm.
This allows us, in particular, to relate directly to standard \gls{smc} correctness results.

Section~\ref{sec:smckernel} concerns the instantiation of the transition kernels from (iii), and also discusses other \gls{smc} algorithms.
Our last contribution is the following.
\begin{enumerate}[align=left]
  \item[(iv)]
    We define a sequence of sub-probability kernels $k_{\term,n}$ induced by a given program $\term$, corresponding to the fundamental \gls{smc} algorithm known as the \emph{\gls{bpf}} for our calculus.
    This is the most common version of \gls{smc}, and we present a concrete \gls{smc} algorithm corresponding to these kernels.
    We also discuss other \gls{smc} algorithms and their relation to our formalization: the resample-move~\cite{gilks2001following}, alive~\cite{kudlicka2019probabilistic}, and auxiliary~\cite{pitt1999filtering} particle filters.
\end{enumerate}
Importantly, by combining the above contributions, we justify that the
implementation strategies of the \glspl{bpf} in WebPPL, Anglican, and Birch are
indeed correct. In fact, our results show that the strategy in Anglican, in
which every evaluation path must resample the same number of times, is too
conservative.

\ifextended
Detailed proofs for many lemmas found in the paper are available in the appendix.
These lemmas are explicitly marked with $^\dagger$.
\fi

\section{A Motivating Example from Phylogenetics}\label{sec:motivation}
\begin{figure}[tb]
  \lstset{%
    basicstyle=\ttfamily\scriptsize,
    numbers=left, showlines=true,
    frame=leftline,
    numberstyle=\tiny,
    numbersep=3pt, 
    framexleftmargin=-2pt,
    xleftmargin=2em,
  }
  \centering
  \noindent\begin{minipage}{.50\textwidth}
    \begin{lstlisting}[name=crbd]
let tree = {
 left:{left:{age:0},right:{age:0},age:4},
 right:{left:{age:0},right:{age:0},age:6},
 age:10
} in

let lambda = 0.2 in let mu = 0.1 in

let crbdGoesExtinct startTime =
 let curTime = startTime
   - (sample (exponential (lambda + mu)))
 in
 if curTime < 0 then false
 else
  let speciation = sample
    (bernoulli (lambda / (lambda + mu))) in
  if !speciation then true
  else crbdGoesExtinct curTime
    && crbdGoesExtinct curTime in

    \end{lstlisting}
  \end{minipage}\hfill
  \begin{minipage}{.50\textwidth}
    \begin{lstlisting}[name=crbd]
let simBranch startTime stopTime =
 let curTime = startTime -
   sample (exponential lambda) in
 if curTime < stopTime then ()
 else if not (crbdGoesExtinct curTime)
 then weight (log 0) // #1
 else (weight (log 2); // #2
       simBranch curTime stopTime) in

let simTree tree parent =
 let w = -mu * (parent.age - tree.age) in
 weight w; // #3
 simBranch parent.age tree.age;
 match tree with
 | {left,right,age} ->
   simTree left tree; simTree right tree
 | {age} -> () in

simTree tree.left tree;
simTree tree.right tree
    \end{lstlisting}
  \end{minipage}
  \caption{%
    A simplified version of a phylogenetic birth-death model
    from~\cite{ronquist2020probabilistic}.
    See the text for a description.
  }
  \label{fig:crbd}
\end{figure}

In this section, we give a motivating example from phylogenetics.
The example is written in a functional \gls{ppl}%
\footnote{%
  The implementation is an interpreter written in OCaml. It largely follows the same approach as Anglican and WebPPL, and uses continuation-passing style in order to pause and resume executions as part of inference.
  It is available at \url{https://github.com/miking-lang/miking-dppl/tree/pplcore}.
  The example in Fig.~\ref{fig:crbd} can be found under \ttt{examples/crbd/crbd-esop.ppl}
} developed as part of this paper, in order
to verify and experiment with the presented concepts and results.
In particular, this \gls{ppl} supports \gls{smc} inference (Algorithm~\ref{alg:smcppl}) with decoupled \ttt{resample}s and \ttt{weight}s\footnote{The implementation uses log weights as arguments to \ttt{weight} for numerical reasons.}, as well as sampling from random distributions with a \ttt{sample} construct.

Consider the program in Fig.~\ref{fig:crbd},
encoding a simplified version of a phylogenetic birth-death model (see  Ronquist et al.~\cite{ronquist2020probabilistic} for the full version).
The problem is to find the model evidence for a particular birth rate (\texttt{lambda = 0.2}) and death rate (\texttt{mu = 0.1}), given an observed phylogenetic \texttt{tree}.
The tree represents known lineages of evolution, where the leaves are
extant (surviving to the present) species.
Most importantly, for illustrating the usefulness of the results in this paper, the recursive function \ttt{simBranch}, with its two \ttt{weight} applications \ttt{\#1} and \ttt{\#2}, is called a random number of times for each branch in the observed \ttt{tree}.
Thus, different \gls{smc} executions encounter differing numbers of calls to \ttt{weight}.
When resampling is performed after every call to \ttt{weight} (\ttt{\#1}, \ttt{\#2}, and \ttt{\#3}), it is, because of the differing numbers of resamples, not obvious that inference is correct (e.g., the equivalent program in Anglican gives a runtime error).
Our results show that such a resampling strategy is indeed correct.

This strategy is far from optimal, however. For instance, only resampling at \ttt{\#3}, which is encountered the same number of times in each execution, performs much better \cite{lunden2018automatic,ronquist2020probabilistic}.
Our results show that this is correct as well, and that it gives the same asymptotic results as the naive strategy in the previous paragraph.

Another strategy is to resample only at \ttt{\#1} and \ttt{\#3}, again causing executions to encounter differing numbers of resamples.
Because \ttt{\#1} weights with (log) 0, this approach gives the same accuracy as resampling only at \ttt{\#3}, but avoids useless computation since a zero-weight execution can never obtain non-zero weight.
Equivalently to resampling at \ttt{\#1}, zero-weight executions can also be identified and stopped automatically at runtime.
This gives a direct performance gain, and both are correct by our results.
We compared the three strategies above for \gls{smc} inference with 50\,000 particles\footnote{%
  We repeated each experiment 20 times on a machine running Ubuntu 20.04 with an Intel i5-2500K CPU (4 cores) and 8GB memory. The standard deviation was under 0.1 seconds in all three cases.
}:
resampling at \ttt{\#1},\ttt{\#2}, and \ttt{\#3} resulted in a runtime of 15.0 seconds, at \ttt{\#3} in a runtime of 12.6 seconds, and at \ttt{\#1} and \ttt{\#3} in a runtime of 11.2 seconds.
Furthermore, resampling at \ttt{\#1},\ttt{\#2}, and \ttt{\#3} resulted in significantly worse accuracy compared to the other two strategies \cite{lunden2018automatic,ronquist2020probabilistic}.

Summarizing the above, the results in this paper ensure correctness when exploring different resampling placement strategies.
As just demonstrated, this is useful, because resampling strategies can have a large impact on SMC accuracy and performance.

\section{A Calculus for Probabilistic Programming Languages}\label{sec:calculus}
In this section, we define the calculus used throughout the paper.
\ifextended
In Section~\ref{sec:syntax}, we begin by defining the syntax, and demonstrate how simple probability distributions can be encoded using it.
In Section~\ref{sec:semantics}, we define the semantics and demonstrate it on the previously encoded probability distributions.
\fi
This semantics is used in Section~\ref{sec:targetmeas} to define the \emph{target measure} for any given program.
In Section~\ref{sec:ressemantics}, we extend the semantics of Section~\ref{sec:semantics} to limit the number of allowed resamples in an evaluation. This extended semantics forms the foundation for formalizing \gls{smc} in Sections~\ref{sec:smcppl} and~\ref{sec:smckernel}.

\subsection{Syntax}\label{sec:syntax}

The main difference between the calculus presented in this section and the standard untyped lambda calculus is the addition of real numbers, functions operating on real numbers, a sampling construct for drawing random values from real-valued probability distributions, and a construct for weighting executions.
The rationale for making these additions is that, in addition to discrete probability distributions, continuous distributions are ubiquitous in most real-world models, and the weighting construct is essential for encoding inference problems.
In order to define the calculus, we let
$X$ be a countable set of variable names;
$D \in \mathbb{D}$ range over a countable set $\mathbb{D}$ of identifiers for families of probability distributions over $\R$, where the family for each identifier $D$ has a fixed number of real parameters $|D|$;
and $g \in \mathbb{G}$ range over a countable set $\mathbb{G}$ of identifiers for real-valued functions with respective arities $|g|$.
More precisely, for each $g$, there is a measurable function $\sigma_g : \R^{|g|} \to \R$.
For simplicity, we often use $g$ to denote both the identifier and its measurable function.
We can now give an inductive definition of the abstract syntax, consisting of values $\val$ and terms $\term$.
\begin{definition}\label{def:syntax}
  \begin{equation}
      \val \Coloneqq
      c \enspace | \enspace
      \lambda x. \term
      \quad
    \begin{aligned}
      \term \Coloneqq &\enspace
      \val
      \enspace | \enspace
      x
      \enspace | \enspace
      \term \enspace \term
      \enspace | \enspace
      \ttt{if} \ \term \ \ttt{then} \ \term \ \ttt{else} \ \term
      \enspace | \enspace
      g(\term_1, \ldots, \term_{|g|} )
      \\& \enspace | \enspace
      \ttt{sample}_D(\term_1, \ldots, \term_{|D|} )
      \enspace | \enspace
      \ttt{weight}(\term)
      \enspace | \enspace
      \ttt{resample}
    \end{aligned}
  \end{equation}
  Here, $c \in \R$, $x \in X$, $D \in \mathbb{D}$, $g \in \mathbb{G}$.
  We denote the set of all terms
  by $\T$ and the set of all values by $\mathbb{V}$.
\end{definition}
The formal semantics is given in Section~\ref{sec:semantics}.
Here, we instead give an informal description of the various language constructs.

Some examples of distribution identifiers are $\mathcal{N} \in \mathbb{D}$, the identifier for the family of normal distributions, and $\mathcal{U} \in \mathbb{D}$, the identifier for the family of continuous uniform distributions.
The semantics of the term $\ttt{sample}_\mathcal{N}(0,1)$ is, informally, ``draw a random sample from the normal distribution with mean 0 and variance 1''.
The \ttt{weight} construct is illustrated later in this section, and we discuss the $\ttt{resample}$ construct in detail in Sections~\ref{sec:ressemantics} and \ref{sec:smcppl}.

We use common syntactic sugar throughout the paper.
Most importantly, we use $\false$ and $\true$ as aliases for 0 and 1, respectively, and $()$ (unit) as another alias for 0.
Furthermore, we often write $g \in \mathbb{G}$ as infix operators.
For instance, $1 + 2$ is a valid term, where $+ \in \mathbb{G}$.
Now, let $\R_+$ denote the non-negative reals.
We define $f_D : \R^{|D|+1} \to \R_+$ as the function $f_D \in \mathbb{G}$ such that $f_D(c_1,\ldots,c_{|D|},\cdot)$ is the probability \emph{density} (continuous distribution) or \emph{mass} function (discrete distribution) for the probability distribution corresponding to $D \in \mathbb{D}$ and $(c_1,\ldots,c_{|D|})$.
For example, $f_\mathcal{N}(0,1,x) = \frac{1}{\sqrt{2\pi}}\cdot e^{-\frac{1}{2}\cdot x^2}$ is the standard probability density of the normal distribution with mean 0 and variance 1.
Lastly, we will also use \ttt{let} bindings, \ttt{let rec} bindings, sequencing using \ttt{;}, and lists (all of which can be encoded in the calculus).
Sequencing is required for the side-effects produced by \ttt{weight} (see Definition~\ref{def:stochsem}) and \ttt{resample} (see Sections~\ref{sec:ressemantics} and~\ref{sec:smcppl}).

\ifextended
The explicit \ttt{if} expressions in the language deserve special mention---as is well known, they can also be encoded in the lambda calculus.
The reason for explicitly including them in the calculus is to connect the lambda calculus to the continuous parts of the language.
That is, we need a way of making control flow depend on the result of calculations on real numbers (e.g., \lstinline!if $c_1 < c_2$ then $\term_1$ else $\term_2$!, where $c_1$ and $c_2$ are real numbers).
An alternative to adding \ttt{if}-expressions is to let comparison functions in $\mathbb{G}$ return Church Booleans, but this requires extending the codomain of primitive functions.
\fi

\pgfplotsset{%
  probplot/.style={%
    width=50mm,
    height=28mm,
    tick style={draw=none},
    tick label style={font=\scriptsize},
    label style={font=\scriptsize},
    x label style={at={(axis description cs:0.5,-0.3)},anchor=north},
  }
}

\ifextended
\begin{figure}[tb]
  \centering
  \begin{subfigure}[c]{0.45\textwidth}
    \centering
    \begin{tabular}{c}
      \begin{lstlisting}
let rec $\x{geometric}$ $\_$ =
  if sample$_{\x{Bern}}(0.6)$ then
    $1$ + $\x{geometric}$ $()$
  else $1$
in $\x{geometric}$ $()$
      \end{lstlisting}
    \end{tabular}
    \caption{}
    \label{fig:geometric:code}
  \end{subfigure}
  \begin{subfigure}[c]{0.45\textwidth}
    \centering
    \begin{tikzpicture}[trim axis left, trim axis right]
      \begin{axis}[%
        probplot,
        y label style={at={(axis description cs:-0.12,0.5)},anchor=south},
        ymin=0,
        xmax=11,
        ylabel=Probability,
        xlabel=Outcome,
        xtick={1,2,3,4,5,6,7,8,9},
        extra x ticks={10},
        extra x tick style={tick label style={xshift=1pt,yshift=-2.5pt}},
        extra x tick labels={$\cdots$},
        bar width=7pt,
        ]
        \addplot [
          fill=gray,
          ybar,
          ] coordinates {%
            (1,0.41) (2,0.6*0.4) (3,0.6^2*0.4)
            (4,0.6^3*0.4) (5,0.6^4*0.4) (6,0.6^5*0.4)
            (7,0.6^6*0.4) (8,0.6^7*0.4) (9,0.6^8*0.4)
            (10,0.6^9*0.4)
          };
      \end{axis}
    \end{tikzpicture}
    \caption{}
    \label{fig:geometric:graph}
  \end{subfigure}
  \caption{%
    The geometric distribution as a program $\termgeo$ in (a), and visualized in (b).
  }
  \label{fig:geometric}
\end{figure}

We now consider a set of examples.
In Section~\ref{sec:semantics} and Section~\ref{sec:targetmeasdef} these examples will be further considered to illustrate the semantics, and target measure, respectively.
Here, we first give the syntax, and informally discuss and visualize the probability distributions (i.e., the target measures, as we will see in Section~\ref{sec:targetmeasdef}) for the examples.

First, consider the program in Fig.~\ref{fig:geometric:code}.
This program encodes a slight variation on the standard geometric distribution: flip a coin with bias 0.6 (i.e., the flip will result in heads, or $\true$, $60\%$ of the time) until a flip results in tails ($\false$).
The probability distribution is over the number of flips before encountering tails (including the final tails flip), and is illustrated in Fig.~\ref{fig:geometric:graph}.
\fi

\begin{figure}[tb]
  \centering
  \begin{subfigure}[c]{0.19\textwidth}
    \centering
    \begin{tabular}{c}
      \begin{lstlisting}
sample$_{\x{Beta}}(2,2)$
      \end{lstlisting}
    \end{tabular}
    \caption{}
    \label{fig:beta:code}
  \end{subfigure}
  \begin{subfigure}[c]{0.48\textwidth}
    \centering
    \begin{tabular}{c}
      \begin{lstlisting}
let $p$ = sample$_{\x{Beta}}(2,2)$ in
let $\x{observe}$ $o$ =
  weight$(f_\x{Bern}(p,o))$ in
$\x{iter}$ $\x{observe}$ [$\true, \false, \true$]; $p$
      \end{lstlisting}
    \end{tabular}
    \caption{}
    \label{fig:betaobs:code}
  \end{subfigure}
  \begin{subfigure}[c]{0.30\textwidth}
    \centering
    \begin{tikzpicture}[trim axis left,trim axis right]
      \begin{axis}[%
        probplot,
        y label style={at={(axis description cs:-0.07,0.5)},anchor=south},
        ymin=0, ymax=2.5,
        xmin=0, xmax=1,
        ylabel=Density,
        xlabel=Outcome,
        xtick={0,0.5,1},
        ]
        \newcommand{\xalpha}{4}
        \newcommand{\xbeta}{3}
        \addplot [
          dashed,
          domain=0:1
          ] {%
            x^(\xalpha-1) * (1-x)^(\xbeta-1) /
            ((\xalpha-1)!*(\xbeta-1)!/(\xalpha+\xbeta-1)!)
          };
        \renewcommand{\xalpha}{2}
        \renewcommand{\xbeta}{2}
        \addplot [
          domain=0:1
          ] {%
            x^(\xalpha-1) * (1-x)^(\xbeta-1) /
            ((\xalpha-1)!*(\xbeta-1)!/(\xalpha+\xbeta-1)!)
          };
      \end{axis}
    \end{tikzpicture}
    \caption{}
    \label{fig:betaobs:graph}
  \end{subfigure}
  \caption{%
    The $\x{Beta}(2,2)$ distribution as a program in (a), and visualized with a solid line in (c).
    Also, the program $\termobs$ in (b), visualized with a dashed line in (c).
    The $\x{iter}$ function in (b) simply maps the given function over the given list and returns $()$.
    That is, it calls $\x{observe}$ $\true$, $\x{observe}$ $\false$, and $\x{observe}$ $\true$ purely for the side-effect of weighting.
  }
\end{figure}

\ifextended
The geometric distribution is a discrete distribution, meaning that the set of possible outcomes is countable.
We can also encode continuous distributions in the language.
\fi
Consider first the program in Fig.~\ref{fig:beta:code}, directly encoding the $\x{Beta}(2,2)$ distribution, illustrated in Fig.~\ref{fig:betaobs:graph}.
This distribution naturally represents the uncertainty in the bias of a coin---in this case, the coin is most likely unbiased (bias 0.5), and biases closer to 0 and 1 are less likely.
In Fig.~\ref{fig:betaobs:code}, we extend Fig.~\ref{fig:beta:code} by observing the sequence $[\true,\false,\true]$ when flipping the coin.
These observations are encoded using the \ttt{weight} construct, which simply accumulates a product (as a side-effect) of all real-valued arguments given to it throughout the execution.
First, recall the standard mass function ($\sigma_{f_\x{Bern}}(p,\true) = p;\; \sigma_{f_\x{Bern}}(p,\false) = (1-p);\; \sigma_{f_\x{Bern}}(p,x) = 0$ otherwise) for the Bernoulli distribution corresponding to $f_\x{Bern} \in \mathbb{G}$.
The observations $[\true,\false,\true]$ are encoded using the $\x{observe}$ function, which uses the \ttt{weight} construct internally to assign weights to the current value $p$ according to the Bernoulli mass function.
As an example, assume we have drawn $p = 0.4$.
The weight for this execution is
$
\sigma_{f_\x{Bern}}(0.4,\true) \cdot
\sigma_{f_\x{Bern}}(0.4,\false) \cdot
\sigma_{f_\x{Bern}}(0.4,\true) = 0.4^2\cdot0.6
$.
Now consider $p = 0.6$ instead.
For this value of $p$ the weight is instead $0.6^2\cdot0.4$.
This explains the shift in Fig.~\ref{fig:betaobs:graph}---a bias closer to 1 is more likely, since we have observed two $\true$ flips, but only one $\false$.

\subsection{Semantics}\label{sec:semantics}
In this section, we define the semantics of our calculus.
The definition is split into two parts: a \emph{deterministic semantics} and a \emph{stochastic semantics}. We use evaluation contexts to assist in defining our semantics.
The evaluation contexts $\econt$ induce a call-by-value semantics, and are defined as follows.
\begin{definition}
  \begin{equation}
    \begin{aligned}
      \econt \Coloneqq& \enspace
      [\cdot]
      \enspace | \enspace
      \econt \enspace \term
      \enspace | \enspace
      (\lambda x. \term) \enspace \econt
      \enspace | \enspace
      \ttt{if} \ \econt \ \ttt{then} \ \term \ \ttt{else} \ \term
      \\& \enspace | \enspace
      g(c_1, \ldots, c_m, \econt, \term_{m+2}, \ldots, \term_{|g|} )
      \\& \enspace | \enspace
      \ttt{sample}_D(c_1, \ldots, c_m, \econt, \term_{m+2}, \ldots, \term_{|D|})
      \enspace | \enspace
      \ttt{weight}(\econt)
    \end{aligned}
  \end{equation}
  We denote the set of all evaluation contexts by $\mathbb{E}$.
\end{definition}

With the evaluation contexts in place, we proceed to define the \emph{deterministic semantics} through a small-step relation $\rightarrow_\textsc{Det}$.
\begin{definition}
  \begin{equation}
    \begin{gathered}
      \frac{}
      {\econt[(\lambda x. \term) \enspace \val]
      \rightarrow_\textsc{Det} \econt[[x \mapsto \val]\term]}
      (\textsc{App}) \quad
      \frac{c=\sigma_g(c_1, \ldots, c_{|g|})}
      {\econt[%
        g(c_1, \ldots, c_{|g|})]
      \rightarrow_\textsc{Det} \econt[c]}
      (\textsc{Prim}) \\[2mm]
      \frac{}
      {\econt[\ttt{if} \ \mathit{true} \ \ttt{then} \ \term_1 \ \ttt{else} \ \term_2]
      \rightarrow_\textsc{Det} \econt[\term_1]}
      (\textsc{IfTrue}) \\
      \frac{}
      {\econt[%
        \ttt{if} \ \mathit{false} \ \ttt{then} \ \term_1 \ \ttt{else} \ \term_2]
      \rightarrow_\textsc{Det} \econt[\term_2]}
      (\textsc{IfFalse}) \\[2mm]
    \end{gathered}
  \end{equation}
\end{definition}
The rules are straightforward, and will not be discussed in further detail here.
We use the standard notation for transitive and reflexive closures (e.g. $\rightarrow_\textsc{Det}^*$), and transitive closures (e.g. $\rightarrow_\textsc{Det}^+$) of relations throughout the paper.

Following the tradition of Kozen~\cite{kozen1981semantics} and Park et al.~\cite{park2008probabilistic}, sampling in our stochastic semantics works by consuming randomness from a tape of real numbers.
We use inverse transform sampling, and therefore the tape consists of numbers from the interval $[0,1]$.
In order to use inverse transform sampling, we require that for each $D \in \mathbb{D}$, there exists a measurable function $F^{-1}_D : \R^{|D|} \times [0,1] \to \R$, such that $F^{-1}_D(c_1,\ldots,c_{|D|},\cdot)$ is the \emph{inverse cumulative distribution function} for the probability distribution corresponding to $D$ and $(c_1,\ldots,c_{|D|})$.
We call the tape of real numbers a \emph{trace}, and make the following definition.
\begin{definition}
  Let $\mathbb{N}_0 = \mathbb{N} \cup \{ 0 \}$.
  The set of all traces is
  $
  \Tr = \bigcup_{n \in \mathbb{N}_0} [0,1]^n.
  $
\end{definition}
We use the notation $\trace{c_1,c_2,\ldots,c_n}$ to indicate the trace consisting of the $n$ numbers $c_1,c_2,\ldots,c_n$.
Given a trace $s$, we denote by $|s|$ the length of the trace.
We also denote the concatenation of two traces $s$ and $s'$ with $s \ast s'$.
Lastly, we let $c :: s$ denote the extension of the trace $s$ with the real number $c$ as head.

With the traces and $F_D^{-1}$ defined, we can proceed to the stochastic\footnote{Note that the semantics models stochastic behavior, but is itself a deterministic relation.} semantics $\rightarrow$ over $\T \times \R_+ \times \Tr$.
\begin{definition}\label{def:stochsem}
  \begin{equation}
    \stopt \Coloneqq
    \val
    \enspace | \enspace
    \econt[\ttt{sample}_D(c_1, \ldots, c_{|D|} )]
    \enspace | \enspace
    \econt[\ttt{weight}(c)]
    \enspace | \enspace
    \econt[\ttt{resample}]
  \end{equation}
  \vspace{-8mm}
  \begin{equation}
    \begin{gathered}
      \frac{\term \rightarrow_\textsc{Det}^{+}\stopt}
      {\term,w,s \rightarrow \stopt,w,s}
      (\textsc{Det}) \
      \frac{c \ge 0}
      {\econt[\ttt{weight}(c)],w,s \rightarrow \econt[()],w\cdot c,s}
      (\textsc{Weight}) \\[2mm]
      \frac{c = F^{-1}_{D}(c_1, \ldots, c_{|D|},p)}
      {\econt[\ttt{sample}_D(c_1, \ldots, c_{|D|} )],w,p::s \rightarrow \econt[c],w,s}
      (\textsc{Sample})  \\[2mm]
      \frac{}
      {\econt[\ttt{resample}],w,s \rightarrow \econt[()],w,s}
      (\textsc{Resample}) \\[2mm]
    \end{gathered}
  \end{equation}
\end{definition}
The rule $(\textsc{Det})$ encapsulates the $\rightarrow_\textsc{Det}$ relation, and states that terms can move deterministically only to terms of the form $\stopt$.
Note that terms of the form $\stopt$ are found at the left-hand side in the other rules.
The $(\textsc{Sample})$ rule describes how random values are drawn from the inverse cumulative distribution functions and the trace when terms of the form $\ttt{sample}_D(c_1, \ldots, c_{|D|})$ are encountered.
Similarly, the $\textsc{Weight}$ rule determines how the weight is updated when $\ttt{weight}(c)$ terms are encountered.
Finally, the \ttt{resample} construct always evaluates to unit, and is therefore meaningless from the perspective of this semantics.
We elaborate on the role of the \ttt{resample} construct in Section~\ref{sec:ressemantics}.

With the semantics in place, we define two important functions over $\mathbb{S}$ for a given term.
In the below definition, assume that a fixed term $\term$ is given.
\begin{definition}
  \begin{equation}
    r_\term(s) =
    \begin{cases}
      \val & \text{if } \term,1,s \rightarrow^* \val,w,\trace{} \\
      () & \text{otherwise}
    \end{cases}
    \quad
    f_\term(s) =
    \begin{cases}
      w & \text{if } \term,1,s \rightarrow^* \val,w,\trace{} \\
      0 & \text{otherwise}
    \end{cases}
  \end{equation}
\end{definition}
Intuitively, $r_\term$ is the function returning the \emph{result value} after having repeatedly applied $\rightarrow$ on the initial trace $s$.
Analogously, $f_\term$ gives the \emph{density} or \emph{weight} of a particular $s$.
Note that, if $(\term,1,s)$ gets stuck or diverges, the result value is $()$, and the weight is 0.
In other words, we disregard such traces entirely, since we are in practice only interested in probability distributions over values.
Furthermore, note that if the final $s \neq \trace{}$, the value and weight are again $()$ and 0, respectively.
The motivation for this is discussed in Section~\ref{sec:targetmeasdef}.

\ifextended
To illustrate $r_\term$ and $f_\term$, first consider the geometric program $\termgeo$ in Fig.~\ref{fig:geometric:code}, and a trace $s = \trace{0.5,0.3,0.7}$.
Let
$
  \econt = \ttt{if} \ [\cdot] \ \ttt{then} \ 1 + \x{geometric} \ () \ \ttt{else} \ 1.
$
It is easy to check that $\termgeo \rightarrow_\textsc{Det}^+ \econt[\ttt{sample}_{\x{Bern}}(0.6)]$.
Now, note that, since $\x{Bern}(0.6)$ is the probability distribution for flipping a coin with bias $0.6$,
\begin{equation}
  F_\x{Bern}^{-1}(0.6,x) = 1 \text{ if } x \leq 0.6 \qquad
  F_\x{Bern}^{-1}(0.6,x) = 0 \text{ if } x > 0.6.
\end{equation}
As such, we have
\begin{equation}
  \begin{aligned}
    &
    \termgeo,1,\trace{0.5,0.3,0.7}
    \rightarrow \econt[\ttt{sample}_{\x{Bern}}(0.6)],1,\trace{0.5,0.3,0.7}
    \\&
    \rightarrow \econt[F^{-1}_\x{Bern}(0.6,0.5)],1,\trace{0.3,0.7}
    = \econt[1],1,\trace{0.3,0.7}
    \\&
    \rightarrow^+ 1 + \econt[F^{-1}_\x{Bern}(0.6,0.3)],1,\trace{0.7}
    \rightarrow^+ 1 + 1 + \econt[F^{-1}_\x{Bern}(0.6,0.7)],1,\trace{}
    \\&
    = 1 + 1 + \econt[0],1,\trace{}
    \rightarrow 3,1,\trace{}.
  \end{aligned}
\end{equation}
It follows that $r_\termgeo(s) = 3$, and that $f_\termgeo(s) = 1$.
Now, instead consider the trace $s_2 = \trace{0.5,0.7,0.3}$. We have
\begin{equation}
  \begin{aligned}
    &\termgeo,1,\trace{0.5,0.7,0.3}
    \rightarrow^+ \econt[1],1,\trace{0.7,0.3}
    \\
    & \quad \rightarrow^+ 1 + \econt[0],1,\trace{0.3}
    = 1 + 1,1,\trace{0.3}
    \rightarrow 2,1,\trace{0.3}.
  \end{aligned}
\end{equation}
The term is now stuck, and because we have not used up the entire trace, we have
$r_\termgeo(s_2) = ()$,
$f_\termgeo(s_2) = 0$.
The opposite of the above can also occur---given the trace $s_3 = \trace{0.5,0.3}$, it
holds that
$r_\termgeo(s_3) = ()$ and
$f_\termgeo(s_3) = 0$, since the provided trace is
not long enough.
In general, we have that $r_\termgeo(s) = n$ and $f_\termgeo(s) = 1$ whenever $s \in [0,0.6]^{n-1} \times (0.6,1]$.
Otherwise, $r_\termgeo(s) = ()$ and $f_\termgeo(s) = 0$.
We will apply this conclusion when reconsidering this example in Section~\ref{sec:targetmeasdef}.
\fi

\ifextended
To illustrate the \ttt{weight} construct, consider the program $\termobs$ in Fig.~\ref{fig:betaobs:code}, and the singleton trace $\trace{0.8}$.
\fi
This program will, in total, evaluate one call to \ttt{sample}, and three calls to \ttt{weight}.
Now, let $h(c) = F^{-1}_\x{Beta}(2,2,c)$ and recall the function $\sigma_{f_\x{Bern}}$ from Section~\ref{sec:syntax}.
Using the notation $\phi(c,x) = \sigma_{f_\x{Bern}}(h(c),x)$, we have, for some evaluation contexts $\econt_1,\econt_2,\econt_3,\econt_4$,
\begin{equation}
  \begin{aligned}
    &
    \termobs,1,\trace{0.8}
    = \econt_1[\ttt{sample}_\x{Beta}(2,2)],1,\trace{0.8}
    \rightarrow \econt_1[h(0.8)],1,\trace{}
    \\&
    \rightarrow \econt_2[\ttt{weight}(\phi(0.8,\true))],1,\trace{}
    \rightarrow \econt_2[()],\phi(0.8,\true),\trace{}
    \\&
    = \econt_2[()],h(0.8),\trace{}
    \rightarrow^+ \econt_3[()], \phi(0.8,\false) \cdot h(0.8),\trace{}
    \\&
    \rightarrow^+ \econt_4[()],
    \phi(0.8,\true) \cdot (1-h(0.8)) \cdot h(0.8) ,\trace{}
    \\&
    \rightarrow^+ h(0.8), h(0.8) \cdot (1-h(0.8)) \cdot h(0.8),\trace{}.
  \end{aligned}
\end{equation}
That is, $r_\termobs(\trace{0.8}) = h(0.8)$ and $f_\termobs(\trace{0.8}) = h(0.8)^2(1-h(0.8))$.
For arbitrary $c$, we see that $r_\termobs(\trace{c}) = h(c)$ and $f_\termobs(\trace{c}) = h(c)^2(1-h(c))$.
For any other trace $s$ with $|s| \neq 1$, $r_\termobs(s) = ()$ and $f_\termobs(s) = 0$.
We will apply this result when reconsidering this example in Section~\ref{sec:targetmeasdef}.

\subsection{Resampling Semantics}\label{sec:ressemantics}
In order to connect \gls{smc} in \glspl{ppl} to the classical formalization of \gls{smc} presented in Section~\ref{sec:formalsmc}---and thus enabling the theoretical treatments in Sections~\ref{sec:smcppl} and~\ref{sec:smckernel}---we need a relation in which terms ``stop'' after a certain number $n$ of encountered \ttt{resample} terms.
In this section, we define such a relation, denoted by $\hookrightarrow$.
Its definition is given below.
\begin{definition}
  \begin{equation}
    \begin{gathered}
      \frac{\term \neq \econt[\ttt{resample}] \quad
      \term,w,s \rightarrow \term',w',s'}
      {\term,w,s,n \hookrightarrow \term',w',s',n}
      \textsc{(Stoch-Fin)}
      \\[2mm]
      \frac{n > 0 \quad \econt[\ttt{resample}],w,s \rightarrow \econt[()],w,s}
      {\econt[\ttt{resample}],w,s,n \hookrightarrow \econt[()],w,s,n-1}
      \textsc{(Resample-Fin)}
    \end{gathered}
  \end{equation}
\end{definition}
This relation is $\rightarrow$ extended with a natural number $n$, indicating how many further \ttt{resample} terms can be evaluated.
We implement this limitation by replacing the rule $(\textsc{Resample})$ of $\rightarrow$ with $(\textsc{Resample-Fin})$ of $\hookrightarrow$ above which decrements $n$ each time it is applied, causing terms to get stuck at the $n+1$th resample encountered.

Now, assume that a fixed term $\term$ is given.
We define $r_{\term,n}$ and $f_{\term,n}$ similar to $r_\term$ and $f_\term$.
\begin{definition}\label{lemma:resulttermn}
  $
    r_{\term,n}(s) =
    \begin{cases}
      \val & \text{if }
      \term,1,s,n \hookrightarrow^*
      \val,w,\trace{},n' \\
      \econt[\ttt{resample}] & \text{if }
      \term,1,s,n \hookrightarrow^*
      \econt[\ttt{resample}],w,\trace{},0 \\
      () & \text{otherwise}
    \end{cases}
  $
\end{definition}
\begin{definition}\label{def:termdensityn}
  $
    f_{\term,n}(s) =
    \begin{cases}
      w & \text{if }
      \term,1,s,n \hookrightarrow^*
      \val,w,\trace{},n' \\
      w & \text{if }
      \term,1,s,n \hookrightarrow^*
      \econt[\ttt{resample}],w,\trace{},0 \\
      0 & \text{otherwise}
    \end{cases}
  $
\end{definition}
As for $r_\term$ and $f_\term$, these functions return the result value and weight, respectively, after having repeatedly applied $\hookrightarrow$ on the initial trace $s$.
There is one difference compared to $\rightarrow$: besides values, we now also allow stopping with non-zero weight at terms of the form $\econt[\ttt{resample}]$.

To illustrate $\hookrightarrow$, $r_{\term,n}(s)$, and $f_{\term,n}(s)$, consider the term $\termseq$ defined by
\begin{equation}\label{eq:seq}
  \begin{tabular}{c}
    \begin{lstlisting}
let $\x{observe}$ $x$ $o$ = weight$(f_\mathcal{N}(x,4,o))$; resample in
let $\x{sim}$ $x_{n-1}$ $o_n$ =
  let $x_{n}$ = sample$_\mathcal{N}(x_{n-1} + 2, 1)$ in $\x{observe}$ $x_{n}$ $o_n$; $x_{n}$ in
let $x_0$ = sample$_{\mathcal{U}}(0,100)$ in
let $f$ = $\x{foldl}$ $\x{sim}$ in $f$ $x_0$ [$c_1,c_2,\ldots,c_{t-1},c_t$]$\text{.}$
    \end{lstlisting}
  \end{tabular}
\end{equation}
This term encodes a model in which an object moves along a real-valued axis in discrete time steps, but where the actual positions ($x_1$, $x_2$, \ldots) can only be observed through a noisy sensor ($c_1$, $c_2$, \ldots). The inference problem consists of finding the probability distribution for the very last position, $x_t$, given all collected observations ($c_1$, $c_2$, \ldots, $c_t$).
Most importantly, note the position of \ttt{resample} in \eqref{eq:seq}---it is evaluated just after evaluating \ttt{weight} in every folding step.
Because of this, for $n < t$ and all traces $s$ such that $f_{\termseq,n}(s) > 0$, we have $r_{\termseq,n}(s) = \ttt{$\econt_{\x{seq}}^n[$resample; $x_{n}]$}$, where
 \lstinline!$\econt_{\x{seq}}^n$ = $f$ $[\cdot]$ [$c_{n+1},c_{n+2},\ldots,c_{t-1},c_{t}$]!
and where $x_n$ is the value sampled in $\x{sim}$ at the $n$th folding step.
That is, we can now ``stop'' evaluation at \ttt{resample}s.
We will revisit this example in Section~\ref{sec:smcppl}.

\section{The Target Measure of a Program}\label{sec:targetmeas}
In this section, we define the \emph{target measure} induced by any given program in our calculus.
We assume basic familiarity with measure theory, Lebesgue integration, and Borel spaces.
McDonald and Weiss~\cite{mcdonald2012course} provide a pedagogical introduction to the subject.
\ifextended
We also summarize the definitions and lemmas used in this article in Appendix~\ref{sec:prelim1}.
\fi
In order to define the target measure of a program as a Lebesgue integral (Section~\ref{sec:targetmeasdef}), we require a \emph{measure space} on traces (Section~\ref{sec:tracespace}), and a \emph{measurable space} on terms (Section~\ref{sec:termspace}).
\ifextended
For illustration, we derive the target measures for two of the example programs from Section~\ref{sec:calculus} in Section~\ref{sec:targetmeasdef}.
\fi
The concepts presented in this section are quite standard, and experienced readers might want to quickly skim it, or even skip it entirely.

\subsection{A Measure Space over Traces}
\label{sec:tracespace}
We use a standard measure space over traces of samples~\cite{Mak2020DensitiesOA}.
First, we define a measurable space over traces.
We denote the Borel $\sigma$-algebra on $\mathbb{R}^n$ with $\Bor^n$, and the Borel $\sigma$-algebra on $[0,1]$ with $\Bor_{[0,1]}^n$.
\begin{definition}
The $\sigma$-algebra $\Trcal$ on $\Tr$ is the $\sigma$-algebra
consisting of sets of the form
$
  S = \bigcup_{n\in\mathbb{N}_0} B_n
$
with $B_n \in \Bor_{[0,1]}^n$.
Naturally, $[0,1]^0$ is the singleton set containing the empty trace.
In other words,
$
  ([0,1]^0, \Bor_{[0,1]}^0)
  =
  (\{\trace{}\},\{\{\trace{}\},\emptyset\} ),
$
where $\trace{}$ denotes the empty trace.
\end{definition}

\begin{lemma}\label{lemma:tracemeas}
  \lemmatracemeas$^\dagger$
\end{lemma}
The most common measure on $\Bor^n$ is the $n$-dimensional Lebesgue measure, denoted $\lambda_n$.
For $n = 0$, we let $\lambda_0 = \delta_{\trace{}}$, where $\delta$ denotes the standard Dirac measure.
By combining the Lebesgue measures for each $n$, we construct a measure $\mu_\Tr$ over $(\Tr,\Trcal)$.
\begin{definition}
  $
  \mu_\Tr(S) = \mu_\Tr\left(\bigcup_{n\in\mathbb{N}_0} B_n\right) =
  \sum_{n\in\mathbb{N}_0} \lambda_n(B_n)
  $
\end{definition}
\begin{lemma}\label{lemma:tracemeasspace}
  \lemmatracemeasspace$^\dagger$
\end{lemma}

A comment on notation: we denote
universal sets by blackboard bold capital letters (e.g., $\Tr$),
$\sigma$-algebras by calligraphic capital letters (e.g., $\Trcal$),
members of $\sigma$-algebras by capital letters (e.g., $S$), and
individual elements by lower case letters (e.g., $s$).

\subsection{A Measurable Space over Terms}\label{sec:termspace}
In order to show that $r_\term$ is measurable, we need a measurable space over terms.
We let $(\T,\Tcal)$ denote the measurable space that we seek to construct, and follow the approach in Staton et al.~\cite{staton2016semantics} and V\'{a}k\'{a}r et al.~\cite{vakar2019domain}.
Because our calculus includes the reals, we would like to at least have $\Bor \subset \Tcal$.
Furthermore, we would also like to extend the Borel measurable sets $\mathcal{B}^n$ to terms with $n$ reals as subterms.
For instance, we want sets of the form
$\{%
  (\lambda x.\ (\lambda y.\ x + y) \enspace c_2)
  \enspace c_1 \mid (c_1,c_2) \in B_2
\}$
to be measurable, where $B_2 \in \Bor^2$.
This leads us to consider terms in a language in which constants (i.e., reals) are replaced with placeholders $[\cdot]$.
\begin{definition}
  Let
  $
    \val_p \Coloneqq
    [\cdot] \enspace | \enspace
    \lambda x. \term
  $
  replace the values $\val$ from Definition~\ref{def:syntax}.
  The set of all terms in the resulting new calculus is denoted with $\T_p$.
\end{definition}
Most importantly, it is easy to verify that $\T_p$ is countable.
Next, we make the following definitions.
\begin{definition}
  For $n \in \mathbb{N}_0$, we denote by $\T_p^n \subset \T_p$ the set of all terms with exactly $n$ placeholders.
\end{definition}
\begin{definition}
  We let $\term_p^n$ range over the elements of $\T_p^n$.
  The $\term_p^n$ can be regarded as functions $\term_p^n : \R^n \to \term_p^n(\R_n)$ which replaces the $n$ placeholders with the $n$ reals given as arguments.
\end{definition}
\begin{definition}
  $\Tcal_{\term_p^n} = \{\term_p^n(B_n) \mid B_n \in \Bor^n\}$.
\end{definition}
From the above definitions, we construct the required $\sigma$-algebra $\Tcal$.
\begin{definition}
  The $\sigma$-algebra $\Tcal$ on $\T$ is the $\sigma$-algebra consisting of sets of the form
  $
    T =
    \bigcup_{n \in N_0}
    \bigcup_{\term^n_p \in \T^n_p}
    \term_p^n(B_n).
  $
\end{definition}
\begin{lemma}\label{lemma:termmeas}
  \lemmatermmeas$^\dagger$
\end{lemma}

\subsection{The Target Measure}\label{sec:targetmeasdef}
We are now in a position to define the target measure.
We will first give the formal definitions, and then illustrate the definitions with \ifextended examples\fi.
The definitions rely on the following result.
\begin{lemma}\label{lemma:rfmeas}
  \lemmarfmeas$^\dagger$
\end{lemma}
We can now proceed to define the measure $\llangle \term \rrangle$ over $\Tr$ induced by a term $\term$ using Lebesgue integration.
\begin{definition}\label{def:tracemeasure}
  $
  \llangle \term \rrangle(S) = \int_S f_\term(s) \ d\mu_\Tr(s)
  $
\end{definition}
\ifextended
Importantly, by Lemma~\ref{lemma:density} and Lemma~\ref{lemma:tracemeasspace}, it holds that the density $f_\term$ is unique $\mu_\Tr$-ae if $\llangle\term\rrangle$ is $\sigma$-finite.
\fi

Using Definition~\ref{def:tracemeasure} and the measurability of $r_\term$, we can also define a corresponding pushforward measure $\llbracket \term \rrbracket$ over $\T$.
\begin{definition}
  $
  \llbracket \term \rrbracket(T) =
  \llangle \term \rrangle(r_\term^{-1}(T)) =
  \int_{r_\term^{-1}(T)} f_\term(s) \ d\mu_\Tr(s).
  $
\end{definition}
The measure $\llbracket \term \rrbracket$ is our \emph{target measure}, i.e., the measure encoded by our program that we are interested in.

\ifextended
Let us now consider the target measures for our earlier examples.
Consider first the program in Fig.~\ref{fig:geometric:code}.
Recall that the density $f_\termgeo$ of a given trace $s$ is 1 if $s \in [0,0.6]^{n-1} \times (0.6,1]$, and 0 otherwise.
Hence, we can write
\begin{equation}
  \begin{aligned}
    \llangle \termgeo \rrangle(S)
    &=
    \int_S f_\termgeo(s) \ d\mu_\Tr(s)
    =
    \mu_\Tr\left(\bigcup_{n \in \mathbb{N}}S
    \cap \left([0,0.6]^{n-1} \times (0.6,1]\right)\right)
    \\
    &=
    \sum_{n \in \mathbb{N}}
    \lambda_n(S \cap \left([0,0.6]^{n-1} \times
    (0.6,1]\right))).
  \end{aligned}
\end{equation}
Since $\termgeo$ is a distribution over $\mathbb{N}$, we always have
\begin{equation}
  r_\termgeo^{-1}(T) = r_\termgeo^{-1}(T \cap \mathbb{N})
  = \bigcup_{n\in T \cap \mathbb{N}} [0,0.6]^{n-1} \times (0.6,1].
\end{equation}
Consequently,
\begin{equation}
  \begin{aligned}
    \llbracket \termgeo \rrbracket(T)
    &=
    \llangle \termgeo \rrangle(r_\termgeo^{-1}(T))
    =
    \sum_{n \in \mathbb{N}}
    \lambda_n\left(r_\termgeo^{-1}(T) \cap \left([0,0.6]^{n-1} \times
    (0.6,1]\right))\right)
    \\ &=
    \sum_{n \in T \cap \mathbb{N}}
    \lambda_n\left([0,0.6]^{n-1} \times (0.6,1]\right)
    =
    \sum_{n \in T \cap \mathbb{N}} 0.6^{n-1}\cdot 0.4.
  \end{aligned}
\end{equation}
As expected, by taking $\llbracket \termgeo \rrbracket(\{1\})$, $\llbracket \termgeo \rrbracket(\{2\})$, $\llbracket \termgeo \rrbracket(\{3\})$, $\ldots$, we exactly recover the graph from Fig.~\ref{fig:geometric:graph}.

Now consider the continuous distribution given by program $\termobs$, and recall that $h = F^{-1}_\x{Beta}(2,2,\cdot)$.
Furthermore, recall that $f_\termobs(\trace{c}) = {h(c)}^2(1-h(c))$, $r_\termobs(\trace{c}) = h(c)$, $f_\termobs(s) = 0$,  and $r_\termobs(s) = ()$ if $|s| \neq 1$.
Because only traces of length 1 have non-zero $f$, we have
\begin{equation}
  \begin{aligned}
    \llangle \termobs \rrangle(S) &= \int_S f_\termobs(s) \ d\mu_\Tr(s)
    =
    \int_{S \cap [0,1]} f_\termobs(c) \ d\lambda(c)
    \\ &=
    \int_{S \cap {[0,1]}} {h(c)}^2(1-h(c)) \ d\lambda(c)
  \end{aligned}
\end{equation}
The $\x{Beta}$ distributions have strictly increasing cumulative distribution functions $F_\x{Beta}(a,b,\cdot)$ for all $a$ and $b$.
It follows that $h$ is the true inverse of this function, and is therefore bijective\footnote{%
This property does \emph{not} hold for probability distributions in general. In particular, if $D$ is discrete, $F_D$ has no inverse, and $F_D^{-1}$ is defined differently (making the standard $^{-1}$ notation confusing).%
}.
Because of this,
\begin{equation}
  \begin{aligned}
    \llbracket \termobs \rrbracket(T) &=
    \int_{r_\termobs^{-1}(T)} f_\termobs(s) \ d\mu_\Tr(s)
    =
    \int_{r_\termobs^{-1}(T \cap \R)} f_\termobs(s) \ d\mu_\Tr(s)
    \\ &=
    \int_{h^{-1}(T \cap \R)} {h(c)}^2(1-h(c)) \ d\lambda(c)
    \\ &=
    \int_{T \cap \R} {h(h^{-1}(c))}^2(1-h(h^{-1}(c))) (h^{-1})'(c) \ d\lambda(c)
    \\ &=
    \int_{T \cap \R} c^2(1-c) \sigma_{f_\x{Beta}}(2,2,c) \ d\lambda(c)
    \\ &\propto
    \int_{T \cap \R} c^2(1-c) c (1-c) \ d\lambda(c)
    \\ &=
    \int_{T \cap \R} c^3{(1-c)}^2 \ d\lambda(c).
  \end{aligned}
\end{equation}
In the third equality, we have used integration by substitution.
We also used the fact that $(h^{-1})'$ is the derivative of the cumulative distribution function $F_\x{Beta}(2,2,\cdot)$.
That is, $(h^{-1})'$ is the probability density function $\sigma_{f_\x{Beta}}(2,2,x) \propto x(1-x)$.
\fi

We should in some way ensure the target measure is finite (i.e., can be normalized to a probability measure), since we are in the end most often only interested in probability measures.
Unfortunately, as observed by Staton~\cite{staton2017commutative}, there is no known useful syntactic restriction that enforces finite measures in \glspl{ppl} while still admitting weights $>1$.
We will discuss this further in Section~\ref{sec:smcpplcorrect} in relation to \gls{smc} in our calculus.

Lastly, from Section~\ref{sec:semantics}, recall that we disallow non-empty final traces in $f_\term$ and $r_\term$.
We see here why this is needed: if allowed, for every trace $s$ with $f_\term(s) > 0$, \emph{all} extensions $s \ast s'$ have the same density $f_\term(s \ast s') = f_\term(s) > 0$.
From this, it is easy to check that if $\llbracket\term\rrbracket \neq 0$ (the zero measure), then $\llbracket\term\rrbracket(\T) = \infty$ (i.e., the measure is not finite).
In fact, for any $T \in \Tcal$,
$\llbracket\term\rrbracket(T) > 0 \implies \llbracket\term\rrbracket(T) = \infty$.
Clearly, this is not a useful target measure.

\section{Formal SMC}\label{sec:formalsmc}
In this section, we give a generic formalization of \gls{smc} based on Chopin~\cite{chopin2004central}.
We assume a basic understanding of \gls{smc}.
\ifextended
For a concrete \gls{smc} example, see Appendix~\ref{sec:smcintro}.
\fi
For a complete introduction to \Gls{smc}, we recommend Naesseth et al.~\cite{naesseth2019elements} and Doucet and Johansen~\cite{doucet2009tutorial}.

First, in Section~\ref{sec:prelim2}, we introduce transition kernels, which is a fundamental concept used in the remaining sections of the paper.
Second, in Section~\ref{sec:smcalg}, we describe Chopin's generic formalization of \gls{smc} as an algorithm for approximating a sequence of distributions based on a sequence of approximating transition kernels.
Lastly, in Section~\ref{sec:smccorrect}, we give standard correctness results for the algorithm.

\subsection{Preliminaries: Transition Kernels}\label{sec:prelim2}
Intuitively, transition kernels describe how elements move between measurable spaces.
For a more comprehensive introduction, see~V\'{a}k\'{a}r and Ong~\cite{vakar2018sfinite}.

\begin{definition}
  Let $(\mathbb{A},\mathcal{A})$ and $(\mathbb{A}',\mathcal{A}')$ be measurable spaces, and let $\Bor_+^* = \{ B \mid B \setminus \{\infty\} \in \Bor_+\}$.
  A function $k : \mathbb{A} \times \mathcal{A}' \to \R_+^*$ is a \emph{(transition) kernel} if
  (1) for all $a \in \mathbb{A}$, $k(a,\cdot) : \mathcal{A}' \to \R_+^*$ is a measure on $\mathcal{A}'$, and
  (2) for all $A' \in \mathcal{A}'$, $k(\cdot,A') : (\mathbb{A},\mathcal{A}) \to (\R_+^*,\Bor_+^*$) is measurable.
\end{definition}
Additionally, we can classify transition kernels according to the below definition.
\begin{definition}
  Let $(\mathbb{A},\mathcal{A})$ and $(\mathbb{A}',\mathcal{A}')$ be measurable spaces.
  A kernel $k : \mathbb{A} \times \mathcal{A}' \to \R_+^*$ is
  a \emph{sub-probability kernel} if $k(a,\cdot)$ is a sub-probability measure for all $a \in \mathbb{A}$;
  a \emph{probability kernel} if $k(a,\cdot)$ is a probability measure for all $a \in \mathbb{A}$; and
  a \emph{finite kernel} if $\sup_{a \in \mathbb{A}} k(a,\mathbb{A}') < \infty$.
\end{definition}

\subsection{Algorithm}\label{sec:smcalg}
The starting point in Chopin's formulation of \gls{smc} is a sequence of probability measures $\pi_n$ (over respective measurable spaces $(\mathbb{A}_n,\mathcal{A}_n)$, with $n\in\mathbb{N}_0$) that are difficult or impossible to directly draw samples from.

The \gls{smc} approach is to generate samples from the $\pi_n$ by first sampling from a sequence of \emph{proposal measures} $q_n$, and then correcting for the discrepancy between these measures by weighting the proposal samples.
The proposal distributions are generated from an initial measure $q_0$ and a sequence of transition kernels $k_n : \mathbb{A}_{n-1} \times \mathcal{A}_n \to [0,1], n \in \mathbb{N}$ as
\begin{equation}\label{eq:qn}
  q_n(A_n) =
  \int_{\mathbb{A}_{n-1}}\!\!\! k_n(a_{n-1}, A_n)\,
  d\pi_{n-1}(a_{n-1}).
\end{equation}
%
%
In order to approximate $\pi_{n}$ by weighting samples from $q_{n}$, we need some way of obtaining the appropriate weights.
Hence, we require each measurable space $(\mathbb{A}_n,\mathcal{A}_n)$ to have a default $\sigma$-finite measure $\mu_{\mathbb{A}_n}$, and the measures $\pi_n$ and $q_n$ to have densities $f_{\pi_n}$ and $f_{q_n}$ with respect to this default measure.
Furthermore, we require that the functions
$f_{\pi_n}$ and $f_{q_n}$ can be efficiently computed pointwise, up to an unknown constant factor per function and value of $n$. More precisely, we can efficiently compute the densities
$f_{\widetilde\pi_n} = Z_{\widetilde\pi_n}\cdot f_{\pi_n}$ and
$f_{\widetilde q_n} = Z_{\widetilde q_n}\cdot f_{q_n}$, corresponding to the unnormalized measures
$\widetilde\pi_n = Z_{\widetilde\pi_n} \cdot \pi_n$ and
$\widetilde q_n = Z_{\widetilde q_n} \cdot q_n$.
Here, $Z_{\widetilde \pi_n} = \widetilde\pi_n(\mathbb{A}_n) \in \R_+$ and $Z_{\widetilde q_n} = \widetilde q_n(\mathbb{A}_n) \in \R_+$ denote the unknown \emph{normalizing constants} for the distributions $\widetilde \pi_n$ and $\widetilde q_n$.

\begin{algorithm}[tb]
  \caption{%
    A generic formulation of sequential Monte Carlo inference based on Chopin~\cite{chopin2004central}.
    In each step, we let $1 \leq j \leq J$, where $J$ is the number of samples.
  }%
  \label{alg:smc}
  \begin{enumerate}

    \item \textbf{Initialization:}
      Set $n = 0$.
      Draw $a_{0}^j \sim q_0$ for $1\leq j \leq J$. \\
      The empirical distribution given by $\{a_0^j\}_{j=1}^J$ approximates $q_0$.

    \item \textbf{Correction:}
      \label{alg:smc:correction}
      Calculate
      $
      w_n^j = \frac{f_{\widetilde\pi_n}(a_n^j)}{f_{\widetilde q_n}(a_n^j)}
      $. \\
      The empirical distribution given by $\{(a_n^j,w_n^j)\}_{j=1}^J$ approximates $\pi_n$.

    \item \textbf{Selection:}
      Resample the empirical distribution $\{(a_n^j,w_n^j)\}_{j=1}^J$. \\
      The new empirical distribution is unweighted and is given by $\{\hat{a}_n^j\}_{j=1}^J$.
      This distribution also approximates $\pi_n$.

    \item \textbf{Mutation:}
      Increment $n$. \\
      Draw $a_{n}^j \sim k_n(\hat{a}_{n-1}^j,\cdot)$ for $1\leq j \leq J$.
      The empirical distribution given by $\{a_n^j\}_{j=1}^J$ approximates $q_n$.
      Go to \eqref{alg:smc:correction}.

  \end{enumerate}
\end{algorithm}

Algorithm~\ref{alg:smc} presents a generic version of \gls{smc} \cite{chopin2004central} for approximating ~$\pi_n$.
We make the notion of approximation used in the algorithm precise in Section~\ref{sec:smccorrect}.
Note that in the correction step, the unnormalized pointwise evaluation of $f_{\pi_n}$ and $f_{q_n}$ is used to calculate the weights.
In the algorithm description, we also use some new terminology.
First, an \emph{empirical distribution} is the discrete probability measure formed by a finite set of possibly weighted samples $\{(a_n^j,w_n^j)\}_{j=1}^J$, where $a_n^j\in\mathbb{A}_n$ and $w_n^j\in\R_+$.
Second, when \emph{resampling} an empirical distribution, we sample $J$ times from it (with replacement), with each sample having its normalized weight as probability of being sampled.
More specifically, this is known as \emph{multinomial resampling}.
Other resampling schemes also exist \cite{douc2005comparison}, and are often used in practice to reduce variance.
After resampling, the set of samples forms a new empirical distribution with $J$ unweighted (all $w_n^j = 1$) samples.

An important feature of \gls{smc} compared to other inference algorithms is that \gls{smc} produces, as a by-product of inference, unbiased estimates $\hat Z_{\widetilde \pi_n}$ of the normalizing constants $Z_{\widetilde \pi_n}$.
Stated differently, this means that Algorithm~\ref{alg:smc} not only approximates the $\pi_n$, but also the unnormalized versions $\widetilde \pi_n$.
From the weights $w_n^j$ in Algorithm~\ref{alg:smc}, the estimates are given by
\begin{equation}\label{eq:normconst}
  \hat Z_{\widetilde \pi_n} = \prod_{i=0}^n \frac{1}{J} \sum_{j=1}^J w_i^j \approx Z_{\widetilde \pi_n}
\end{equation}
for each $\widetilde \pi_n$.
We give the unbiasedness result of $\hat Z_{\widetilde \pi_n}$ in Lemma~\ref{lemma:smccorrect} (item \ref{item:smcunbiased}) below. The normalizing constant is often used to compare the accuracy of different probabilistic models, and as such, it is also known as the \emph{marginal likelihood}, or \emph{model evidence}.
For an example application, see Ronquist et al.~\cite{ronquist2020probabilistic}.

To conclude this section, note that many sequences of probability kernels $k_n$ can be used to approximate the same sequence of measures $\pi_n$.
The only requirement on the $k_n$ is that $f_{\pi_n}(a_n) > 0 \implies f_{q_n}(a_n) > 0$ must hold for all $n \in \mathbb{N}_0$ and $a_n \in \mathbb{A}_n$ (i.e., the proposals must ``cover'' the $\pi_n$) \cite{doucet2001sequential}. We call such a sequence of kernels $k_n$ \emph{valid}.
Different choices of $k_n$ induce different proposals $q_n$, and hence capture different \gls{smc} algorithms. The most common example is the \gls{bpf}, which directly uses the kernels from the model as the sequence of kernels in the \gls{smc} algorithm (hence the ``bootstrap'').
In Section~\ref{sec:bpf}, we formalize the bootstrap kernels in the context of our calculus.
However, we may want to choose other probability kernels that satisfy the covering condition, since the choice of kernels can have major implications for the rate of convergence \cite{pitt1999filtering}.

\subsection{Correctness}\label{sec:smccorrect}
We begin by defining the notion of approximation used in Algorithm~\ref{alg:smc}.

\begin{definition}[Based on {Chopin~\cite[p.~2387]{chopin2004central}}]\label{def:approx}
  \newcommand{\triarr}{\{\{(a^{j,J},w^{j,J})\}_{j=1}^J\}_{J \in \mathbb{N}}}
  Let $(\mathbb{A},\mathcal{A})$ denote a measurable space,
  $\triarr$ a triangular array of random variables in $\mathbb{A} \times \R$, and $\pi : \mathcal{A} \to \R_+^*$ a probability measure.
  We say that $\triarr$ \emph{approximates} $\pi$ if the equality
  $\displaystyle
    \lim_{J\to\infty}
    \frac{\sum_{j=1}^J w^{j,J}\varphi(a^{j,J})}
    {\sum_{j=1}^J w^{j,J}}
    =
    \mathbb{E}_{\pi}(\varphi)
  $
  holds almost surely for all measurable functions $\varphi : (\mathbb{A},\mathcal{A}) \to (\R,\Bor)$ such that $\mathbb{E}_{\pi}(\varphi)$---the expected value of the function $\varphi$ over the distribution $\pi$---exists.
\end{definition}
First, note that the triangular array can also be viewed as a sequence of random empirical distributions (indexed by $J$).
Precisely such sequences are formed by the random empirical distributions in Algorithm~\ref{alg:smc} when indexed by the increasing number of samples $J$.
For simplicity, we often let context determine the sequence, and directly state that a random empirical distribution approximates some distribution (as in Algorithm~\ref{alg:smc}).

Two classical results in \gls{smc} literature are given in the following lemma: a law of large numbers and the unbiasedness of the normalizing constant estimate.
We take these results as the definition of \gls{smc} correctness used in this paper.
\begin{lemma}\label{lemma:smccorrect}
  Let $\pi_n$, $n \in \mathbb{N}_0$, be a sequence of probability measures over measurable spaces $(\mathbb{A}_n, \mathcal{A}_n)$ with default $\sigma$-finite measures $\mu_{\mathbb{A}_n}$, such that the $\pi_n$ have densities $f_{\pi_n}$ with respect to these default measures.
  Furthermore, let $q_0$ be a probability measure with density $f_{q_0}$ with respect to $\mu_{\mathbb{A}_0}$, and $k_n$ a sequence of probability kernels inducing a sequence of proposal probability measures $q_n$, given by \eqref{eq:qn},
  over $(\mathbb{A}_n, \mathcal{A}_n)$ with densities $f_{q_n}$ with respect to $\mu_{\mathbb{A}_n}$. Also, assume the
  $k_n$ are valid, i.e., that that $f_{\pi_n}(a_n) > 0 \implies f_{q_n}(a_n) > 0$ holds for all $n \in \mathbb{N}_0$ and $a_n \in \mathbb{A}_n$.
  Then
  \begin{enumerate}
    \item \label{item:smcconvergence}
      the empirical distributions ${\{(a_n^j,w_n^j)\}}_{j=1}^J$ and ${\{\hat{a}_n^j\}}_{j=1}^J$ produced by Algorithm~\ref{alg:smc} approximate $\pi_n$ for each $n \in \mathbb{N}_0$; and
    \item \label{item:smcunbiased}
      $\mathbb{E}(\hat Z_{\widetilde \pi_n}) = Z_{\widetilde \pi_n}$ for each $n \in \mathbb{N}_0$, where the expectation is taken with respect to the weights produced when running Algorithm~\ref{alg:smc}, and $\hat Z_{\widetilde \pi_n}$ is given by \eqref{eq:normconst}.
  \end{enumerate}
\end{lemma}
\begin{proof}
  As referenced in~Naesseth et al.~\cite{naesseth2019elements}, see~Del Moral~\cite{moral2004feynman}[Theorem 7.4.3] for 1. For 2, see~Naesseth et al.~\cite{naesseth2019elements}[Appendix 4.A].
\end{proof}

Chopin~\cite{chopin2004central}[Theorem 1] gives another \gls{smc} convergence result in the form of a central limit.
This result, however, requires further restrictions on the weights $w_n^j$ in Algorithm~\ref{alg:smc}.
It is not clear when these restrictions are fulfilled when applying \gls{smc} on a program in our calculus.
This is an interesting topic for future work.

\section{Formal SMC for Probabilistic Programming Languages}%
\label{sec:smcppl}
This section contains our main contribution: how to interpret the operational semantics of our calculus as the unnormalized sequence of measures $\widetilde \pi_n$ in Chopin's formalization (Section~\ref{sec:seqmeas}), as well as sufficient conditions for this sequence of approximating measures to converge to $\llangle \term \rrangle$ and for the normalizing constant estimate to be correct (Section~\ref{sec:smcpplcorrect}).

An important insight during this work was that it is more convenient to find an approximating sequence of measures $\llangle\term\rrangle_n$ to the trace measure $\llangle\term\rrangle$, compared to finding a sequence of measures $\llbracket\term\rrbracket_n$ directly approximating the target measure $\llbracket\term\rrbracket$.
In Section~\ref{sec:seqmeas}, we define $\llangle\term\rrangle_n$ similarly to $\llangle\term\rrangle$, except that at most $n$ evaluations of \ttt{resample} are allowed.
This upper bound on the number of \ttt{resample}s is formalized through the relation $\hookrightarrow$ from Section~\ref{sec:ressemantics}.

In Section~\ref{sec:smcpplcorrect}, we obtain two different conditions for the convergence of the sequence $\llangle \term \rrangle_n$ to $\llangle \term \rrangle$: Theorem~\ref{theorem:main} states that for programs with an upper bound $N$ on the number of \ttt{resample}s they evaluate, $\llangle \term \rrangle_N=\llangle \term \rrangle$.
This precondition holds in many practical settings, for instance where each resampling is connected to a datum collected before inference starts.
Theorem~\ref{theorem:dct} states another convergence result for programs without such an upper bound but with dominated weights.
Because of these convergence results, we can often approximate $\llangle \term \rrangle$ by approximating $\llangle \term \rrangle_n$ with Algorithm~\ref{alg:smc}.
When this is the case, Lemma~\ref{lemma:smccorrect} implies that Algorithm~\ref{alg:smc}, either after a sufficient number of time steps or asymptotically, correctly approximates $\llangle\term\rrangle$ and the normalizing constant $Z_{\llangle\term\rrangle}$.
This is the content of Theorem~\ref{theorem:consistency}.
We conclude Section~\ref{sec:smcpplcorrect} by discussing \ttt{resample} placements and their relation to Theorem~\ref{theorem:consistency}, as well as practical implications of Theorem~\ref{theorem:consistency}.

\subsection{The Sequence of Measures Generated by a Program}\label{sec:seqmeas}
We now apply the formalization from Section~\ref{sec:targetmeasdef} again, but with $f_{\term,n}$ and $r_{\term,n}$ (from Section~\ref{sec:ressemantics}) replacing $f_\term$ and $r_\term$.
Intuitively, this yields a sequence of measures $\llbracket\term\rrbracket_n$ indexed by $n$, which are similar to $\llbracket\term\rrbracket$, but only allow for evaluating at most $n$ resamples.
To illustrate this idea, consider again the program $\termseq$ in \eqref{eq:seq}.
Here, $\llbracket\termseq\rrbracket_0$ is a distribution over terms of the form
\ttt{$\econt_{\x{seq}}^1[$resample; $x_{1}]$}, $\llbracket\termseq\rrbracket_1$ a distribution over terms of the form \ttt{$\econt_{\x{seq}}^2[$resample; $x_{2}]$}, and so forth.
For $n \geq t$, $\llbracket\termseq\rrbracket_n = \llbracket\termseq\rrbracket$, because it is clear that $t$ is an upper bound on the number of resamples evaluated in $\termseq$.

While the measures $\llbracket \term \rrbracket_n$ are useful for giving intuition, it is easier from a technical perspective to define and work with $\llangle \term \rrangle_n$, the sequence of measures over \emph{traces} where at most $n$ \texttt{resample}s are allowed.
First, we need the following result, analogous to Lemma~\ref{lemma:rfmeas}.
\begin{lemma}\label{lemma:rfnmeas}
  \lemmarfnmeas$^\dagger$
\end{lemma}
This allows us to define $\llangle \term \rrangle_n$ (cf. Definition~\ref{def:tracemeasure}).
\begin{definition}\label{def:seqtarget}
  $
  \llangle \term \rrangle_n(S) = \int_S f_{\term,n}(s) \ d\mu_\Tr(s)
  $
\end{definition}
\ifextended
Analogously to Definition~\ref{def:tracemeasure}, by Lemma~\ref{lemma:density} and Lemma~\ref{lemma:tracemeasspace}, it holds that the density $f_{\term,n}$ is unique $\mu_\Tr$-ae if $\llangle\term\rrangle_n$ is $\sigma$-finite.
\fi
\ifextended
We can now also clarify how the \ttt{resample} construct relates to the resampling in the selection step of Algorithm~\ref{alg:smc}.
If we approximate the sequence $\llangle\term\rrangle_n$ with Algorithm~\ref{alg:smc}, at the $n$th selection step of the algorithm, all traces $s$ with non-zero weight must have $r_{\term,n}(s) = \val$ or $r_{\term,n}(s) = \econt[\ttt{resample}]$, by Definitions~\ref{lemma:resulttermn} and~\ref{def:termdensityn}.
That is, having a $q_n$ in Algorithm~\ref{alg:smc} proposing traces other than these is wasteful, since they will in any case have weight zero.
We illustrate this further when considering the bootstrap kernel in Section~\ref{sec:bpf}.
\fi

\subsection{Correctness}\label{sec:smcpplcorrect}
We begin with a convergence result for when the number of calls to $\ttt{resample}$ in a program is upper bounded.
\begin{theorem}\label{theorem:main}
  \theoremmain
\end{theorem}
This follows directly since $f_{\term,n}$ not only converges to $f_\term$, but is also equal to $f_\term$ for all $n > N$.
However, even if the number of calls to \ttt{resample} in $\term$ is upper bounded, there is still one concern with using $\llangle\term\rrangle_n$ as $\widetilde \pi_n$ in Algorithm~\ref{alg:smc}: there is no guarantee that the measures $\llangle\term\rrangle_n$ can be normalized to probability measures and have unique densities (i.e., that they are finite).
This is a requirement for the correctness results in Lemma~\ref{lemma:smccorrect}.
Unfortunately, recall from Section~\ref{sec:targetmeasdef} that there is no known useful syntactic restriction that enforces finiteness of the target measure.
This is clearly true for the measures $\llangle\term\rrangle_n$ as well, and as such, we need to \emph{make the assumption} that the $\llangle \term \rrangle_n$ are finite---otherwise, it is not clear that Algorithm~\ref{alg:smc} produces the correct result, since the conditions in Lemma~\ref{lemma:smccorrect} are not fulfilled.
Fortunately, this assumption is valid for most, if not all, models of practical interest.
Nevertheless, investigating whether or not the restriction to probability measures in Lemma~\ref{lemma:smccorrect} can be lifted to some extent is an interesting topic for future work.
\ifextended
Note that, even if the target measure is finite, this does not necessarily imply that all measures $\llangle\term\rrangle_n$ are finite.
For example, consider the program
\begin{center}
  \begin{tabular}{c}
    \begin{lstlisting}
let rec $\x{inflate}$ $\_$ =
  if sample$_{\x{Bern}}(0.5)$ then weight $2$; 1 + $\x{inflate}$ $()$
  else $0$ in
let $\x{deflate}$ $n$ = weight $1/2^n$ in
let $n$ = $\x{inflate}$ $()$ in
resample; deflate $n$; $n$,
    \end{lstlisting}
  \end{tabular}
\end{center}
adapted from \cite{borgstrom2016lambda}.
Clearly, $\llangle\term\rrangle_0$ is not finite (in fact, it is not even $\sigma$-finite), while $\llangle\term\rrangle_1 = \llangle\term\rrangle$ is.
\fi

Although of limited practical interest, programs with an unbounded number of calls to \ttt{resample} are of interest from a semantic perspective.
If we have $\lim_{n \to \infty} \llangle \term \rrangle_n = \llangle \term \rrangle$ pointwise, then any \gls{smc} algorithm approximating the sequence $\llangle \term \rrangle_n$ also approximates $\llangle \term \rrangle$, at least asymptotically in the number of steps $n$.
\ifextended
First, consider the variation $\term_\x{geo-res}$ of the geometric program $\termgeo$ in Fig.~\ref{fig:geometric} given by
\begin{equation}\label{eq:geores}
  \hspace{-2cm}
  \begin{tabular}{c}
    \begin{lstlisting}
let rec $\x{geometric}$ $\_$ =
  (*@\fbox{resample;}\hspace{-4pt}@*) if sample$_{\x{bern}}(0.6)$ then $1$ + $\x{geometric}$ $()$ else $1$
in $\x{geometric}$ $()$(*@\text{.}@*)
    \end{lstlisting}
  \end{tabular}
\end{equation}
The only difference is the added \ttt{resample} (marked with a box).
Here $\llangle\termgeo\rrangle = \llangle\term_\x{geo-res}\rrangle$, since, in general, $\llangle\term\rrangle$ is unaffected by placing \ttt{resample}s in $\term$.
\fi
Note that $\term_\x{geo-res}$ has no upper bound on the number of calls to \ttt{resample}, and therefore Theorem~\ref{theorem:main} is not applicable.
\ifextended
We have, however, that
\begin{equation}
  \begin{aligned}
    &\llangle \term_\x{geo-res} \rrangle_n(S) =
    \\&\qquad
    \sum_{i = 1}^{n-1}
    \lambda_i(S \cap \left([0,0.6]^{i-1} \times
    (0.6,1]\right)))
    +
    \lambda_n(S \cap \left([0,0.6]^{i-1} \times
    [0,1]\right)),
  \end{aligned}
\end{equation}
and as a consequence, $\lim_{n\to\infty} \llangle \term_\x{geo-res} \rrangle_n = \llangle \term_\x{geo-res} \rrangle$ pointwise.
The question is then if
$\lim_{n \to \infty} \llangle \term \rrangle_n = \llangle \term \rrangle$
pointwise holds in general? The answer is no, as we demonstrate next.
\fi

For $\lim_{n\to\infty} \llangle \term \rrangle_n = \llangle \term \rrangle$ to hold pointwise, it must hold that $\lim_{n\to\infty} f_{\term,n} = f_\term$ pointwise $\mu_\Tr$-ae.
Unfortunately, this does not hold for all programs.
Consider the program $\term_\x{loop}$ defined by
\lstinline!let rec loop _ = resample; loop $()$ in loop $()$!.
Here, $f_\termloop = 0$ since the program diverges deterministically, but $f_{\termloop,n}(\trace{}) = 1$ for all $n$.
Because $\mu_\Tr(\{{\trace{}}\}) \neq 0$, we do not have $\lim_{n\to\infty} f_{\termloop,n} = f_\termloop$ pointwise $\mu_\Tr$-ae.

Even if we have $\lim_{n \to \infty} f_{\term,n} = f_\term$ pointwise $\mu_\Tr$-ae, we might not have $\lim_{n\to\infty} \llangle \term \rrangle_n = \llangle \term \rrangle$ pointwise.
Consider, for instance, the program $\termunit$ given by
\begin{equation}\label{eq:dctexample}
  \begin{tabular}{c}
    \begin{lstlisting}
let s = sample$_{\mathcal{U}}(0,1)$ in
let rec foo $n$ =
  if s $\leq$ $1/n$ then resample; weight $2$; foo ($2\cdot n$) else weight $0$ in
foo 1
    \end{lstlisting}
  \end{tabular}
\end{equation}
We have $f_\termunit = 0$ and $f_{\termunit,n} = 2^n\cdot\mathbf{1}_{[0,1/2^n]}$ for $n > 0$.
Also,
$
  \lim_{n \to \infty} f_{\termunit,n} = f_{\termunit}
$
pointwise.
\ifextended
However,
\begin{equation}
  \begin{aligned}
    \lim_{n \to \infty}
    \llangle \termunit \rrangle_n
    (\Tr)
    &=
    \lim_{n \to \infty}
    \int_\Tr f_{\termunit,n} d\mu_\Tr(s)
    =
    \lim_{n \to \infty}
    \int_{[0,1]} f_{\termunit,n} d\lambda(x)
    \\
    &=
    \lim_{n \to \infty}
    \int_{[0,1/2^n]} 2^n d\lambda(x)
    = 1
    \\
    &\neq
    0 = \int_\Tr f_\termunit d\mu_\Tr(s)
    = \llangle \termunit \rrangle(\Tr).
  \end{aligned}
\end{equation}
\fi
This shows that the limit may fail to hold, even for programs that terminate almost surely, as is the case for the program $\termunit$ in \eqref{eq:dctexample}.
In fact, this program is positively almost surely terminating~\cite{bournez2005proving} since the expected number of recursive calls to \ttt{foo} is 1.

Guided by the previous example, we now state the dominated convergence theorem---a fundamental result in measure theory---in the context of \gls{smc} inference in our calculus.
\begin{theorem}\label{theorem:dct}
  \theoremdct
\end{theorem}
For a proof, see McDonald and Weiss~\cite[Theorem 4.9]{mcdonald2012course}.
It is easy to check that for our example in \eqref{eq:dctexample}, there is no dominating and integrable $g$ as is required in Theorem~\ref{theorem:dct}.
We have already seen that the conclusion of the theorem fails to hold here.
As a corollary, if there exists a dominating and integrable $g$, the measures $\llangle\term\rrangle_n$ are always finite.
\begin{corollary}
  If there exists a measurable function $g :
  (\Tr,\Trcal) \to (\R_+,\Bor_+)$ such that $f_{\term,n} \leq g$ $\mu_\Tr$-ae for
  all $n$, and $\int_\Tr g(s) d\mu_\Tr(s) < \infty$, then
  $\llangle\term\rrangle_n$ is finite for each $n \in \mathbb{N}_0$.
\end{corollary}
This holds because
$
\llangle \term \rrangle_n(\Tr)
=
\int_\Tr f_{\term,n}(s) d\mu_\Tr(s)
\leq
\int_\Tr g(s) d\mu_\Tr(s) < \infty.
$
Hence, we do not need to assume the finiteness of $\llangle\term\rrangle_n$ in order for Algorithm~\ref{alg:smc} to be applicable, as was the case for the setting of Theorem~\ref{theorem:main}.

In Theorem~\ref{theorem:consistency}, we summarize and combine the above results with Lemma~\ref{lemma:smccorrect}.
\begin{theorem}\label{theorem:consistency}
  Let $\term$ be a term, and apply Algorithm~\ref{alg:smc} with $\llangle\term\rrangle_n$ as $\widetilde \pi_n$, and with arbitrary valid kernels $k_n$.
  If the condition of Theorem~\ref{theorem:main} holds and $\llangle\term\rrangle_n$ is finite for each $n \in \mathbb{N}_0$, then Algorithm~\ref{alg:smc} approximates $\llangle\term\rrangle$ and its normalizing constant after a finite number of steps.
  Alternatively, if the condition of Theorem~\ref{theorem:dct} holds, then Algorithm~\ref{alg:smc} approximates $\llangle\term\rrangle$ and its normalizing constant in the limit $n \to \infty$.
\end{theorem}
This follows directly from Theorem~\ref{theorem:main}, Theorem~\ref{theorem:dct}, and Lemma~\ref{lemma:smccorrect}.

We conclude this section by discussing \ttt{resample} placements, and the practical implications of Theorem~\ref{theorem:consistency}.
First, we define a \emph{\ttt{resample} placement} for a term $\term$ as the term resulting from replacing arbitrary subterms $\term'$ of $\term$ with $\ttt{resample; } \term'$.
Note that such a placement directly corresponds to constructing the sequence $\llangle \term \rrangle_n$.
Second, note that the measure $\llangle \term \rrangle$ and the target measure $\llbracket \term \rrbracket$ are clearly \emph{unaffected} by such a placement---indeed, \ttt{resample} simply evaluates to $()$, and for $\llangle \term \rrangle$ and $\llbracket \term \rrbracket$, there is no bound on how many \ttt{resample}s we can evaluate.
As such, we conclude that \emph{all} resample placements in $\term$ fulfilling one of the two conditions in Theorem~\ref{theorem:consistency} leads to a correct approximation of $\llangle\term\rrangle$ when applying Algorithm~\ref{alg:smc}.
Furthermore, there is always, in practice, an upper bound on the number of calls to \ttt{resample}, since any concrete run of \gls{smc} has an (explicit or implicit) upper bound on its runtime.
This is a powerful result, since it implies that when implementing \gls{smc} for \glspl{ppl}, any method for selecting resampling locations in a program is correct under mild conditions (Theorem~\ref{theorem:main} or Theorem~\ref{theorem:dct}) that are most often, if not always, fulfilled in practice.
Most importantly, this justifies the basic approach for placing \ttt{resample}s found in WebPPL, Anglican, and Birch, in which every call to \ttt{weight} is directly followed (implicitly) by a call to \ttt{resample}.
It also justifies the approach to placing \ttt{resample}s described in Lundén et al.~\cite{lunden2018automatic}. This latter approach is essential in, e.g., Ronquist et al.~\cite{ronquist2020probabilistic}, in order to increase inference efficiency.

Our results also show that the restriction in Anglican requiring all executions to encounter the same number of \ttt{resample}s, is too conservative. Clearly, this is not a requirement in either Theorem~\ref{theorem:main} or Theorem~\ref{theorem:dct}.
For instance, the number of calls to \ttt{resample} varies significantly in \eqref{eq:geores}.

\section{SMC Algorithms}\label{sec:smckernel}
In this section, we take a look at how the kernels $k_n$ in Algorithm~\ref{alg:smc} can be instantiated to yield the concrete \gls{smc} algorithm known as the \acrlong{bpf} (Section~\ref{sec:bpf}), and also discuss other \gls{smc} algorithms and how they relate to Algorithm~\ref{alg:smc} (Section~\ref{sec:othersmc}).

\subsection{The Bootstrap Particle Filter}\label{sec:bpf}
We define for each term $\term$ a particular sequence of kernels $k_{\term,n}$, that gives rise to the \gls{smc} algorithm known as the \acrfull{bpf}.
Informally, these kernels correspond to simply continuing to evaluate the program until either arriving at a value $\val$ or a term of the form $\econt[\ttt{resample}]$.
For the bootstrap kernel, calculating the weights $w_n^j$ from Algorithm~\ref{alg:smc} is particularly simple.

Similarly to $\llangle \term \rrangle_n$, it is more convenient to define and work with sequences of kernels over traces, rather than terms.
We will define $k_{\term,n}(s,\cdot)$ to be the sub-probability measure over extended traces $s \ast s'$ resulting from evaluating the term $r_{\term,n-1}(s)$ until the next \ttt{resample} or value $\val$, ignoring any call to \ttt{weight}.
First, we immediately have that the set of all traces that do not have $s$ as prefix must have measure zero.
To make this formal, we will use the inverse images of the functions $\x{prepend}_s(s') = s \ast s'$, $s \in \Tr$ in the definition of the kernel.
\begin{lemma}\label{lemma:prependmeas}
  \lemmaprependmeas$^\dagger$
\end{lemma}
The next ingredient for defining the kernels $k_{\term,n}$ is a function $p_{\term,n}$ that indicates what traces are possible when executing $\term$ until the $n + 1$th \ttt{resample} or value.
\begin{definition}\label{def:termdensitykerneln}
  $
    p_{\term,n}(s) =
    \begin{cases}
      1 & \text{if }
      \term,\cdot,s,n \hookrightarrow^*
      \val,\cdot,\trace{},\cdot \\
      1 & \text{if }
      \term,\cdot,s,n \hookrightarrow^*
      \econt[\ttt{resample}],\cdot,\trace{},0 \\
      0 & \text{otherwise}
    \end{cases}
  $
\end{definition}
Note the similarities to Definition~\ref{def:termdensityn}.
In particular, $f_{\term,n}(s) > 0$ implies $p_{\term,n}(s) = 1$.
However, note that $f_{\term,n}(s) = 0$ does not imply $p_{\term,n}(s) = 0$, since $p_{\term,n}$ ignores weights.
As an example, \mbox{$f_{(\ttt{weight $0$}),n}(\trace{}) = 0$,} while $p_{(\ttt{weight $0$}),n}(\trace{}) = 1$.

\begin{lemma}
  $p_{\term,n} : (\Tr,\Trcal) \to (\R_+,\Bor_+)$ is measurable.
\end{lemma}
The proof is analogous to that of Lemma~\ref{lemma:rfnmeas}.
We can now formally define the kernels $k_{\term,n}$.
\begin{definition}\label{def:kernelktermn}
  $
    k_{\term,n}(s,S) =
    \int_{prepend_s^{-1}(S)}
    p_{r_{\term,n-1}(s),1}(s')
    \ d\mu_\Tr(s') \\
  $
\end{definition}
By the definition of $p_{\term,n}$, the $k_{\term,n}$ are \emph{sub-probability} kernels rather than probability kernels.
Intuitively, the reason for this is that during evaluation, terms can get stuck, deterministically diverge, or even stochastically diverge.
Such traces are assigned 0 weight by $p_{\term,n}$.
\begin{lemma}\label{lemma:finkern}
  \lemmafinitekernel$^\dagger$\footnote{We only give a partial proof of this lemma.}%
\end{lemma}

We get a natural starting measure $q_{0}$ from the sub-probability distribution resulting from running the initial program $\term$ until reaching a value or a call to \ttt{resample}, ignoring weights.
\begin{definition}\label{def:bootstrapstartingmeasure}
  $\langle \term \rangle_0(S) = \int_S p_{\term,0}(s) d\mu_\Tr(s)$.
\end{definition}

Now we have all the ingredients for the general \gls{smc} algorithm described in Section~\ref{sec:smcalg}: a sequence of target measures $\llangle \term \rrangle_n = \widetilde \pi_n$ (Definition~\ref{def:seqtarget}), a starting measure $\langle\term\rangle_0 \propto q_0$ (Definition~\ref{def:bootstrapstartingmeasure}), and a sequence of kernels $k_{\term,n} \propto k_n$ (Definition~\ref{def:kernelktermn}).
These then induce a sequence of proposal measures $\langle \term \rangle_n = \widetilde q_n$ as in Equation~(\ref{eq:qn}), which we instantiate in the following definition.
\begin{definition}
  $
    \langle \term \rangle_n(S) =
    \int_\Tr k_{\term,n}(s, S) f_{\term,n-1}(s)d\mu_\Tr(s)
    , \quad n > 0
  $
\end{definition}

Intuitively, the measures $\langle\term\rangle_n$ are obtained by evaluating the terms in the support of the measure $\llangle\term\rrangle_{n-1}$ until reaching the next \ttt{resample} or value.
For an efficient implementation, we need to factorize this definition into the history and the current step, which amounts to splitting the traces.
Each feasible trace can be split in such a way.
\begin{lemma}\label{lemma:decompose}
  \lemmadecompose$^\dagger$
\end{lemma}
This gives a more efficiently computable definition of the density.
\begin{lemma}\label{lemma:propdensity}
  \lemmapropdensity%
  $^\dagger$\footnote{We only give a proof sketch for this lemma.}
\end{lemma}
\ifextended
Analogously to Definitions~\ref{def:tracemeasure} and~\ref{def:seqtarget}, by Lemma~\ref{lemma:density} and Lemma~\ref{lemma:tracemeasspace}, it holds that
the density
$
  s \mapsto
  p_{r_{\term,n-1}(\underline{s}),1}(\overline{s})
  f_{\term,n-1}(\underline{s})
$
is unique $\mu_\Tr$-ae if $\langle\term\rangle_n$ is $\sigma$-finite.
\fi

Since the kernels $k_{\term,n}$ are sub-probability kernels, the measures $\langle \term \rangle_n$ are finite given that the $\llangle\term\rrangle_{n}$ are finite.
\begin{lemma}\label{lemma:propfin}
  \lemmapropfin$^\dagger$
\end{lemma}
\begin{algorithm}[tb]
  \caption{%
    A concrete instantiation of Algorithm~\ref{alg:smc} with $\widetilde \pi_n = \llangle\term\rrangle_n$, $k_n \propto k_{\term,n}$, $q_0 \propto \langle\term\rangle_0$, and as a consequence $\widetilde q_n = \langle\term\rangle_n$ (for $n > 0$).
    In each step, we let $1 \leq j \leq J$, where $J$ is the number of samples.
  }
  \label{alg:smcppl}
  \begin{enumerate}
    \item \textbf{Initialization:}
      \label{alg:smcppl:init}
      Set $n = 0$.
      Draw $s_{0}^j \sim \langle\term\rangle_0$ for $1\leq j \leq J$.

      That is, run the program $\term$, and draw from $\mathcal{U}(0,1)$ whenever required by a $\ttt{sample}_D$.
      Record these draws as the trace $s_0^j$.
      Stop when reaching a term of the form $\econt[\ttt{resample}]$ or a value $\val$.
      The empirical distribution $\{s_0^j\}_{j=1}^J$ approximates $\langle\term\rangle_0$.
    \item \textbf{Correction:}
      \label{alg:smcppl:correction}
      Calculate $w_n^j = \frac {f_{\llangle\term\rrangle_n}(s_n^j)} {f_{\langle\term\rangle_n}(s_n^j)}$ for $1\leq j \leq J$.

      As a consequence of Lemma~\ref{lemma:weights}, this is trivial.
      Simply set $w_n^j$ to the weight accumulated while running $\term$ in step \eqref{alg:smcppl:init}, or $r_{\term,n-1}(\hat{s}_{n-1}^j)$ in step \eqref{alg:smcppl:mutation}.
      The empirical distribution given by $\{(s_n^j,w_n^j)\}_{j=1}^J$ approximates $\llangle\term\rrangle_n/Z_{\llangle\term\rrangle_n}$.

    \item \textbf{Termination:}
      If all samples $r_\term(s_{n}^j)$ are values, terminate and output $\{(s_n^j,w_n^j)\}_{j=1}^J$. If not, go to the next step.

      We cannot evaluate values further, so running the algorithm further if all samples are values is pointless.
      When terminating, assuming the conditions in Theorem~\ref{theorem:main} or Theorem~\ref{theorem:dct} holds, $\{(s_n^j,w_n^j)\}_{j=1}^J$ approximates $\llangle\term\rrangle/Z_{\llangle\term\rrangle_n}$.
      Also, by the definition of $\llbracket\term\rrbracket$, $\{(r_\term(s_n^j),w_n^j)\}_{j=1}^J$ approximates $\llbracket\term\rrbracket/Z_{\llbracket\term\rrbracket_n}$, the normalized version of $\llbracket\term\rrbracket$.

    \item \textbf{Selection:}
      Resample the empirical distribution $\{(s_n^j,w_n^j)\}_{j=1}^J$.
      The new empirical distribution is unweighted and given by $\{\hat{s}_n^j\}_{j=1}^J$.
      This distribution also approximates $\llangle\term\rrangle_n/Z_{\llangle\term\rrangle_n}$.

    \item \textbf{Mutation:}
      \label{alg:smcppl:mutation}
      Increment $n$.
      Draw $s_{n}^j \sim k_{\term,n}(\hat{s}_{n-1}^j,\cdot)$ for $1\leq j \leq J$.

      That is, simply run the intermediate program $r_{\term,n-1}(\hat{s}_{n-1}^j)$, and draw from $\mathcal{U}(0,1)$ whenever required by a $\ttt{sample}_D$.
      Record these draws and append them to $\hat{s}_{n-1}^j$, resulting in the trace $s_n^j$.
      Stop when reaching a term of the form $\econt[\ttt{resample}]$ or a value $\val$.
      The empirical distribution $\{s_n^j\}_{j=1}^J$ approximates $\langle\term\rangle_n/Z_{\langle\term\rangle_n}$.
      Go to \eqref{alg:smcppl:correction}.

    \end{enumerate}
\end{algorithm}
As discussed in Section~\ref{sec:smcpplcorrect}, the $\llangle\term\rrangle_n$ are finite, either by assumption (Theorem~\ref{theorem:main}) or as a consequence of the dominating function of Theorem~\ref{theorem:dct}.
From this and Lemma~\ref{lemma:propfin}, the $\langle\term\rangle_n$ are also finite.
Furthermore, checking that $\langle\term\rangle_n$ are valid, i.e.\ that the density $f_{\langle \term \rangle_n}$ of each $\langle\term\rangle_n$ covers the density $f_{\llangle \term \rrangle_n}$ of $\llangle\term\rrangle_n$ is trivial.
As such, by Lemma~\ref{lemma:smccorrect}, we can now correctly approximate $\llangle\term\rrangle_n$ using Algorithm~\ref{alg:smc}.
The details are given in Algorithm~\ref{alg:smcppl}, which closely resembles the standard \gls{smc} algorithm in WebPPL\@.
For ease of notation, we assume it possible to draw samples from $\langle\term\rangle_0$ and $k_{\term,n}(s,\cdot)$, even though these are sub-probability measures.
This essentially corresponds to assuming evaluation never gets stuck or diverges.
Making sure this is the case is not within the scope of this paper.
The weights in Algorithm~\ref{alg:smcppl} at time step $n$ can easily be calculated according to the following lemma.
\begin{lemma}\label{lemma:weights}
  \lemmaweights$^\dagger$
\end{lemma}
\ifextended
Finally, it is now obvious how the \ttt{resample} construct relates to the resampling in the selection step in Algorithm~\ref{alg:smcppl}---only traces for which $r_{\term,n}(s_n^j)$ is a term of the form $\econt[\ttt{resample}]$, or a value, will issue from the mutation step and thus participate in resampling at the selection step.
As a consequence of how the kernels $k_{\term,n}$ are constructed, we only stop at such terms in steps \eqref{alg:smcppl:init} and \eqref{alg:smcppl:mutation} when running the program.
This is the reason for naming the construct \ttt{resample}.
\fi

\subsection{Other SMC Algorithms}\label{sec:othersmc}
In this section, we discuss \gls{smc} algorithms other than the \gls{bpf}.

First, we have the \emph{resample-move} algorithm by Gilks and Berzuini~\cite{gilks2001following}, which is also implemented in WebPPL~\cite{goodman2008church}, and treated by~Chopin~\cite{chopin2004central} and \'{S}cibior et al.~\cite{scibior2017denotational}.
In this algorithm, the \gls{smc} kernel is composed with a suitable \gls{mcmc} kernel, such that one or more \gls{mcmc} steps are taken for each sample after each resampling.
This helps with the so-called degeneracy problem in \gls{smc}, which refers to the tendency of \gls{smc} samples to share a common ancestry as a result of resampling.
We can directly achieve this algorithm in our context by simply choosing appropriate transition kernels in Algorithm~\ref{alg:smc}.
Let $k_{\text{MCMC},n}$ be \gls{mcmc} transition kernels with $\widetilde\pi_{n-1} = \llangle\term\rrangle_{n-1}$ as \emph{invariant distributions}.
Using the bootstrap kernels as the main kernels, we let $k_n = k_{\term,n} \circ k_{\text{MCMC},n}$ where $\circ$ denotes kernel composition.
The sequence $k_n$ is valid because of the validity of the main \gls{smc} kernels and the invariance of the \gls{mcmc} kernels.

While Algorithm~\ref{alg:smc} captures different \gls{smc} algorithms by allowing the use of different kernels, some algorithms require changes to Algorithm~\ref{alg:smc} itself.
The first such variation of Algorithm~\ref{alg:smc} is the \emph{alive} particle filter, recently discussed by Kudlicka et al.~\cite{kudlicka2019probabilistic}, which reduces the tendency to degeneracy by not including sample traces with zero weight in resampling. This is done by repeating the selection and mutation steps (for each sample individually) until a trace with non-zero weight is proposed; the corresponding modifications to Algorithm~\ref{alg:smc} are straightforward. The unbiasedness result of Kudlicka et al.~\cite{kudlicka2019probabilistic} can easily be extended to our \gls{ppl} context, with another minor modification to Algorithm~\ref{alg:smc}.

Another variation of Algorithm~\ref{alg:smc} is the auxiliary particle filter~\cite{pitt1999filtering}.
Informally, this algorithm allows the selection and mutation steps of Algorithm~\ref{alg:smc} to be guided by future information regarding the weights $w_n$.
For many models, this is possible since the weighting functions $w_n$ from Algorithm~\ref{alg:smc} are often parametric in an explicitly available sequence of \emph{observation data points}, which can also be used to derive better kernels $k_n$.
Clearly, such optimizations are model-specific, and can not directly be applied in expressive \gls{ppl} calculi such as ours.
However, the general idea of using look-ahead in general-purpose \glspl{ppl} to guide selection and mutation is interesting, and should be explored.

\section{Related Work}\label{sec:rel}

The only major previous work related to formal \gls{smc} correctness in \glspl{ppl} is \'{S}cibior et al.~\cite{scibior2017denotational} (see Section~\ref{sec:intro}).
They validate both the \gls{bpf} and the resample-move \gls{smc} algorithms in a denotational setting.
In a companion paper, \'{S}cibior et al.~\cite{scibior2018functional} also give a Haskell implementation of these inference techniques.

Although formal correctness proofs of \gls{smc} in \glspl{ppl} are sparse, there are many languages that implement \gls{smc} algorithms.
Goodman and Stuhlm\"{u}ller~\cite{goodman2014design} describe \gls{smc} for the probabilistic programming language WebPPL\@.
They implement a basic \gls{bpf} very similar to Algorithm~\ref{alg:smcppl}, but do not show correctness with respect to any language semantics.
Also, related to WebPPL, Stuhlm\"{u}ller et al.~\cite{stuhlmuller2015coarsetofine} discuss a coarse-to-fine \gls{smc} inference technique for probabilistic programs with independent sample statements.

Wood et al.~\cite{wood2014new} describe PMCMC, an \gls{mcmc} inference technique that uses \gls{smc} internally, for the probabilistic programming language Anglican \cite{design2016tolpin}.
Similarly to WebPPL, Anglican also includes a basic \gls{bpf} similar to Algorithm~\ref{alg:smcppl}, with the exception that every execution needs to encounter the same number of calls to \ttt{resample}.
They use various types of empirical tests to validate correctness, in contrast to the formal proof found in this paper.
Related to Anglican, a brief discussion on resample placement requirements can be found
in van de Meent et al.~\cite{vandemeent2018introduction}.

Birch \cite{murray2018automated} is an imperative object-oriented \gls{ppl}, with a particular focus on \gls{smc}.
It supports a number of \gls{smc} algorithms, including the \gls{bpf}~\cite{gordon1993novel} and the auxiliary particle filter~\cite{pitt1999filtering}.
Furthermore, they support dynamic analytical optimizations, for instance using locally-optimal proposals and Rao--Blackwellization~\cite{murray2018delayed}.
As with WebPPL and Anglican, the focus is on performance and efficiency, and not on formal correctness.

There are quite a few papers studying the correctness of \gls{mcmc} algorithms for \glspl{ppl}.
Using the same underlying framework as for their \gls{smc} correctness proof, \'{S}cibior et al.~\cite{scibior2017denotational} also validates a trace \gls{mcmc} algorithm.
Another proof of correctness for trace \gls{mcmc} is given in Borgström et al.~\cite{borgstrom2016lambda}, which instead uses an untyped lambda calculus and an operational semantics.
Much of the formalization in this paper is based on constructions used as part of their paper.
For instance, the functions $f_\term$ and $r_\term$ are defined similarly, as well as the measure space $(\Tr,\Trcal,\mu_\Tr)$ and the measurable space $(\T,\Tcal)$.
Our measurability proofs of $f_\term$, $r_\term$, $f_{\term,n}$, and $r_{\term,n}$ largely follow the same strategies as found in their paper.
Similarly to us, they also relate their proof of correctness to classical results from the \gls{mcmc} literature.
A difference is that we use inverse transform sampling, whereas they use probability density functions.
As a result of this, our traces consist of numbers on $[0,1]$, while their traces consist of numbers on $\R$.
Also, inverse transform sampling naturally allows for built-in discrete distributions.
In contrast, discrete distributions must be encoded in the language itself when using probability densities.
Another difference is that they restrict the arguments to \ttt{weight} to $[0,1]$, in order to ensure the finiteness of the target measure.

Other work related to ours include Jacobs~\cite{jacobs2021paradoxes}, V\'{a}k\'{a}r et al.~\cite{vakar2019domain}, and Staton et al.~\cite{staton2016semantics}.
Jacobs~\cite{jacobs2021paradoxes} discusses problems with models in which \ttt{observe} (related to \ttt{weight}) statements occur conditionally. While our results show that \gls{smc} inference for such models is correct, the models themselves may not be useful.
V\'{a}k\'{a}r et al.~\cite{vakar2019domain} develops a powerful domain theory for term recursion in \glspl{ppl}, but does not cover \gls{smc} inference in particular.
Staton et al.~\cite{staton2016semantics} develops both operational and denotational semantics for a \gls{ppl} calculus with higher-order functions, but without recursion.
They also briefly mention \gls{smc} as a program transformation.

Classical work on \gls{smc} includes Chopin~\cite{chopin2004central}, which we use as a basis for our formalization.
In particular, Chopin~\cite{chopin2004central} provides a general formulation of \gls{smc}, placing few requirements on the underlying model.
The book by Del Moral~\cite{moral2004feynman} contains a vast number of classical \gls{smc} results, including the law of large numbers and unbiasedness result from Lemma~\ref{lemma:smccorrect}.
A more accessible summary of the important \gls{smc} convergence results from Del Moral~\cite{moral2004feynman} can be found in~Naesseth et al.~\cite{naesseth2019elements}.


\section{Conclusions}\label{sec:conclusions}
In conclusion, we have formalized \gls{smc} inference for an expressive functional \gls{ppl} calculus, based on the formalization by Chopin~\cite{chopin2004central}.
We showed that in this context, \gls{smc} is correct in that it approximates the target measures encoded by programs in the calculus under mild conditions.
Furthermore, we illustrated a particular instance of \gls{smc} for our calculus, the \acrlong{bpf}, and discussed other variations of \gls{smc} and their relation to our calculus.

As indicated in Section~\ref{sec:motivation}, the approach used for selecting resampling locations can have a large impact on \gls{smc} accuracy and performance.
This leads us to the following general question: can we select optimal resampling locations in a given program, according to some formally defined measure of optimality?
We leave this important research direction for future work.


\subsection*{Acknowledgments}
We thank our colleagues Lawrence Murray and Fredrik Ronquist for fruitful discussions and ideas.
We also thank Sam Staton and the anonymous reviewers at ESOP for their detailed and helpful comments.

%
%
%
\clearpage
\bibliographystyle{splncs04}
\bibliography{references}

\ifextended
\appendix

\clearpage
\section{SMC\@: an Illustrative Example}\label{sec:smcintro}
%
%
%
In order to fully appreciate the contributions of this paper, we devote this section to introducing SMC inference for the unfamiliar with an informal example. The example is based on Lindholm~\cite{lindholm2013particle}.

\subsection{Model}
\begin{figure}[tb]

  \pgfplotsset{particles/.style={%
      scatter,
      only marks,
      scatter src=explicit,
      fill opacity=0.5,
      draw opacity=0,
      scatter/use mapped color=black,
      scatter/@pre marker code/.append style=
      {/tikz/mark size=1+\pgfplotspointmetatransformed/200}
  }}

  \pgfplotsset{aircraft/.style={%
      enlargelimits=false,
      ymin=0,
      ymax=90,
      xmin=0,
      xmax=200,
      ticks=none,
      xlabel=Position,
      ylabel=Altitude
  }}

  \pgfdeclareplotmark{plane}{\node {\includegraphics[scale=0.03]{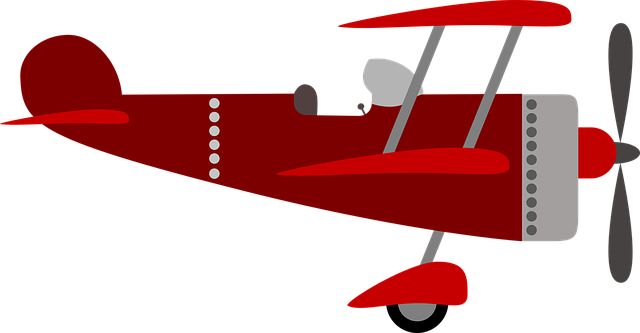}};}

  \newcommand{\addplotplane}[2]{%
    \addplot[mark=plane,nodes near coords,
    node near coords style={yshift=1mm},
    point meta={symbolic=(#1)}]
    table [header=false] {data/pos_#2.dat};
    \addplot[dashed] table [header=false] {data/dist_#2.dat};%
  }

  \newcommand{\particles}[3]{%
    \begin{tikzpicture}[trim axis left, trim axis right]
      \begin{axis}[
        width=0.9\textwidth,
        height=16mm,
        enlargelimits=false,
        xmin=0,
        xmax=200,
        ticks=none,
        axis line style={draw=none},
        ylabel=(#1),
        ylabel style={rotate=-90}
        ]
        \addplot[particles] table [x index=0,y expr=0,meta index=2,header=false]
          {data/BPF_#3#2.dat};
      \end{axis}
    \end{tikzpicture}
  }

  \centering
  \begin{tabular}{l}
    \begin{tikzpicture}[trim axis left, trim axis right]
      \begin{axis}[%
        width=0.9\textwidth,
        height=45mm,
        aircraft
        ]
        \addplot[no markers] table [header=false] {data/map.dat};
        \addplotplane{1}{1};
        \addplotplane{2}{25};
        \addplotplane{3}{50};
        \addplotplane{4}{70};
      \end{axis}
    \end{tikzpicture} \\[-3mm]
    \particles{1.1}{1}{trans_} \\[-3mm]
    \particles{1.2}{1}{} \\[-3mm]
    \particles{1.3}{1}{resample_} \\[-3mm]
    \particles{1.4}{2}{trans_} \\[-3mm]
    \particles{2}{25}{resample_} \\[-3mm]
    \particles{3}{50}{resample_} \\[-3mm]
    \particles{4}{70}{resample_}
  \end{tabular}

  \caption{%
    An illustration of the \gls{bpf} for estimating the position of an aircraft.
    A detailed description is given in the text.
  }
  \label{fig:smcaircraft}
\end{figure}

%
%
%

Consider the following scenario: a pilot is flying an aircraft in bad weather with zero visibility, and is attempting to estimate the aircraft's position.
In order to do this, available is an elevation map of the area, a noisy altimeter, and a noisy sensor for measuring the vertical distance to the ground (see Fig.~\ref{fig:smcaircraft} for an illustration).
Concretely, assume that
\begin{enumerate}
  \item[(a)]
    $X_{0:t} = X_0,X_1,\ldots,X_t$ are real-valued random variables
    representing the \emph{true} horizontal position of the aircraft at the
    discrete time steps $0,1,\ldots,t$, and
  \item[(b)]
    $Y_{0:t} = Y_0,Y_1,\ldots,Y_t$ are
    real-valued random variables for the measurements given by subtracting the
    vertical distance sensor reading from the altimeter sensor reading.
\end{enumerate}
The problem we consider is to estimate the positions $X_n$, $n \leq t$, given all combined sensor measurements $Y_{0:n}$ collected up until time $n$.
This random variable is denoted $X_n \mid Y_{0:n}$, and the distribution for this random variable is known as the \emph{target measure}.
In general, $X \mid Y$ denotes the random variable $X$ conditioned on $Y$ having been observed.

Concretely, we assume the following model for $n \in \mathbb{N}$:
\begin{equation}
  \begin{gathered}
    X_0 \sim \mathcal{U}(0, 100) \qquad
    X_n \mid X_{n-1} \sim \mathcal{N}(X_{n-1} + 2, 1) \\
    Y_n \mid X_n \sim \mathcal{N}(\mathit{elevation}(X_n), 2).
  \end{gathered}
\end{equation}
In other words, we have that the initial position $X_0$ of the aircraft is uniformly distributed between 0 and 100, and at each time step $n$, $X_n$ is normally distributed around $X_{n-1} + 2$ with variance $1$ (the conditional distribution of $X_n \mid X_{n-1}$ is known as a \emph{transition kernel}).
Finally the combined measurement $Y_n$ from the sensors is normally distributed around the true elevation of the ground at the current horizontal position $X_n$ with variance $2$, where the true position is given by our elevation map, here modeled as a function $\x{elevation}$.

\subsection{Inference}
With the model in place, we can proceed to sequentially estimating the probability distributions for the random variables $X_n \mid Y_{0:n}$ using the \gls{bpf}, a fundamental \gls{smc} algorithm.
In Section~\ref{sec:bpf}, we will give a formal definition of this algorithm for models encoded in our calculus.
Here, we instead give an informal description for our current aircraft model.
In Fig.~\ref{fig:smcaircraft}, we show the true initial aircraft position (1), and the true position at three later time steps, denoted by (2), (3), and (4).
In addition, for each of these time steps, we show the empirical \gls{smc} approximations to the distributions for $X_n \mid Y_{0:n}$, where $n$ is increasing for each of the four positions.
Step (1) is further broken down into substeps (1.1)-(1.4).
The empirical approximations are given by a set of \emph{weighted random samples}, where the weights are indicated by the sizes of each individual sample.
We give the details for each time step in Fig.~\ref{fig:smcaircraft} below.
\begin{enumerate}[align=left]
  \item[(1.1)]
    Begin by drawing many samples from $\mathcal{U}(0,100)$.
    These samples represent the distribution for $X_0$, the initial horizontal position.
  \item[(1.2)]
    Consider the first observation $Y_0$, given by the sensors at time step 0.
    For each drawn sample in (1.1), the relative likelihood of seeing the particular observation $Y_0$ varies.
    For example, the position (1) in Fig.~\ref{fig:smcaircraft} is much more likely to have produced the sensor reading $Y_0$ compared to (3) (since (1) is in fact the true position).
    Because of this, we \emph{weight} each sample according to the observation $Y_0$.
    Logically, we see that positions with ground elevation similar to the true position (1) are assigned the most weight.
  \item[(1.3)]
    Next, we take the set of weighted particles from the previous time step and \emph{resample} them according to their weights.
    That is, we draw (with replacement) a set of new samples from the previous set of samples, based on their relative weights.
    We see that the samples with high weight are indeed the ones to survive this resampling step.
    Note that after resampling, we also reset the weights (which is required for correctness).
  \item[(1.4)]
    For each sample of $X_0$, draw from the distribution of $X_1 \mid X_0$ to propagate it forwards by one time step.
  \item[(2)]
    At this point, we have completed many iterations of the above four sub-steps---the exception being that in the first sub-step, we don't draw from $\mathcal{U}(0,100)$, but instead reuse the set of particles from the previous step.
    We see that the set of samples now correctly cluster on the true position.
  \item[(3)]
    Here, we have flown over a body of water for some time.
    Due to this, the recent sensor readings have not been very informative, and the set of samples have diverged slightly, representing the increased uncertainty in the aircraft's position.
  \item[(4)]
    When encountering more varied terrain once again, the uncertainty is reduced, and the set of samples again cluster more closely on the true position.
\end{enumerate}
The key step in every \gls{smc} algorithm is the resampling step illustrated above.
Resampling allows for focusing the empirical approximations on regions of the sample space with high probability, yielding efficient inference for many models of practical interest.
For instance, \gls{smc} is commonly used in tracking problems \cite{arulampalam2002tutorial,isard1998condensation}.


\begin{figure}[tb]
  \centering
  \begin{tabular}{c}
    \begin{lstlisting}
let $\x{observations}$ = [$c_0,c_1,c_2,\ldots,c_{t-1}$] in
let $\x{observe}$ $x$ $o$ = weight$(f_\mathcal{N}(\x{elevation}$ $x,2,o))$ in
let $\x{sim}$ $x_n$ $o$ = $\x{observe}$ $x_n$ $o$; sample$_\mathcal{N}(x_n + 2, 1)$ in
let $x_0$ = sample$_{\mathcal{U}}(0,100)$ in
let $x_t$ = $\x{foldl}$ $\x{sim}$ $x_0$ $\x{observations}$ in
$\x{observe}$ $x_t$ $c_t$; $x_t$
    \end{lstlisting}
  \end{tabular}
  \caption{%
    The aircraft model, encoded as a program $\termair$.
    We assume that the $\x{elevation}$ function is previously defined, and that $\x{foldl}$ implements a standard left fold.
    For various empirical visualizations of this model (for increasing sets of observations), see Fig.~\ref{fig:smcaircraft}.
  }
  \label{fig:smcaircraftcode}
\end{figure}
It is also possible to encode the example as a program in the calculus from Section~\ref{sec:calculus}.
This is done in Fig.~\ref{fig:smcaircraftcode}.
The real numbers $c_0,c_1,c_2,\ldots,c_t$ in the program correspond to the observations of $Y_{0:t}$.

\section{Definitions and Proofs}
In this appendix, we prove lemmas found throughout the main article.
First, we introduce measure theory and Borel spaces (Section~\ref{sec:prelim1}), and define pointwise convergence of functions (Section~\ref{sec:prelimconvergence}).
Then, we introduce metric spaces and their properties (Section~\ref{sec:prelimmetric}), and look closer at the measure space $(\Tr,\Trcal,\mu_\Tr)$ (Section~\ref{seca:tracespace}) and the measurable space $(\T,\Tcal)$ (Section~\ref{seca:termspace}).
In Section~\ref{sec:prelimscont} and Section~\ref{sec:functioninduce}, we establish further results required for proving the measurability of $r_\term$ and $f_\term$ (Section~\ref{sec:rfmeasurable}), and $r_{\term,n}$ and $f_{\term,n}$ (Section~\ref{sec:rfnmeasurable}).
Lastly, we look at the bootstrap particle filter kernels $k_{\term,n}$ and induced proposal measures $\langle \term \rangle_n$ (Section~\ref{sec:kernelproof}).

\subsection{Preliminaries: Measure Theory and Borel Spaces}
\label{sec:prelim1}
This section gives fundamental definitions and lemmas from measure theory, and defines Borel spaces.
For a more pedagogical introduction to the subject, we recommend McDonald and Weiss~\cite{mcdonald2012course}.

\begin{definition}
  Let $\mathbb{A}$ be a set.
  We say that $\mathcal{A} \subset \mathcal{P}(\mathbb{A})$ is a \emph{$\sigma$-algebra} on $\mathbb{A}$ if
  (1) $\mathcal{A} \neq \emptyset$,
  (2) if $A \in \mathcal{A}$, then $A^c \in \mathcal{A}$, and
  (3) if $\{A_n\}_n \subset \mathcal{A}$ is countable, then $\bigcup_n A_n \in \mathcal{A}$.
  Furthermore, we call $(\mathbb{A},\mathcal{A})$ a \emph{measurable space} if $\mathcal{A}$ is a $\sigma$-algebra on $\mathbb{A}$.
\end{definition}

\begin{definition}
  Let $(\mathbb{A},\mathcal{A})$ and $(\mathbb{A}',\mathcal{A}')$ be measurable spaces.
  A function $f : \mathbb{A} \to \mathbb{A}'$ is called \emph{measurable} if $f^{-1}(A') \in \mathcal{A}$ for each $A' \in \mathcal{A}'$.
  To indicate that a function is measurable with respect to specific measurable spaces, we write $f : (\mathbb{A},\mathcal{A}) \to (\mathbb{A}',\mathcal{A}')$.
\end{definition}

\begin{definition}
  Let $(\mathbb{A},\mathcal{A})$ be a measurable space, and let $\R_+^* = \R_+ \cup \{\infty\}$.
  A function $\mu : \mathcal{A} \to \R_+^*$ is called a \emph{measure} if
  (1) $\mu(A) \geq 0$ for all $A \in \mathcal{A}$,
  (2) $\mu(\emptyset) = 0$, and
  (3) if $\{A_n\}_n \subset \mathcal{A}$ is countable, and such that $A_i \cap A_j = \emptyset$ for $i \neq j$, then $\mu\left(\bigcup_n A_n\right) = \sum_n \mu(A_n)$.
  Furthermore, we call $(\mathbb{A},\mathcal{A},\mu)$ a \emph{measure space} if $\mathcal{A}$ is a $\sigma$-algebra on $\mathbb{A}$, and $\mu$ is a measure on $\mathcal{A}$.
\end{definition}

\begin{definition}
  Let $(\mathbb{A},\mathcal{A},\mu)$ be a measure space.
  We say that the measure $\mu$ is
  a sub-probability measure if $\mu(\mathbb{A}) < 1$,
  a probability measure if $\mu(\mathbb{A}) = 1$,
  finite if $\mu(\mathbb{A}) < \infty$, and
  $\sigma$-finite if there exists a countable set $\{A_n\}_n$ such that $\bigcup_n A_n = \mathbb{A}$ and $\mu(A_n) < \infty$ for each $A_n$.
\end{definition}


\begin{definition}
  Let $\mathbb{A}$ be a set, and $\mathbf{A} \subset \mathcal{P}(\mathbb{A})$.
  We denote by $\sigma(\mathbf{A})$ the smallest $\sigma$-algebra such that $\mathbf{A} \subset \sigma(\mathbf{A}) \subset \mathcal{P}(\mathbb{A})$.
\end{definition}

\begin{definition}
  Let $\mathbf{A} \subset \mathcal{P}(\mathbb{A})$ and $A \subset \mathbb{A}$.
  We define
  $
  \mathbf{A}|_A = \{ A' \cap A \mid A' \in \mathbf{A} \}.
  $
  $\mathbf{A}|_A$ is known as the restriction of $\mathbf{A}$ on $A$.
\end{definition}

\begin{lemma}\label{lemma:measrest}
  If $(\mathbb{A},\mathcal{A})$ is a measurable space, and $A \in \mathcal{A}$, then
  $
    \mathcal{A}|_A = \{ A' \subset A \mid A' \in \mathcal{A} \}.
  $
  Furthermore, $(A,\mathcal{A}|_A)$ is a measurable space.
\end{lemma}
\begin{proof}
  See Papadimitrakis~\cite[Proposition 1.8]{papadimitrakis2004measure}.
\end{proof}

\begin{definition}
  Denote by $\mathbf{O}_n$ the standard set of all open sets in $\mathbb{R}^n$.
  We let
  $\Bor^n = \sigma(\mathbf{O}_n)$,
  $\Bor_+^n = \Bor^n|_{\R_+^n}$, and
  $\Bor_{[0,1]}^n = \Bor^n|_{[0,1]^n}$.
  $\Bor^n$ is known as the \emph{Borel $\sigma$-algebra on $\mathbb{R}^n$}.
  Furthermore, we define $\Bor = \Bor^1$, $\Bor_+ = \Bor_+^1$, and $\Bor_{[0,1]} = \Bor_{[0,1]}^1$.
\end{definition}

\begin{definition}\label{def:lebesgue}
  Let $(\mathbb{A},\mathcal{A},\mu)$ be a measure space and $f : (\mathbb{A},\mathcal{A}) \to (\R,\Bor)$ be a measurable function.
  For $A \in \mathcal{A}$, we denote by
  $
    \int_A f(a) d\mu(a)
  $
  the standard \emph{Lebesgue integral of $f$ over $A$ with respect to $\mu$}.
\end{definition}

\begin{definition}\label{def:ae}
  Let $(\mathbb{A},\mathcal{A},\mu)$ be a measure space.
  We say that a property holds $\mu$ almost everywhere, or $\mu$-ae for short, if there is a set $B\in\mathcal{A}$ of $\mu$-measure 0 such that the property holds on $\mathbb{A}\setminus B$.
\end{definition}
When $\mu$ is a (sub-)probability measure, the term ``almost surely'' is used interchangeably with ``almost everywhere''.

\begin{lemma}\label{lemma:density}
  Let $(\mathbb{A},\mathcal{A},\mu)$ be a measure space.
  If $f : (\mathbb{A},\mathcal{A}) \to (\R_+,\Bor_+)$ is a measurable function, then
  $
    \mu'(A) = \int_A f(a)d\mu(a)
  $
  is a measure on $(\mathbb{A},\mathcal{A})$.
  We say that $f$ is a \emph{density} of $\mu'$ with respect to $\mu$.
  Also, if $\mu$ and $\mu'$ are $\sigma$-finite measures on $(\mathbb{A},\mathcal{A})$, then $f$ is unique $\mu$-ae and is denoted with $f_{\mu'}$.
\end{lemma}
\begin{proof}
  See, e.g., Papadimitrakis~\cite[Exercise 4.61 and Theorem 6.10]{mcdonald2012course}.
\end{proof}

\subsection{Preliminaries: Convergence}\label{sec:prelimconvergence}
In this section, we recall the definition of pointwise convergence of sequences of functions.
Convergence is used to define correctness in Section~\ref{sec:smccorrect} and Section~\ref{sec:smcpplcorrect}.
For a more comprehensive introduction to convergence, we recommend~McDonald and Weiss~\cite{mcdonald2012course}.

\begin{definition}
  Let ${\{x_n\}}_{n\in\mathbb{N}}$ be a sequence of real numbers, and $x$ a real number.
  We say that $\lim_{n \to \infty} x_n = x$ if for all $\varepsilon > 0$, there exists an $N$ such that $|x_n-x| < \varepsilon$ for all $n > N$.
\end{definition}

\begin{definition}\label{def:convergence}
  Let ${\{f_n : \mathbb{A} \to \R\}}_{n\in\mathbb{N}}$ be a sequence of functions, and $f : \mathbb{A} \to \R$ a function.
  We say that $\lim_{n \to \infty} f_n = f$ pointwise if for all $x \in \mathbb{A}$,
  it holds that $\lim_{n \to \infty} f_n(x) = f(x)$.
\end{definition}
In particular, we say that $\lim_{n \to \infty} f_n = f$ $\mu$-ae if the sequence $f_n$ converges pointwise to $f$, except on a set of $\mu$-measure 0.

\subsection{Preliminaries: Metric Spaces}\label{sec:prelimmetric}
\begin{definition}
  Given a space $\mathbb{M}$, a function $d : \mathbb{M} \times \mathbb{M} \to \R_+^*$ is called a \emph{metric on $\mathbb{M}$} if for all $m,m' \in
  \mathbb{M}$
  \begin{enumerate}
    \item $d(m,m') = 0 \iff m = m'$
    \item $d(m,m') = d(m',m)$
    \item $d(m,m'') \leq d(m,m') + d(m',m'')$
  \end{enumerate}
  Furthermore, we call $(\mathbb{M},d)$ a \emph{metric space} if $d : \mathbb{M} \times \mathbb{M} \to \R_+^*$ is a metric on $\mathbb{M}$.
\end{definition}

\begin{definition}
  For $n \in \mathbb{N}$, we let
  \begin{equation}
    \drn{n}((x_1,x_2,\ldots,x_n),(y_1,y_2,\ldots,y_n)) =
    |x_1-y_1| + |x_2-y_2| + \cdots + |x_n-y_n|,
  \end{equation}
  and $\dr = \drn{1}$.
  It is easy to verify that $\drn{n}$ is a metric for each $n$.
\end{definition}

\begin{definition}
  Let $d : \mathbb{M} \times \mathbb{M} \to \R_+^*$ be a metric, and let $M \subset \mathbb{M}$.
  We denote by $d|_M : M \times M \to \R_+^*$ the \emph{restriction} of $d$ to $M$.
  It is easy to verify that $d|_M$ is in fact a metric.
\end{definition}


\begin{definition}
  We let $\drpos = \dr|_{\mathbb{R}_+}$ and $d_{[0,1]} = \dr|_{[0,1]}$.
\end{definition}

\begin{definition}
  A subset $M \subset \mathbb{M}$ of a metric space is called \emph{dense} if for all $m \in \mathbb{M}$ and $\varepsilon > 0$, there exists $m' \in M$ such that $d(m,m') < \varepsilon$.
\end{definition}

\begin{definition}
  A metric space is called \emph{separable} if it has a countable dense subset.
\end{definition}

\begin{lemma}\label{lemma:rsep}
  For every $n \in \mathbb{N}$, $(\R^n,\drn{n})$ is a separable metric space.
\end{lemma}
\begin{proof}
  Standard result which follows from the fact that $\mathbb{Q}^n$ is a countable dense subset.
\end{proof}

\begin{definition}
  The \emph{Borel} $\sigma$-algebra induced by a separable metric space $(\mathbb{M},d)$ is defined as
  \begin{equation}
    \Bor_{(\mathbb{M},d)}
    = \sigma(\{B_{(\mathbb{M},d)}(r,m) \mid m \in \mathbb{M}, r \in \R_+ \})
  \end{equation}
  where $B_{(\mathbb{M},d)}(r,m) = \{ m' \in \mathbb{M} \mid d(m,m') < r \}$ is known as the \emph{ball of radius $r$ centered at $m$ (in $\mathbb{M}$)}.
  We call the tuple $(\mathbb{M},\Bor_{(\mathbb{M},d)})$ the \emph{Borel space} corresponding to $(\mathbb{M},d)$.
\end{definition}

\begin{lemma}\label{lemma:openmetric}
  $\Bor_{(\R^n,d_{\R^n})} = \Bor^n$,
\end{lemma}
\begin{proof}
  Standard result in topology.
\end{proof}

\begin{lemma}\label{lemma:genrest}
  Let $\mathbf{A} \subset \mathcal{P}(\mathbb{A})$ and $\emptyset \neq A \subset \mathbb{A}$.
  Then $\sigma(\mathbf{A}|_A) = \sigma(\mathbf{A})|_A$.
\end{lemma}
\begin{proof}
  See Papadimitrakis~\cite[Theorem 1.2]{papadimitrakis2004measure}.
\end{proof}

\begin{lemma}\label{lemma:metricrest}
  Let $(\mathbb{M},d)$ be a separable metric space, $(\mathbb{M},\Bor_{(\mathbb{M},d)})$ the corresponding Borel space, and $\emptyset \neq M \in \Bor_{(\mathbb{M},d)}$.
  Then $(M,d|_M)$ is a separable metric space, and $\Bor_{(M,d|_M)} = \Bor_{(\mathbb{M},d)}|_M$.
\end{lemma}
\begin{proof}
%
  That $(M,d|_M)$ is separable is a standard result in topology.
  The second part follows from Munkres~\cite[Exercise 1, p. 133]{munkres2000topology} together with Lemma~\ref{lemma:genrest}. We will not give the details here, since this requires more definitions and results from topology.
\end{proof}

\begin{lemma}
  $\Bor_{(\R^n_+,d_{\R_+})} = \Bor^n_+$ and
  $\Bor_{({[0,1]}^n,d_{[0,1]})} = \Bor^n_{[0,1]}$.
\end{lemma}
\begin{proof}
  By Lemma~\ref{lemma:openmetric} and Lemma~\ref{lemma:metricrest}.
\end{proof}

\subsection{The Measure Space $(\mathbb{S},\mathcal{S},\mu_{\mathbb{S}})$}%
\label{seca:tracespace}

\repeatlemma{lemma:tracemeas}{\lemmatracemeas}
\begin{proof}
  We have to show that $\Trcal$ is a $\sigma$-algebra:
  \begin{enumerate}
    \item Clearly, $\Trcal \neq \emptyset$.

    \item $S \in \Trcal \implies S^c \in \Trcal$.
      \begin{equation}
        S^c
        = {\left(\bigcup_{n\in\mathbb{N}_0} B_n\right)}^c
        = \bigcup_{n\in\mathbb{N}_0} B_n^c
      \end{equation}
      Since $B_n^c \in \Bor^n_{[0,1]}$, the implication holds.

    \item
      ${\{ S_i \}}_i \subset \Trcal \implies \bigcup_i S_i \in \Trcal$.
      \begin{equation}
        \bigcup_i S_i
        = \bigcup_i \bigcup_{n\in\mathbb{N}_0} B_{n,i}
        = \bigcup_{n\in\mathbb{N}_0} \bigcup_i B_{n,i}
      \end{equation}
      Since $\bigcup_i B_{n,i} \in \Bor^n_{[0,1]}$, the implication holds.

  \end{enumerate}
\end{proof}

\repeatlemma{lemma:tracemeasspace}{\lemmatracemeasspace}
\begin{proof}
  We begin by showing that $\mu_\Tr$ is a measure.
  \begin{enumerate}
    \item $\mu_\Tr(S) \geq 0, S \in \Trcal$.
      Follows since
      \begin{equation}
        \mu_\Tr(S) = \mu_\Tr\left(\bigcup_{n\in\mathbb{N}_0} B_n\right)
        = \sum_{n\in\mathbb{N}_0} \lambda_n(B_n)
        \geq \sum_{n\in\mathbb{N}_0} 0 = 0.
      \end{equation}
    \item $\mu_\Tr(\emptyset) = 0$.
      Follows since
      \begin{equation}
        \mu_\Tr(\emptyset) =
        \mu_\Tr\left(\bigcup_{n\in\mathbb{N}_0} \emptyset\right)
        = \sum_{n\in\mathbb{N}_0} \lambda_n(\emptyset) = 0.
      \end{equation}
    \item If ${\{ S_i \}}_i \subset \Trcal$ with $S_j \cap S_k = \emptyset$
      for $j \neq k$, then $\mu_\Tr\left(\bigcup_i S_i\right) = \sum_i
      \mu_\Tr(S_i)$
      \begin{equation}\begin{aligned}
        \mu_\Tr\left(\bigcup_i S_i\right)
        &= \mu_\Tr\left(\bigcup_i \bigcup_{n\in\mathbb{N}_0} B_{i,n}\right) \\
        &= \mu_\Tr\left(\bigcup_{n\in\mathbb{N}_0} \bigcup_i B_{i,n}\right) \\
        &= \sum_{n\in\mathbb{N}_0} \lambda_n\left( \bigcup_i B_{i,n} \right)
        & &
        \text{%
          (By definition of $\mu_\Tr$)
        } \\
        &= \sum_{n\in\mathbb{N}_0} \sum_i \lambda_n\left( B_{i,n} \right)
        & &
        \text{%
          (The $B_{i,n}$ are disjoint)
        } \\
        &= \sum_i \sum_{n\in\mathbb{N}_0} \lambda_n\left( B_{i,n} \right) \\
        &= \sum_i \mu_\Tr(S_i)
      \end{aligned}\end{equation}
  \end{enumerate}
  Next, we need to show that $\mu_\Tr$ is $\sigma$-finite.
  To do this, we show that there is a sequence ${\{S_i\}}_i \subset \Trcal$, $\mu_\Tr(S_i) < \infty$ for all $i$, such that $\bigcup_i S_i = \Tr$.
  We can choose these $S_i$ simply as $S_i = {[0,1]}^i$, $i \in \mathbb{N}_0$.
  Trivially, $\bigcup_i S_i = \Tr$, and $\mu_\Tr(S_i) = \lambda_n({[0,1]}^i) = 1 < \infty$.
\end{proof}

We now define a metric on $\Tr$.
\begin{definition}
  Let $c_i$ and $c'_i$ denote the $i$th element of $s \in \Tr$ and $s' \in \Tr$, respectively.
  \begin{equation}
    d_\Tr(s,s') =
    \begin{cases}
      \sum_{i=1}^{|s|} |c_i-c'_i| & \text{if } |s| = |s'| \\
      \infty & \text{otherwise}
    \end{cases}
  \end{equation}
\end{definition}

\begin{lemma}\label{lemma:tracesep}
  $(\Tr,d_\Tr)$ is a separable metric space.
\end{lemma}
\begin{proof}
  Consider $\Tr_\mathbb{Q} = \bigcup_{n \in \mathbb{N}_0} {([0,1] \cap \mathbb{Q})}^n$.
  It is easy to verify that $\Tr_\mathbb{Q}$ is a countable dense subset of $\Tr$, from which the result follows.
\end{proof}

\begin{lemma}\label{lemma:traceinducedmetric}
  $\Bor_{(\Tr,d_\Tr)} = \Trcal$.
\end{lemma}
\begin{proof}
  Informally, this follows since $\Tr$ is the union of a countable set of isolated subspaces (the distance from each element in a subset to all elements of other subsets is $\infty$) which are all isomorphic to $\mathbb{R}^n$, for some $n \in \mathbb{N}_0$.

  More formally, note that $\mathcal{S} = \sigma\left(\bigcup_{n\in\mathbb{N}_0}
  \Bor_{[0,1]}^n\right)$.
  Clearly, by definition,
  \begin{equation}
    \bigcup_{n\in\mathbb{N}_0} \Bor_{[0,1]}^n \subset \Bor_{(\Tr,d_\Tr)}.
  \end{equation}
  Hence,
  \begin{equation}
    \mathcal{S} = \sigma\left(\bigcup_{n\in\mathbb{N}_0}
    \Bor_{[0,1]}^n\right) \subset \sigma(\Bor_{(\Tr,d_\Tr)}) =
    \Bor_{(\Tr,d_\Tr)}.
  \end{equation}

  Next, because the distance between traces of different length is $\infty$, we note that
  \begin{equation}
    \Bor_{(\Tr,d_\Tr)} =
    \sigma(\{B_{(\Tr,d_\Tr)}(r,s) \mid s \in \Tr, r \in \R_+ \})
    \subset
    \sigma\left(\bigcup_{n\in\mathbb{N}_0} \Bor_{[0,1]}^n\right)
    = \mathcal{S}.
  \end{equation}
  The result follows.
\end{proof}

\subsection{The Measurable Space $(\mathbb{T},\mathcal{T})$}%
\label{seca:termspace}

\begin{lemma}\label{lemma:placeholderbijection}
  The $\term_p^n$ are bijections.
\end{lemma}
\begin{proof}
  It is easy to verify that $\term_p^n$ is both one-to-one and onto.
\end{proof}

\begin{lemma}
  $\Tcal_{\term_p^n}$ is a $\sigma$-algebra.
\end{lemma}
\begin{proof}
  Follows from Lemma~\ref{lemma:placeholderbijection}, and because $\mathcal{B}^n$ is a $\sigma$-algebra for each $n \in \mathbb{N}_0$.
\end{proof}

\repeatlemma{lemma:termmeas}{\lemmatermmeas}
\begin{proof}
  We have to show that $\Tcal$ is a $\sigma$-algebra.
  \begin{enumerate}
    \item Clearly, $\mathcal{T} \neq \emptyset$.

    \item $T \in \Tcal \implies T^c \in \Tcal$.
      \begin{equation}
        \begin{aligned}
          T^c
          &=
          {%
            \left(\bigcup_{n \in \mathbb{N}_0}
            \bigcup_{\term_p^n \in \T^n_p}
            \term_p^n(B_n)
            \right)
          }^c
          \\ &=
          \bigcup_{n \in \mathbb{N}_0}
          \bigcup_{\term_p^n \in \T^n_p}
          \term_p^n{(B_n)}^c
        \end{aligned}
      \end{equation}
      Since $\term_p^n{(B_n)}^c \in \Tcal_{\term_p^n}$, the implication holds.

    \item
      ${\{ T_i \}}_i \subset \Tcal \implies \bigcup_i T_i \in \Tcal$.
      \begin{equation}
        \bigcup_i T_i
        = \bigcup_i \bigcup_{n \in \mathbb{N}_0}
          \bigcup_{\term_p^n \in \T^n_p}
          \term_p^n(B_{n,i})
        = \bigcup_{n \in \mathbb{N}_0}
          \bigcup_{\term_p^n \in \T^n_p}
          \bigcup_i
          \term_p^n(B_{n,i})
      \end{equation}
      Since $\bigcup_i \term_p^n(B_{n,i}) \in \Tcal_{\term_p^n}$, the
      implication holds.
  \end{enumerate}
\end{proof}

\begin{definition}
  \begin{equation}
    \begin{aligned}
      \dt(c_1,c_2) &= |c_1 - c_2| \\
      \dt(\lambda x. \term,\lambda x. \term') &= \dt(\term,\term') \\
      \dt(x,x) &= 0 \\
      \dt(\term_1 \enspace \term_2,\term_1' \enspace \term_2') &=
      \dt(\term_1,\term_2) + \dt(\term_1',\term_2') \\
      \dt(\ttt{if} \ \term_1 \ \ttt{then} \ \term_2 \ \ttt{else} \ \term_3,
      \quad & \\
      \ttt{if} \ \term_1' \ \ttt{then} \ \term_2' \ \ttt{else} \ \term_3')
      &= \dt(\term_1,\term_1') + \dt(\term_2,\term_2') + \dt(\term_3,\term_3') \\
      \dt(g(\term_1, \ldots, \term_{|g|}),G(\term_1', \ldots, \term_{|g|}' )) &=
      \dt(\term_1,\term_1') + \cdots + \dt(\term_{|g|},\term_{|g|}') \\
      \dt(\ttt{sample}_D(\term_1, \ldots, \term_{|D|} ),
      \quad & \\
      \ttt{sample}_D(\term_1', \ldots, \term_{|D|}' ))
      &= \dt(\term_1,\term_1') + \cdots + \dt(\term_{|D|},\term_{|D|}') \\
      \dt(\ttt{weight}({\term}),\ttt{weight}({\term'}))
      &= \dt(\term,\term') \\
      \dt(\term,\term') &= \infty \ \text{if none of the above applies}
    \end{aligned}
  \end{equation}
\end{definition}
Note that distances between terms are taken modulo $\alpha$-conversion.
Hence,
\begin{equation}
  \dt(\lambda x.x,\lambda y.y) = 0,
\end{equation} but
\begin{equation}
  \dt(\lambda x.\lambda y.x, \lambda x.\lambda y.y) = \infty
\end{equation}

\begin{lemma}\label{lemma:termsep}
  $(\T,d_\T)$ is a separable metric space.
\end{lemma}
\begin{proof}
  Replace $\T$ with a language where constants are rational.
  This is a countable dense subset of $\T$, and the result follows.
\end{proof}

\begin{lemma}\label{lemma:metricequalterm}
  $\Bor_{(\T,\dt)} = \Tcal$.
\end{lemma}
\begin{proof}
  As for Lemma~\ref{lemma:traceinducedmetric}, the result follows since $\T$ is the union of a countable set of isolated subsets  which are all isomorphic to $\mathbb{R}^n$, for some $n \in \mathbb{N}_0$.
\end{proof}

We now extend the above metric to evaluation contexts.
\begin{definition}
  \begin{equation}
    \begin{aligned}
      \de([\cdot],[\cdot]) &= 0 \\
      \de(\econt_1 \enspace \term_1,\econt_2 \enspace \term_2) &=
      \de(\econt_1,\econt_2) + \de(\term_1, \term_2) \\
      \de((\lambda x.\term_1) \enspace \econt_1,
      (\lambda x.\term_2) \enspace \econt_2) &=
      \de(\term_1,\term_2) + \de(\econt_1, \econt_2) \\
      \de(\ttt{if} \ \econt \ \ttt{then} \ \term_1 \ \ttt{else} \ \term_2,
      \enspace & \\
      \ttt{if} \ \econt' \ \ttt{then} \ \term_1' \ \ttt{else} \ \term_2') &=
      \de(\econt_1, \econt')
      + \de(\term_1,\term_1')
      \\ & \qquad
      + \de(\term_2,\term_2') \\
      \de(g(c_1, \ldots, c_m, \econt, \term_{m+2}, \ldots, \term_{|g|}),
      \enspace & \\
      g(c_1', \ldots, c_m', \econt', \term_{m+2}', \ldots, \term_{|g|}')) &=
      \de(c_1,c_1') + \cdots + \de(c_{m},c_{m}') \\ & \qquad
      + \de(\econt,\econt')
      \\ & \qquad
      + \de(\term_{m+2},\term_{m+2}')
      \\ & \qquad
      + \cdots + \de(\term_{|g|},\term_{|g|}')
      \\
      \de(\ttt{sample}_D
      (c_1, \ldots, c_m, \econt, \term_{m+2}, \ldots, \term_{|D|}),
      \enspace & \\
      \ttt{sample}_D
      (c_1', \ldots, c_m', \econt', \term_{m+2}', \ldots, \term_{|D|}')) &=
      \de(c_1,c_1') + \cdots + \de(c_{m},c_{m}') \\ & \qquad
      + \de(\econt,\econt')
      \\ & \qquad
      + \de(\term_{m+2},\term_{m+2}')
      \\ & \qquad
      + \cdots + \de(\term_{|g|},\term_{|g|}')
      \\
      \de(\ttt{weight}(\econt_1),\ttt{weight}(\econt_2)) &= \de(\econt_1,\econt_2)
      \\
      \de(\econt,\econt') &= \infty \ \text{if none of the above applies}
    \end{aligned}
  \end{equation}
\end{definition}

\begin{definition}
  The set of \emph{reducible expressions} is given by
  \begin{equation}
    \begin{aligned}
      \redex \Coloneqq& \enspace
      (\lambda x. \term) \enspace \val
      \enspace | \enspace
      g(c_1, \ldots, c_{|g|})
      \enspace | \enspace
      \ttt{sample}_D(c_1, \ldots, c_{|D|} )
      \\& \enspace | \enspace
      \ttt{if} \ \mathit{true} \ \ttt{then} \ \term_1 \ \ttt{else} \ \term_2
      \enspace | \enspace
      \ttt{if} \ \mathit{false} \ \ttt{then} \ \term_1 \ \ttt{else} \ \term_2
      \\& \enspace | \enspace
      \ttt{weight}(c)
      \enspace | \enspace
      \ttt{resample}
    \end{aligned}
  \end{equation}
\end{definition}

\begin{lemma}\label{lemma:econtleq}
  $\dt(\econt_1[\term_1],\econt_2[\term_2]) \leq \de(\econt_1,\econt_2) + \dt(\term_1,\term_2)$.
\end{lemma}
\begin{proof}
  Follows by induction over the structure of $\econt_1$ and $\econt_2$.
  For a partial proof, see Borgström et al.~\cite[Lemma 63]{borgstrom2015lambda}.
\end{proof}

\begin{lemma}\label{lemma:econtinfty}
  If $\de(\econt_1,\econt_2) = \infty$, then $\dt(\econt_1[\redex_1],\econt_2[\redex_2]) = \infty$ for all $\redex_1$ and $\redex_2$.
\end{lemma}
\begin{proof}
  Follows by induction over the structure of $\econt_1$.
  For a partial proof, see Borgström et al.~\cite[Lemma 64]{borgstrom2015lambda}.
\end{proof}

\begin{lemma}\label{lemma:econtredex}
  $\dt(\econt_1[\redex_1],\econt_2[\redex_2]) = \de(\econt_1,\econt_2) + \dt(\redex_1,\redex_2)$.
\end{lemma}
\begin{proof}
  Follows by induction over the structure of $\econt_1$ and $\econt_2$.
  The proof uses Lemma~\ref{lemma:econtinfty} and is otherwise similar to Lemma~\ref{lemma:econtleq}.
\end{proof}

\begin{lemma}\label{lemma:econtval}
  $\dt(\econt_1[\val_1],\econt_2[\val_2]) = \de(\econt_1,\econt_2) + \dt(\val_1,\val_2)$.
\end{lemma}
\begin{proof}
  Analogous to the proof for Lemma~\ref{lemma:econtredex}.
\end{proof}

\begin{lemma}\label{lemma:econtsubst}
  $\dt([x \mapsto \val_1]\term_1,[x \mapsto \val_2]\term_2) \leq \dt(\term_1,\term_2) + k \cdot \dt(\val_1,\val_2)$, where $k$ is the maximum of the number of occurrences of $x$ in $\term_1$ and $\term_2$.
\end{lemma}
\begin{proof}
  The result follows immediately if $\dt(\term_1,\term_2) = \infty$.
  Therefore, assume $\dt(\term_1,\term_2) < \infty$.
  We now proceed by induction over the structure of $\term_1$ and $\term_2$.
  \begin{itemize}
    \item
      \textbf{Case %
        $\term_1 = c_1$,
        $\term_2 = c_2$.
      }
      We have
      \begin{equation}
        \dt([x \mapsto \val_1]c_1,[x \mapsto \val_2]c_2) = \dt(c_1,c_2).
      \end{equation}
      The result follows immediately.

    \item
      \textbf{Case %
        $\term_1 = (\lambda x'. \term_1')$,
        $\term_2 = (\lambda x'. \term_2')$.
      }
      By using the induction hypothesis, we have
      \begin{equation}
        \begin{aligned}
          &\dt([x \mapsto \val_1](\lambda x'. \term_1'),
          [x \mapsto \val_2](\lambda x'. \term_2'))
          \\ &= \dt([x \mapsto \val_1]\term_1',
          [x \mapsto \val_2]\term_2') \\ &\leq
          \dt(\term_1', \term_2') + k\cdot \dt(\val_1,\val_2) \\ &=
          \dt(\lambda x'. \term_1', \lambda x'. \term_2') + k\cdot
          \dt(\val_1,\val_2)
        \end{aligned}
      \end{equation}
      because the number of occurrences $k$ of $x$ are the same in $(\lambda x'. \term)$ and $\term$.

    \item
      \textbf{Case %
        $\term_1 = x'$,
        $\term_2 = x'$.
      }
      In this case, we have two subcases: either $x = x'$ or $x \neq x'$.
      In the case $x = x'$,
      \begin{equation}
        \begin{aligned}
          \dt([x \mapsto \val_1]x', [x \mapsto \val_2]x') &=
          \dt(\val_1,\val_2)
        \end{aligned}
      \end{equation}
      and the result follows immediately ($k = 1$).
      In the case $x \neq x'$,
      \begin{equation}
        \begin{aligned}
          \dt([x \mapsto \val_1]x', [x \mapsto \val_2]x') &=
          \dt(x',x') = 0.
        \end{aligned}
      \end{equation}
      and the result follows immediately ($k = 0$).

    \item
      \textbf{Case %
        $\term_1 = \term_1' \enspace \term_2'$,
        $\term_2 = \term_1'' \enspace \term_2''$.
      }
      By using the induction hypothesis, we have
      \begin{equation}
        \begin{aligned}
          \dt([x \mapsto \val_1](\term_1' \enspace \term_2'),
          \quad & \\
          [x \mapsto \val_2](\term_1'' \enspace \term_2''))
          &=
          \dt([x \mapsto \val_1]\term_1' \enspace [x \mapsto \val_1]\term_2',
          \\ & \hspace{11mm}
          [x \mapsto \val_2]\term_1'' \enspace [x \mapsto \val_2]\term_2'')
          \\ &=
          \dt([x \mapsto \val_1]\term_1', [x \mapsto \val_2]\term_1'')
          \\ & \hspace{11mm}
          + \dt([x \mapsto \val_1]\term_2',[x \mapsto \val_2]\term_2'')
          \\ &\leq
          \dt(\term_1',\term_1'') + k_1 \cdot \dt(\val_1,\val_2)
          \\ & \hspace{11mm}
          + \dt(\term_2',\term_2'') + k_2 \cdot \dt(\val_1,\val_2)
        \end{aligned}
      \end{equation}
      where $k_1 + k_2 = k$, by definition.
      Now
      \begin{equation}
        \begin{aligned}
          \dt(\term_1',\term_1'')
          + k_1 \cdot \dt(\val_1,\val_2)
          + \dt(\term_2',\term_2'')
          & \\
          + k_2 \cdot \dt(\val_1,\val_2)
          &=
          \dt(\term_1',\term_1'') + \dt(\term_2',\term_2'')
          \\ & \hspace{2cm}
          + k \cdot \dt(\val_1,\val_2)
          \\ &=
          \dt(\term_1,\term_2) + k \cdot \dt(\val_1,\val_2).
        \end{aligned}
      \end{equation}

    \item
      The remaining cases follow by largely similar arguments.
  \end{itemize}
\end{proof}

\subsection{Preliminaries: Measure Theory and Metric Spaces, Continued}%
\label{sec:prelimscont}

\begin{definition}
  Given a finite set of measurable spaces ${\{(\mathbb{A}_i,\mathcal{A}_i)\}}_{i=1}^n$, we define the \emph{product $\sigma$-algebra} on $\bigtimes_{i=1}^n \mathbb{A}_i$ as
  \begin{equation}
    \bigotimes_{i=1}^n \mathcal{A}_i =
    \sigma\left(\left\{\bigtimes_{i=1}^n A_i \mid A_i \in
    \mathcal{A}_i\right\}\right).
  \end{equation}
  where $\bigtimes$ denotes the usual Cartesian product of sets.
\end{definition}

\begin{lemma}\label{lemma:decompmeas}
Let
\begin{itemize}
  \item
    $(\mathbb{A},\mathcal{A})$ and $(\mathbb{A}',\mathcal{A}')$ be measurable
    spaces,
  \item
    $f : \mathbb{A} \to \mathbb{A}'$ be a function,
  \item
    ${\{A_i\}}_i$ be a countable set of elements of $\mathcal{A}$
    such that $\bigcup_i A_i = \mathbb{A}$,
  \item
    $\mathcal{A}_i = \{ A \in
    \mathcal{A} \mid A \subset A_i \}$, and
  \item
    $f_i = f|_{A_i}$ be the restriction of $f$ to $A_i$.
\end{itemize}
Then $f : (\mathbb{A},\mathcal{A}) \to (\mathbb{A}',\mathcal{A}')$ is measurable iff $f_i : (A_i,\mathcal{A}_i) \to (\mathbb{A}',\mathcal{A}')$ is measurable for each $i$.
\end{lemma}
\begin{proof}
  See Billingsley~\cite[Problem 13.1]{billingsley1995probability}.
\end{proof}

\begin{definition}
  Given a finite set of metric spaces ${\{(\mathbb{M}_i,d_i)\}}_{i=1}^n$, we can form the \emph{product metric space}
  \begin{equation}
    \left(\bigtimes_{i=1}^n \mathbb{M}_i, \sum_{i=1}^n d_i\right),
  \end{equation}
  where $\sum_{i=1}^n d_i$ is the Manhattan metric formed from the component metrics $d_i$.
\end{definition}

\begin{lemma}\label{lemma:boreqmeas}
  Let ${\{(\mathbb{M}_i,d_i)\}}_{i=1}^n$ be a set of separable metric spaces.
  Then
  \begin{equation}
    \left(\bigtimes_{i=1}^n \mathbb{M}_i, \sum_{i=1}^n d_i\right)
  \end{equation}
  is a separable metric space, and
  \begin{equation}
    \Bor_{\left(\bigtimes_{i=1}^n \mathbb{M}_i, \sum_{i=1}^n d_i\right)}
    =
    \bigotimes_{i=1}^n \Bor_{(\mathbb{M}_i,d_i)}.
  \end{equation}
\end{lemma}
\begin{proof}
  See Kallenberg~\cite[Lemma 1.2]{kallenberg2002foundations}.
\end{proof}

\begin{definition}\label{definition:cont}
  Given metric spaces $({\mathbb{M}_1},d_{\mathbb{M}_1})$ and $({\mathbb{M}_2},d_{\mathbb{M}_2})$, a function $f : {\mathbb{M}_1} \to {\mathbb{M}_2}$ is \emph{continuous} if for each $m \in {\mathbb{M}_1}$ and $\varepsilon > 0$, there is a $\delta > 0$ such that for all $m' \in \mathbb{M}_1$, $d_{\mathbb{M}_1}(m,m') < \delta \implies d_{\mathbb{M}_2}(f(m),f(m')) < \varepsilon$.
\end{definition}

\begin{lemma}\label{lemma:contmeas}
  If a function $f : \mathbb{M}_1 \to \mathbb{M}_2$ between metric spaces $(\mathbb{M}_1,d_{\mathbb{M}_1})$ and $(\mathbb{M}_2,d_{\mathbb{M}_2})$ is continuous, it is measurable with respect to the induced Borel $\sigma$-algebras $\Bor_{(\mathbb{M}_1,d_{\mathbb{M}_1})}$ and $\Bor_{(\mathbb{M}_2,d_{\mathbb{M}_2})}$.
\end{lemma}
\begin{proof}
  See Kallenberg~\cite[Lemma 1.5]{kallenberg2002foundations}.
\end{proof}

\begin{lemma}\label{lemma:measurablegenerate}
  Let $\mathbf{A} \subset \mathcal{P}(\mathbb{A})$.
  Furthermore, let $(\mathbb{A}',\mathcal{A}')$, and $(\mathbb{A}, \sigma(\mathbf{A}))$ be measurable spaces.
  Then $f : (\mathbb{A}',\mathcal{A}') \to (\mathbb{A}, \sigma(\mathbf{A}))$ is measurable iff $f^{-1}(A) \in \mathcal{A}'$ for each $A \in \mathbf{A}$.
\end{lemma}
\begin{proof}
  The ``only if'' part is trivial.
  We now show the ``if'' part.
  Consider the set $\mathbf{B} = \{ A \in \mathcal{P}(\mathbb{A}) \mid f^{-1}(A) \in \mathcal{A}'\}$.
  Obviously, $\mathbf{A} \subset \mathbf{B}$.
  Furthermore, from properties of the preimage, it is easy to check that $\mathbf{B}$ is a $\sigma$-algebra.
  Therefore, $\sigma(\mathbf{A}) \subset \mathbf{B}$, and $f^{-1}(A) \in \mathcal{A'}$ for each $A \in \sigma(\mathbf{A})$.
  Hence, $f$ is measurable.
\end{proof}

\begin{lemma}\label{lemma:compmeasurable}
  Let
  \begin{equation}
    {\{
        f_i : (\mathbb{A},\mathcal{A})
        \to (\mathbb{A}_i, \mathcal{A}_i)
    \}}_{i=1}^n
  \end{equation}
  be a finite set of measurable functions.
  Then
  \begin{equation}
    f = (f_1,f_2,\ldots,f_n) : (\mathbb{A},\mathcal{A}) \to
    \left(\bigtimes_{i=1}^n \mathbb{A}_i,
      \bigotimes_{i=1}^n \mathcal{A}_i \right)
  \end{equation}
  is measurable.
\end{lemma}
\begin{proof}
  By Lemma~\ref{lemma:measurablegenerate}, it suffices to check that $f^{-1}(A_\times) \in \mathcal{A}$, for all $A_\times \in \{\bigtimes_{i=1}^n A_i \mid A_i \in \mathcal{A}_i\}$.
  Hence, for all $A_\times = \bigtimes_{i=1}^n A_i$, by properties of the preimage and the measurability of the $f_i$,
  \begin{equation}
    f^{-1}(A_\times) = f^{-1}\left(\bigtimes_{i=1}^n A_i\right)
    = \bigcap_{i=1}^n f_i^{-1}(A_i)
    \in \mathcal{A}.
  \end{equation}
  The result follows.
\end{proof}

\subsection{The Big-Step Function Induced by a Small-Step Relation.}%
\label{sec:functioninduce}
Assume there is a small-step relation $\rightarrow$ which can be regarded as a measurable function \begin{equation}\label{eq:startmulti} \rightarrow : (A',\mathcal{A}|_{A'}) \to (\mathbb{A},\mathcal{A}), \end{equation} with $A' \in \mathcal{A}$.
We complete this function, forming the function $\mathit{step}_\rightarrow : \mathbb{A} \to \mathbb{A}$.
\begin{definition}
  $
    \mathit{step}_\rightarrow = \ \rightarrow \cup \ \id|_{\mathbb{A} \setminus A'}.
  $
\end{definition}
\begin{lemma}\label{lemma:stepmeasurable}
  $\mathit{step}_\rightarrow : (\mathbb{A},\mathcal{A}) \to
  (\mathbb{A},\mathcal{A})$ is measurable.
\end{lemma}
\begin{proof}
  It holds that $A = (A \cap A') \cup (A \cap A'^c)$ for any $A \in \mathcal{A}$.
  Hence,
  \begin{equation}
    \begin{aligned}
      \mathit{step}_\rightarrow^{-1}(A) &=
      \mathit{step}_\rightarrow^{-1}((A \cap A') \cup (A \cap A'^c)) \\
      &=
      \mathit{step}_\rightarrow^{-1}(A \cap A') \cup
      \mathit{step}_\rightarrow^{-1}(A \cap A'^c) \\
      &=
      \ \rightarrow^{-1}\!\!(A \cap A') \cup
      \id^{-1}(A \cap A'^c) \\
    \end{aligned}
  \end{equation}
  Because $\rightarrow$ and $\id$ are measurable, we have $\mathit{step}_\rightarrow^{-1}(A) \in \mathcal{A}$, as required.
\end{proof}
In the following, we use the notation
\begin{equation}
  \mathit{step}_\rightarrow^n =
  \underbrace{%
    \mathit{step}_\rightarrow
    \circ \mathit{step}_\rightarrow \circ
    \ldots
    \circ \mathit{step}_\rightarrow
  }_\text{$n$ times}
\end{equation}
with $n \in \mathbb{N}_0$.
Next, assume that we have a measurable function $\mathit{extract} : (\mathbb{A},\mathcal{A}) \to (\mathbb{H},\mathcal{H})$.
We require that $\mathbb{H}$ has a bottom element $\bot$ (such that $\{\bot\} \in \mathcal{H}$) and that $\mathbb{H}$ is equipped with a flat partial order $\leq_\mathbb{H}$ (i.e., the smallest partial order with $\bot \leq_\mathbb{H} h$ for all $h \in \mathbb{H}$).
Furthermore, we require that $\mathit{extract}$ has the following property with respect to the function $\mathit{step}_\rightarrow$.
\begin{condition}\label{cond:extractprop}
  $
    \mathit{extract}(a)
    \leq_\mathbb{H} \mathit{extract}(\mathit{step}_\rightarrow(a))
  $
\end{condition}
\begin{lemma}\label{lemma:extractlemma}
  $\mathit{extract}(a) \neq \bot \implies
  \forall n>0. \ \mathit{extract}(a) =
  \mathit{extract}(\mathit{step}_\rightarrow^n(a))$
\end{lemma}
\begin{proof}
If $\mathit{extract}(a) \neq \bot$, then $\mathit{extract}(a) \leq_\mathbb{H} h$ implies $\mathit{extract}(a) = h$.
From this and by Condition~\ref{cond:extractprop}, we must have $\mathit{extract}(a) = \mathit{extract}(\mathit{step}_\rightarrow(a))$.
The result now follows by induction.
\end{proof}

Now, we make the following definition.
\begin{definition}
  $
    \mathit{final}_{\rightarrow,\mathit{extract}} = \sup
    \{\mathit{extract} \circ \mathit{step}_\rightarrow^n \mid n \in \mathbb{N}_0 \}
  $
  where the supremum is taken with respect to the pointwise order on functions induced by $\leq_\mathbb{H}$.
\end{definition}

\begin{lemma}
  $\mathit{final}_{\rightarrow,\mathit{extract}}$ exists.
\end{lemma}
\begin{proof}
  Take an arbitrary $a \in \mathbb{A}$.
  From~\eqref{cond:extractprop} and Lemma~\ref{lemma:extractlemma}, we must have
  \begin{equation}
    \exists n \in \mathbb{N}. \ \forall m \geq n. \
    \mathit{extract} \circ \mathit{step}_\rightarrow^n(a)
    = \mathit{extract} \circ \mathit{step}_\rightarrow^m(a) = h(a).
  \end{equation}
  The function mapping $a$ to $h(a)$ for all $a$ is the only upper bound of the set.
  Hence, the supremum necessarily exists.
\end{proof}

\begin{lemma}
  $\mathit{final}_{\rightarrow,\mathit{extract}}
  : (\mathbb{A},\mathcal{A}) \to (\mathbb{H},\mathcal{H})$ is measurable.
\end{lemma}
\begin{proof}
  This proof is based on Borgström et al.~\cite[Lemma 89]{borgstrom2015lambda}.
  Let $f_n = \mathit{extract} \circ \mathit{step}^n_\rightarrow$.
  The function $f_n$ is clearly measurable, since it is a composition of measurable functions ($\mathit{step}^n_\rightarrow$ is measurable as a consequence of Lemma~\ref{lemma:stepmeasurable}).
  Next, let $\sup f_n = \mathit{final}_{\rightarrow,\mathit{extract}}$, and pick an arbitrary $H \in \mathcal{H}$ such that $\bot \not\in H$.
  Then
  \begin{equation}\label{eq:nobot}
    {(\sup f_n)}^{-1}(H) = \bigcup_n f_n^{-1}(H),
  \end{equation}
  which is measurable by definition.
  Also,
  \begin{equation}\label{eq:bot}
    {(\sup f_n)}^{-1}(\{\bot\}) = \bigcap_n f_n^{-1}(\{\bot\})
  \end{equation}
  is also measurable by definition.
  Now assume $\bot \in H$.
  Then
  \begin{equation}
    \begin{aligned}
      {(\sup f_n)}^{-1}(H) &= {(\sup f_n)}^{-1}((H \setminus \{\bot\}) \cup \{\bot\}) \\
      &= {(\sup f_n)}^{-1}(H \setminus \{\bot\}) \cup {(\sup f_n)}^{-1}(\{\bot\}),
    \end{aligned}
  \end{equation}
  which is also measurable by~\eqref{eq:nobot} and~\eqref{eq:bot}.
\end{proof}

We summarize all of the above in the following lemma.
\begin{lemma}\label{lemma:final}
  Given
  \begin{enumerate}
    \item a measurable function
      $\rightarrow : (A',\mathcal{A}|_{A'}) \to (\mathbb{A},\mathcal{A})$, and

    \item
      a measurable space $(\mathbb{H},\mathcal{H})$ equipped with a flat partial order $\leq_\mathbb{H}$ (we require $\{\bot\} \in \mathcal{H}$), and
    \item
      a measurable function $\mathit{extract} : (\mathbb{A},\mathcal{A}) \to (\mathbb{H},\mathcal{H})$, such that for all $a \in \mathbb{A}$,
      \begin{equation}
        \mathit{extract}(a)
        \leq_\mathbb{H} \mathit{extract}(\mathit{step}_\rightarrow(a))
      \end{equation}
      where $\mathit{step}_\rightarrow = \ \rightarrow \cup \ \id|_{\mathbb{A} \setminus A'}$,
  \end{enumerate}
  the function $\mathit{final}_{\rightarrow,\mathit{extract}} = \sup \{\mathit{extract} \circ \mathit{step}_\rightarrow^n \mid n \in \mathbb{N}_0 \} : (A,\mathcal{A}) \to (H,\mathcal{H})$ exists and is measurable
\end{lemma}
\subsection{The Measurable Functions $r_\term$ and $f_\term$}%
\label{sec:rfmeasurable}
In this section, we prove that $r_\term$ and $f_\term$ are measurable.
We follow the proof strategy from Borgström et al.~\cite{borgstrom2015lambda}.

\begin{condition}
  We require that, for each identifier $D \in \mathbb{D}$, the function
  \begin{equation}
    F^{-1}_D :
    (\R^{|D|} \times [0,1], \Bor^{|D|} \otimes
    \Bor_{[0,1]})
    \to (\R, \Bor)
  \end{equation}
  is measurable.
\end{condition}

\begin{condition}
  We require that, for each identifier $g \in \mathbb{G}$, the function
  \begin{equation}
    \sigma_g : (\R^{|g|},\Bor^{|g|}) \to (\R, \Bor)
  \end{equation}
  is measurable.
\end{condition}

\begin{definition}
  \begin{equation}
    \begin{aligned}
      \T_{\textsc{App}} &=
      \{
        \econt[(\lambda x. \term) \enspace \val] \mid
        \econt \in \mathbb{E},
        ((\lambda x. \term) \enspace \val) \in \T
      \}
      \\
      \T_{\textsc{Prim}} &=
      \{
        \econt[g(c_1, \ldots, c_{|g|})] \mid
        \econt \in \mathbb{E},
        g \in \mathbb{G},
        (c_1, \ldots, c_{|g|}) \in \R^{|g|}
      \}
      \\
      \T_{\textsc{IfTrue}} &=
      \{
        \econt[\ttt{if} \ \mathit{true} \ \ttt{then} \ \term_1 \ \ttt{else} \ \term_2]
        \mid
        \econt \in \mathbb{E},
        \term_1 \in \T,\term_2 \in \T
      \}
      \\
      \T_{\textsc{IfFalse}} &=
      \{
        \econt[\ttt{if} \ \mathit{false} \ \ttt{then} \ \term_1 \ \ttt{else} \ \term_2]
        \mid
        \econt \in \mathbb{E},
        \term_1 \in \T,\term_2 \in \T
      \}
      \\
      \T_{d} &=
      \T_\textsc{App} \cup
      \T_\textsc{Prim} \cup
      \T_\textsc{IfTrue} \cup
      \T_\textsc{IfFalse}
    \end{aligned}
  \end{equation}
\end{definition}
\begin{lemma}\label{lemma:detsetmeasurable}
    $\T_\textsc{App}$,
    $\T_\textsc{Prim}$,
    $\T_\textsc{IfTrue}$,
    $\T_\textsc{IfFalse}$,
    and $\T_{d}$
    are $\Tcal$-measurable.
\end{lemma}
\begin{proof}
  We can write all of these sets as countable unions of sets of the form $\term_p^n(\R_n)$.
  Hence, they must be $\Tcal$-measurable.
\end{proof}

\begin{definition}
  \begin{equation}
    \begin{gathered}
      \Tcal_\textsc{App} = \Tcal|_{\T_\textsc{App}} \quad
      \Tcal_\textsc{Prim} = \Tcal|_{\T_\textsc{Prim}} \\
      \Tcal_\textsc{IfTrue} = \Tcal|_{\T_\textsc{IfTrue}} \quad
      \Tcal_\textsc{IfFalse} = \Tcal|_{\T_\textsc{IfFalse}} \\
      \Tcal_d = \Tcal|_{\T_{d}}.
    \end{gathered}
  \end{equation}
\end{definition}
\begin{lemma}
    $\Tcal_\textsc{App}$,
    $\Tcal_\textsc{Prim}$,
    $\Tcal_\textsc{IfTrue}$,
    $\Tcal_\textsc{IfFalse}$,
    and $\Tcal_d$
    are $\sigma$-algebras.
\end{lemma}
\begin{proof}
  By Lemma~\ref{lemma:measrest}.
\end{proof}

\begin{lemma}\label{lemma:detinduce}
  \begin{equation}
    \begin{gathered}
      \Bor_{(\T_\textsc{App},\dt)} = \Tcal_\textsc{App} \quad
      \Bor_{(\T_\textsc{Prim},\dt)} = \Tcal_\textsc{Prim} \\
      \Bor_{(\T_\textsc{IfTrue},\dt)} = \Tcal_\textsc{IfTrue} \quad
      \Bor_{(\T_\textsc{IfFalse},\dt)} = \Tcal_\textsc{IfFalse} \\
      \Bor_{(\T_{d},\dt)} = \Tcal_d.
    \end{gathered}
  \end{equation}
\end{lemma}
\begin{proof}
  By Lemma~\ref{lemma:metricequalterm} and Lemma~\ref{lemma:metricrest}.
\end{proof}

\begin{definition}
  \begin{align}
    \mathit{step}_{\textsc{App}}(\econt[(\lambda x. \term) \enspace \val])
    &= \econt[[x \mapsto \val]\term] \\
    \mathit{step}_{\textsc{Prim}}(\econt[g(c_1, \ldots, c_{|g|})])
    &= \econt[\sigma_g(c_1, \ldots, c_{|g|})] \\
    \mathit{step}_{\textsc{IfTrue}}
    (\econt[\ttt{if} \ \mathit{true} \ \ttt{then} \ \term_1 \ \ttt{else} \ \term_2])
    &= \econt[\term_1] \\
    \mathit{step}_{\textsc{IfFalse}}
    (\econt[\ttt{if} \ \mathit{false} \ \ttt{then} \ \term_1 \ \ttt{else} \ \term_2])
    &= \econt[\term_2]
  \end{align}
\end{definition}

\begin{lemma}\label{lemma:detunion}
  $\rightarrow_\textsc{Det} =
  \mathit{step}_{\textsc{App}} \cup
  \mathit{step}_{\textsc{Prim}} \cup
  \mathit{step}_{\textsc{IfTrue}} \cup
  \mathit{step}_{\textsc{IfFalse}}$
\end{lemma}
\begin{proof}
  By inspection.
\end{proof}
\begin{lemma}
  The relation $\rightarrow_\textsc{Det}$ is a function.
\end{lemma}
\begin{proof}
  The functions
  $\mathit{step}_{\textsc{App}}$,
  $\mathit{step}_{\textsc{Prim}}$,
  $\mathit{step}_{\textsc{IfTrue}}$,
  and $\mathit{step}_{\textsc{IfFalse}}$ have disjoint domains.
  It follows that $\rightarrow_\textsc{Det}$ is a function.
\end{proof}

\begin{lemma}
  $\mathit{step}_\textsc{App} :
  (\T_\textsc{App},\Tcal_\textsc{App})
  \to
  (\T,\Tcal)$ is measurable.
\end{lemma}
\begin{proof}
  We show that $\mathit{step}_\textsc{App}$ is continuous as a function between
  the metric spaces $(\T_\textsc{App},\dt)$ and $(\T,\dt)$.
  By Lemma~\ref{lemma:detinduce} and Lemma~\ref{lemma:contmeas}, the result then follows.

  Pick arbitrary $\econt[(\lambda x. \term) \enspace \val] \in \T_\textsc{App}$ and $\varepsilon > 0$.
  Following Definition~\ref{definition:cont}, we want to show that there exists a $\delta > 0$ such that for all $\econt'[(\lambda x. \term') \enspace \val'] \in \T_\textsc{App}$,
  \begin{equation}\label{eq:stepappcont}
    \dt(\econt[(\lambda x. \term) \enspace \val],
    \econt'[(\lambda x. \term') \enspace \val']) < \delta
    \implies
    \dt(\econt[[x \mapsto \val]\term],
    \econt'[[x \mapsto \val']\term']) < \varepsilon
  \end{equation}
  By applying Lemma~\ref{lemma:econtleq}, Lemma~\ref{lemma:econtsubst}, and Lemma~\ref{lemma:econtredex} (in that order), we have
  \begin{equation}
    \begin{aligned}
      \dt(\econt[[x \mapsto \val]\term],
      \quad & \\
      \econt'[[x \mapsto \val']\term'])
      &\leq
      \de(\econt,\econt')
      + \dt([x \mapsto \val]\term, [x \mapsto \val']\term')
      \\ &\leq
      \de(\econt,\econt')
      + \dt(\term,\term') + k \cdot \dt(\val,\val')
      \\ &\leq
      (k + 1) \cdot ( \de(\econt,\econt') + \dt(\term,\term') + \dt(\val,\val'))
      \\ &=
      (k + 1) \cdot ( \de(\econt,\econt')
      + \dt((\lambda x. \term) \enspace \val,
      (\lambda x. \term') \enspace \val'))
      \\ &=
      (k + 1) \cdot \dt(\econt[(\lambda x. \term) \enspace
      \val],\econt'[(\lambda x. \term') \enspace \val'])
    \end{aligned}
  \end{equation}
  Hence, we see that by selecting $\delta = \frac{\varepsilon}{k+1}$, we get the implication~\eqref{eq:stepappcont} and the function is continuous, and hence measurable.
\end{proof}

\begin{lemma}\label{lemma:primmeasurable}
  $\mathit{step}_\textsc{Prim} :
  (\T_\textsc{Prim},\Tcal_\textsc{Prim})
  \to
  (\T,\Tcal)$ is measurable.
\end{lemma}
\begin{proof}
  Define
  \begin{equation}
    \begin{aligned}
      \mathit{unbox}(\econt[g(c_1, \ldots, c_{|g|})])
      &= (c_1,\ldots,c_{|g|}) \\
      \mathit{box}_{\econt}(c)
      &= (\econt[c]).
    \end{aligned}
  \end{equation}
  For any $\econt'[g(c_1', \ldots, c_{|g|}')] \in \T_\textsc{Prim}$, by Lemma~\ref{lemma:econtredex}, we have
  \begin{equation}
    \begin{aligned}
      \drn{|g|}((c_1, \ldots, c_{|g|}),
      \quad & \\
      (c_1', \ldots, c_{|g|}'))
      &=
      \dr(c_1,c_2) + \cdots + \dr(c_{|g|},c_{|g|}')
      \\ &\leq
      \de(\econt,\econt') + \dt(c_1,c_2) + \cdots + \dt(c_{|g|},c_{|g|}')
      \\ &=
      \de(\econt[g(c_1, \ldots, c_{|g|})], \econt[g(c_1, \ldots, c_{|g|})])
    \end{aligned}
  \end{equation}
  From this, it follows that $\mathit{unbox}$ is continuous (set $\delta = \varepsilon$) and hence measurable.
  Furthermore,
  \begin{equation}
    \dt(\econt[c],\econt[c'])
    =
    \de(\econt,\econt) + \dt(c,c')
    =
    \dt(c,c'),
  \end{equation}
  implying that $\mathit{box}_\econt$ is continuous (set $\delta = \varepsilon$) and measurable as well.

  Lastly, we have
  \begin{equation}
    \mathit{step}_\textsc{Prim} =
    \bigcup_{\econt \in \mathbb{E}}
    \bigcup_{g \in \mathbb{G}} \mathit{box}_\econt \circ \sigma_g \circ
    \mathit{unbox}.
  \end{equation}
  It holds that $\mathit{box}_\econt \circ \sigma_g \circ \mathit{unbox}$ is measurable (composition of measurable functions) for each $g$ and $\econt$.
  Because $\mathbb{E}$ and $\mathbb{G}$ are countable, by Lemma~\ref{lemma:decompmeas}, $\mathit{step}_\textsc{Prim}$ is measurable.
\end{proof}

\begin{lemma}\label{lemma:contiftrue}
  $\mathit{step}_\textsc{IfTrue} :
  (\T_\textsc{IfTrue},\Tcal_\textsc{IfTrue})
  \to
  (\T,\Tcal)$ is measurable.
\end{lemma}
\begin{proof}
  We show that $\mathit{step}_\textsc{IfTrue}$ is continuous as a function between the metric spaces $(\T_\textsc{IfTrue},\dt)$ and $(\T,\dt)$.
  By Lemma~\ref{lemma:detinduce} and Lemma~\ref{lemma:contmeas}, the result then follows.

  Pick arbitrary $\econt[\ttt{if } \mathit{true} \ttt{ then } \term_1 \ttt{ else } \term_2] \in \T_\textsc{IfTrue}$ and $\varepsilon > 0$.
  Following Definition~\ref{definition:cont}, we want to show that there exists a $\delta > 0$ such that for all $\econt'[\ttt{if } \mathit{true} \ttt{ then } \term_1' \ttt{ else } \term_2'] \in \T_\textsc{IfTrue}$,
  \begin{equation}\label{eq:contiftrue}
    \begin{aligned}
      &
      \dt(
      \econt[\ttt{if} \ \mathit{true} \ \ttt{then} \ \term_1 \ \ttt{else} \ \term_2],
      \econt'[\ttt{if} \ \mathit{true} \ \ttt{then} \ \term_1' \ \ttt{else} \ \term_2'])
      < \delta
      \\ & \quad
      \implies
      \dt(\econt[\term_1], \econt'[\term_1'])
      < \varepsilon
    \end{aligned}
  \end{equation}
  We have
  \begin{equation}
    \begin{aligned}
      \dt(\econt[\term_1], \econt'[\term_1'])
      &\leq
      \de(\econt, \econt') + \dt(\term_1,\term_1')
      \\ &\leq
      \de(\econt, \econt') + \dt(\term_1,\term_1') + \dt(\term_2,\term_2')
      \\ &=
      \dt(
      \econt[\ttt{if} \ \mathit{true} \ \ttt{then} \ \term_1 \ \ttt{else} \ \term_2],
      \\ & \hspace{15mm}
      \econt'[\ttt{if} \ \mathit{true} \ \ttt{then} \ \term_1' \ \ttt{else} \ \term_2'])
    \end{aligned}
  \end{equation}
  Hence, we see that by selecting $\delta = \varepsilon$, we get the implication~\eqref{eq:contiftrue} and the function is continuous, and hence measurable.
\end{proof}

\begin{lemma}
  $\mathit{step}_\textsc{IfFalse} :
  (\T_\textsc{IfFalse},\Tcal_\textsc{IfFalse})
  \to
  (\T,\Tcal)$ is measurable.
\end{lemma}
\begin{proof}
  Analogous to the proof of Lemma~\ref{lemma:contiftrue}
\end{proof}

\begin{lemma}
  $\rightarrow_\textsc{Det} : (\T_{d},\Tcal_d) \to (\T,\Tcal)$ is measurable.
\end{lemma}
\begin{proof}
  Follows from Lemma~\ref{lemma:detunion} and Lemma~\ref{lemma:decompmeas}.
\end{proof}

Let us now make the following definition
\begin{definition}
  $
    \mathit{extract}_{\rightarrow_\textsc{Det}}(\term) =
    \begin{cases}
      \term & \text{if } \term \not\in \T_d \\
      \bot & \text{otherwise}.
    \end{cases}
  $
\end{definition}
With $\T_\bot = \T \cup \{\bot\}$, and $\Tcal_\bot$ the corresponding least $\sigma$-algebra such that $\Tcal \subset \Tcal_\bot$ (which must necessarily contain $\{\bot\}$), we have the following lemma.
\begin{lemma}
  $\mathit{extract}_{\rightarrow_\textsc{Det}} : (\T,\Tcal) \to
  (\T_\bot,\Tcal_\bot)$
  is measurable.
\end{lemma}
\begin{proof}
  We have $\mathit{extract}_{\rightarrow_\textsc{Det}} = \id|_{\T_d^c} \cup \bot|_{\T_d}$, where $\bot$ here denotes the constant function producing $\bot$.
  Because $\id$, $\bot$, and $\T_d$ are measurable, the result follows by Lemma~\ref{lemma:decompmeas}.
\end{proof}
\begin{definition}
  The partial order $\leq_d$ is the least partial order on $\T_\bot$ with $\bot \leq_d \term$.
\end{definition}
\begin{lemma}
  \begin{equation}
    \mathit{extract}_{\rightarrow_\textsc{Det}}(\term)
    \leq_d \mathit{extract}_{\rightarrow_\textsc{Det}}(\mathit{step}_{\rightarrow_\textsc{Det}}(\term)),
  \end{equation}
  where $\mathit{step}_{\rightarrow_\textsc{Det}} = \ \rightarrow_\textsc{Det} \cup \ \id|_{\T \setminus \T_{d}}$.
\end{lemma}
\begin{proof}
  Consider first $\term \in \T_{d}$.
  We then have $\mathit{extract}_{\rightarrow_\textsc{Det}}(\term) = \bot$ by definition, and the result holds immediately.
  Now consider $\term \not\in \T_{d}$.
  By definition, $\mathit{step}_{\rightarrow_\textsc{Det}}(\term) = \term$, and the result holds.
\end{proof}
Lastly, we apply Lemma~\ref{lemma:final} to get the measurable function $\mathit{final}_{\rightarrow_\textsc{Det}}$.
\begin{definition}
  $
  \mathit{final}_{\rightarrow_\textsc{Det}} =
  \mathit{final}_{\rightarrow_\textsc{Det},\mathit{extract}_{\rightarrow_\textsc{Det}}}
  : (\T,\Tcal)
  \to (\T_\bot,\Tcal_\bot)
  $
\end{definition}

We now proceed to the stochastic semantics.
\begin{definition}
  \begin{align}
    \T_{\textsc{Sample}} &=
    \{
      \econt[\ttt{sample}_D(c_1,\ldots,c_{|D|})]
      \mid \econt \in \mathbb{E}, D \in \mathbb{D},
      \\ & \hspace{4cm}
      (c_1,\ldots,c_{|D|}) \in \R^{|D|}
    \}
    \\
    \T_{\textsc{Weight}} &=
    \{
      \econt[\ttt{weight}(c)] \mid \econt \in \mathbb{E}, c \in \R_+
    \}
    \\
    \T_{\textsc{Resample}} &=
    \{
      \econt[\ttt{resample}] \mid \econt \in \mathbb{E}, c \in \R_+
    \}
    \\
    \T_{\textsc{Det}} &=
    \mathit{final}_{\rightarrow_\textsc{Det}}^{-1}(
    \mathbb{V}
    \cup \T_\textsc{Sample}
    \cup \T_\textsc{Weight}
    \cup \T_{\textsc{Resample}}
    )
    \\ & \hspace{2cm} \setminus
    (\mathbb{V}
    \cup \T_\textsc{Sample}
    \cup \T_\textsc{Weight}
    \cup \T_{\textsc{Resample}}
    )
    \\
    \T_{s} &=
    \T_{\textsc{Det}} \cup
    \T_\textsc{Sample} \cup
    \T_\textsc{Weight} \cup
    \T_\textsc{Resample}
  \end{align}
\end{definition}
\begin{lemma}\label{lemma:stochsetmeasurable}
    $\T_\textsc{Sample}$,
    $\T_\textsc{Weight}$,
    $\T_\textsc{Resample}$,
    $\T_{\textsc{Det}}$,
    and $\T_{s}$ are $\Tcal$-measurable.
\end{lemma}
\begin{proof}
  We can write the sets $\T_\textsc{Sample}$, $\T_\textsc{Weight}$, and $\T_\textsc{Resample}$ as countable unions of sets of the form $\term_p^n(\R_n)$.
  Hence, they must be $\Tcal$-measurable.
  $\T_{\textsc{Det}}$ is measurable because $\mathit{final}_{\rightarrow_\textsc{Det}}$ is a measurable function, and $\mathbb{V}$, $\T_\textsc{Sample}$, $\T_\textsc{Weight}$, and $\T_\textsc{Resample}$ are measurable.
  Finally, $\T_s$ is measurable because it is a finite union of measurable sets.
\end{proof}

The below Lemma allows us to ignore the element $\bot$ introduced by the function $\mathit{extract}_{\rightarrow_\textsc{Det}}$.
\begin{lemma}\label{lemma:botremove}
  $\mathit{final}_{\rightarrow_\textsc{Det}}({\T_{s}})
  \subset \T$
\end{lemma}
\begin{proof}
  By definition, $\mathit{final}_{\rightarrow_\textsc{Det}}(\T_{s})
  \subset
  \mathbb{V} \cup
  \T_\textsc{Sample}' \cup
  \T_\textsc{Weight}' \cup
  \T_\textsc{Resample}'$.
  The result follows.
\end{proof}

\begin{lemma}\label{lemma:finalrestmeas}
  $\mathit{final}_{\rightarrow_\textsc{Det}}|_{\T_{s}} : (\T,\Tcal) \to (\T,\Tcal)$ is measurable.
\end{lemma}
\begin{proof}
  The restriction of a measurable function to a measurable set is also a measurable function (follows from Lemma~\ref{lemma:decompmeas}).
  Furthermore, we can restrict the codomain from $(\T_\bot,\Tcal_\bot)$ to $(\T,\Tcal)$ as a result of Lemma~\ref{lemma:botremove} and by the definition of $(\T_\bot,\Tcal_\bot)$.
\end{proof}

\begin{lemma}\label{lemma:finalrestrictedmeasurable}
  For $T \subset \T_{s}$, $T \in \Tcal$, $\mathit{final}_{\rightarrow_\textsc{Det}}|_T : (\T,\Tcal) \to (\T,\Tcal)$ is measurable.
\end{lemma}
\begin{proof}
  Follows directly from Lemma~\ref{lemma:finalrestmeas}.
\end{proof}

\begin{definition}
  \begin{equation}
    \begin{gathered}
      \Tcal_\textsc{Det} =
      \Tcal|_{\T_{\textsc{Det}}} \quad
      \Tcal_\textsc{Sample} =
      \Tcal|_{\T_\textsc{Sample}} \\
      \Tcal_\textsc{Weight} =
      \Tcal|_{\T_\textsc{Weight}} \quad
      \Tcal_\textsc{Resample} =
      \Tcal|_{\T_\textsc{Resample}} \quad
      \Tcal_s =
      \Tcal|_{\T_{s}}.
    \end{gathered}
  \end{equation}
\end{definition}

\begin{lemma}
  $\Tcal_\textsc{Det}$,
  $\Tcal_\textsc{Sample}$,
  $\Tcal_\textsc{Weight}$,
  $\Tcal_\textsc{Resample}$,
  and $\Tcal_s$
  are $\sigma$-algebras.
\end{lemma}
\begin{proof}
  By Lemma~\ref{lemma:measrest}.
\end{proof}

\begin{lemma}
  \begin{equation}
    \begin{gathered}
      \Bor_{(\T_{\textsc{Det}},\dt)} = \Tcal_\textsc{Det} \quad
      \Bor_{(\T_\textsc{Sample},\dt)} = \Tcal_{\textsc{Sample}} \quad
      \Bor_{(\T_\textsc{Weight},\dt)} = \Tcal_{\textsc{Weight}} \\
      \Bor_{(\T_\textsc{Resample},\dt)} = \Tcal_{\textsc{Resample}} \quad
      \Bor_{(\T_{s},\dt)} = \Tcal_s.
    \end{gathered}
  \end{equation}
\end{lemma}
\begin{proof}
  By Lemma~\ref{lemma:metricequalterm} and Lemma~\ref{lemma:metricrest}.
\end{proof}

\begin{definition}
  Let ${\{\mathbb{M}_i\}}_{i=1}^n$ be a finite set of spaces.
  We define the \emph{$j$-th projection} $\pi_j : \bigtimes_{i=1}^n \mathbb{M}_i \to \mathbb{M}_j$ as
  \begin{equation}
    \pi_j(m_1,\ldots,m_j,\ldots,m_n) = m_j
  \end{equation}
\end{definition}

\begin{lemma}
  If ${\{\mathbb{M}_i,d_i\}}_{i=1}^n$ is a set of metric spaces, then $\pi_j$ is continuous as a function between the product metric space $\left(\bigtimes_{i=1}^n \mathbb{M}_i,\sum_{i=1}^n d_i\right)$ and the metric space $(\mathbb{M}_j,d_j)$.
\end{lemma}
\begin{proof}
  Pick
  \begin{equation}
    (m_1,\ldots,m_j,\ldots,m_n) \in \bigtimes_{i=1}^n \mathbb{M}_i
  \end{equation}
  and $\varepsilon > 0$.
  Now, for all
  \begin{equation}
    (m_1',\ldots,m_j',\ldots,m_n') \in \bigtimes_{i=1}^n \mathbb{M}_i,
  \end{equation}
  we have
  \begin{equation}
    d_j(m_j,m_j')
    \leq
    \sum_{i=1}^n d_i(m_i,m_i')
  \end{equation}
  Hence, by choosing $\delta = \varepsilon$, we see that $\pi_j$ is continuous.
\end{proof}

\begin{definition}
  \begin{equation}
    \begin{aligned}
      \mathit{step}_\textsc{Det}
      &=
      (\mathit{final}_{\rightarrow_\textsc{Det}}|_{\T_{\textsc{Det}}} \circ \pi_1
      ,\pi_2,\pi_3)
      \\
      \mathit{step}_\textsc{Sample}
      (\econt[\ttt{sample}_D(c_1,\ldots,c_{|D|})],
      \quad & \\ w,p::s) &=
      \econt[F^{-1}_{D}(c_1, \ldots, c_{|D|},p)],w,s
      \\
      \mathit{step}_\textsc{Weight}
      (\econt[\ttt{weight}(c)],w,s) &= \econt[()],w\cdot c,s
      \\
      \mathit{step}_\textsc{Resample}
      (\econt[\ttt{resample}],w,s) &= \econt[()],w,s
      \\
    \end{aligned}
  \end{equation}
\end{definition}
\begin{lemma}\label{lemma:stochunion}
  $\rightarrow =
  \mathit{step}_\textsc{Det} \cup
  \mathit{step}_\textsc{Sample} \cup
  \mathit{step}_\textsc{Weight} \cup
  \mathit{step}_\textsc{Resample}$.
\end{lemma}
\begin{proof}
  By inspection.
\end{proof}

\begin{lemma}
  The relation $\rightarrow$ is a function.
\end{lemma}
\begin{proof}
  The functions
  $\mathit{step}_\textsc{Det}$,
  $\mathit{step}_\textsc{Sample}$,
  $\mathit{step}_\textsc{Weight}$,
  and $\mathit{step}_\textsc{Resample}$ have disjoint domains.
  It follows that $\rightarrow$ is a function.
\end{proof}

Now, let $\Tr_{1:} = \Tr \setminus \{\trace{}\}$.
We make the following definitions.
\begin{definition}
  $\X = \T \times \R_+ \times \Tr$
  and
  $\Xcal = \Tcal \otimes \Bor_+ \otimes \Trcal$.
\end{definition}
\begin{definition}
  \begin{equation}
    \begin{aligned}
      \X_\textsc{Det} &= \T_\textsc{Det} \times \R_+ \times \Tr \\
      \X_\textsc{Sample} &= \T_\textsc{Sample} \times \R_+ \times
      \Tr_{1:} \\
      \X_\textsc{Weight} &= \T_\textsc{Weight} \times \R_+ \times \Tr
      \\
      \X_\textsc{Resample} &=
      \T_\textsc{Resample} \times \R_+ \times \Tr,
      \\
      \X_s &= \X_\textsc{Det}
      \cup \X_\textsc{Sample}
      \cup \X_\textsc{Weight}
      \cup \X_\textsc{Resample}
    \end{aligned}
  \end{equation}
\end{definition}
\begin{lemma}
  $\X_\textsc{Det},
  \X_\textsc{Sample},
  \X_\textsc{Weight},
  \X_\textsc{Resample},
  $ and $\X_s$ are all $\Xcal$-measurable.
\end{lemma}
\begin{proof}
  $\X_\textsc{Det},
  \X_\textsc{Sample},
  \X_\textsc{Weight},
  $ and $\X_\textsc{Resample}$
  are the Cartesian products of measurable sets, hence measurable.
  $\X_s$ is a finite union of measurable sets, hence measurable.
\end{proof}

\begin{definition}
  \begin{equation}
    \begin{gathered}
      \Xcal_\textsc{Det} = \Xcal|_{\X_\textsc{Det}} \qquad
      \Xcal_\textsc{Sample} = \Xcal|_{\X_\textsc{Sample}} \\
      \Xcal_\textsc{Weight} = \Xcal|_{\X_\textsc{Weight}} \qquad
      \Xcal_\textsc{Resample} = \Xcal|_{\X_\textsc{Resample}} \qquad
      \Xcal_s = \Xcal|_{\X_s}
    \end{gathered}
  \end{equation}
\end{definition}
\begin{lemma}
  $\Xcal_\textsc{Det},
  \Xcal_\textsc{Sample},
  \Xcal_\textsc{Weight},
  \Xcal_\textsc{Resample},
  $ and $\Xcal_s$ are $\sigma$-algebras.
\end{lemma}
\begin{proof}
  Follows from Lemma~\ref{lemma:measrest}.
\end{proof}

\begin{lemma}\label{lemma:xsep}
  Let $d_\X = \dt + \drpos + \ds$.
  Then
  \begin{equation}
    \begin{gathered}
      \Bor_{(\X, d_\X)} = \Xcal \quad
      \Bor_{(\X_\textsc{Det}, d_\X)} = \Xcal_\textsc{Det} \quad
      \Bor_{(\X_\textsc{Sample}, d_\X)} = \Xcal_\textsc{Sample} \\
      \Bor_{(\X_\textsc{Weight}, d_\X)} = \Xcal_\textsc{Weight} \quad
      \Bor_{(\X_\textsc{Resample}, d_\X)} = \Xcal_\textsc{Resample} \quad
      \Bor_{(\X_s, d_\X)} = \Xcal_s.
    \end{gathered}
  \end{equation}
  Furthermore, $(\X, d_\X)$ is a separable metric space.
\end{lemma}
\begin{proof}
  Follows from Lemmas~\ref{lemma:termsep}, \ref{lemma:rsep}, \ref{lemma:metricrest}, \ref{lemma:tracesep}, and \ref{lemma:boreqmeas}.
\end{proof}

\begin{lemma}
  $\mathit{step}_\textsc{Det} : (\X_\textsc{Det},\Xcal_\textsc{Det}) \to (\X,\Xcal)$ is measurable.
\end{lemma}
\begin{proof}
  The projections $\pi_1$,$\pi_2$, and $\pi_3$ are continuous and hence measurable.
  From Lemma~\ref{lemma:finalrestrictedmeasurable}, $\mathit{final}_{\rightarrow_\textsc{Det}}|_{\T_{\textsc{Det}}}$ is measurable, and therefore, so is the composition $\mathit{final}_{\rightarrow_\textsc{Det}}|_{\T_{\textsc{Det}}} \circ \pi_1$.
  By Lemma~\ref{lemma:compmeasurable}, the result now follows.
\end{proof}

\begin{lemma}
  $\mathit{step}_\textsc{Sample} : (\X_\textsc{Sample},\Xcal_\textsc{Sample}) \to (\X,\Xcal)$ is measurable.
\end{lemma}
\begin{proof}
  Pick arbitrary $\econt \in \mathbb{E}$ and $D \in \mathbb{D}$, and define
  \begin{equation}
    \begin{aligned}
      \mathit{unbox}(\econt[\ttt{sample}_D(c_1, \ldots, c_{|D|})])
      &= (c_1,\ldots,c_{|D|}) \\
      \mathit{box}_{\econt}(c)
      &= (\econt[c]).
    \end{aligned}
  \end{equation}
  By copying the arguments from the proof of Lemma~\ref{lemma:primmeasurable}, $\mathit{unbox}$ and $\mathit{box}_{\econt}$ are measurable.
  Next, define
  \begin{equation}
    \begin{aligned}
      \mathit{head}(p :: s) &= p \\
      \mathit{tail}(p :: s) &= s.
    \end{aligned}
  \end{equation}
  Pick an arbitrary $p :: s \in \Tr_{1:}$.
  Clearly, by Lemma~\ref{lemma:metricrest}, $(\Tr_{1:},d_\Tr)$ is a separable metric space.
  For any $p' :: s' \in \Tr_{1:}$, we have
  \begin{equation}
    \dr(p, p') \leq \dr(p, p') + d_\Tr(s, s') = d_\Tr(p :: s, p' :: s').
  \end{equation}
  By letting $\delta = \varepsilon$, we see that $\mathit{head}$ is continuous and hence measurable.
  Furthermore, by a similar argument, $\mathit{tail}$ is continuous and measurable.

  Now, we note that
  \begin{equation}
    \mathit{step}_\textsc{Sample} =
    \bigcup_{\econt \in \mathbb{E}}
    \bigcup_{D \in \mathbb{D}}
    (\mathit{box}_\econt \circ F^{-1}_D
    \circ (\mathit{unbox} \circ \pi_1, \mathit{head} \circ \pi_2),
    \pi_2, \mathit{tail} \circ \pi_3).
  \end{equation}
  By the measurability of the component functions, Lemma~\ref{lemma:compmeasurable} (applied twice), and Lemma~\ref{lemma:decompmeas}, the result follows.
\end{proof}

\begin{lemma}
  $\mathit{step}_\textsc{Weight} : (\X_\textsc{Weight},\Xcal_\textsc{Weight}) \to (\X,\Xcal)$ is measurable.
\end{lemma}
\begin{proof}
  Pick arbitrary $\econt \in \mathbb{E}$ and define
  \begin{equation}
    \begin{aligned}
      \mathit{unbox}(\econt[\ttt{weight}(c)]) &= c \\
      \mathit{box}_\econt(c)(c) &= \econt[c]
    \end{aligned}
  \end{equation}
\end{proof}
By using similar arguments as in the proof of Lemma~\ref{lemma:primmeasurable}, it holds that $\mathit{unbox}$ and $\mathit{box}$ are measurable.

Now, we note that
\begin{equation}
  \mathit{step}_\textsc{Weight} =
  \bigcup_{\econt \in \mathbb{E}}
  (\mathit{box}_\econt \circ \mathit{unbox} \circ \pi_1,
  (\mathit{unbox} \circ \pi_1) \cdot \pi_2,
  \pi_3)
\end{equation}
Here, $\cdot$ denotes the pointwise function product.
That is, for two functions $f$ and $g$, $(f \cdot g)(x) = f(x)\cdot f(g)$.
It is a standard result in measure theory that the function product of two measurable functions is measurable.
By the measurability of the component functions, Lemma~\ref{lemma:compmeasurable}, and Lemma~\ref{lemma:decompmeas}, the result now follows.

\begin{lemma}\label{lemma:stepresampmeas}
  $\mathit{step}_\textsc{Resample} : (\X_\textsc{Resample},\Xcal_\textsc{Resample}) \to (\X,\Xcal)$ is measurable.
\end{lemma}
\begin{proof}
  Let $\mathit{eval}(\econt[\ttt{resample}]) = \econt[()]$.
  Clearly, $\mathit{eval}$ is continuous and hence measurable.
  Now,
  \begin{equation}
    \mathit{step}_\textsc{Resample} =
    \bigcup_{\econt \in \mathbb{E}}
    (\mathit{eval} \circ \pi_1, \pi_2, \pi_3).
  \end{equation}
  By the measurability of the component functions, Lemma~\ref{lemma:compmeasurable}, and Lemma~\ref{lemma:decompmeas}, the result now follows.
\end{proof}

\begin{lemma}\label{lemma:arrowsmeas}
  $\rightarrow : (\X_\textsc{Det},\Xcal_\textsc{Det}) \to (\X,\Xcal)$ is measurable.
\end{lemma}
\begin{proof}
  Follows from Lemma~\ref{lemma:detunion} and Lemma~\ref{lemma:decompmeas}.
\end{proof}

We make the following definitions.
\begin{definition}
  $
    \mathit{extract}_{\rightarrow,\text{term}}(\term,w,s) =
    \begin{cases}
      \term & \text{if } \term \in \mathbb{V}, s = \trace{} \\
      () & \text{otherwise}
    \end{cases}
  $
\end{definition}

\begin{definition}
  $
    \mathit{extract}_{\rightarrow,\text{weight}}(\term,w,s) =
    \begin{cases}
      w & \text{if } \term \in \mathbb{V}, s = \trace{} \\
      0 & \text{otherwise }
    \end{cases}
  $
\end{definition}

\begin{lemma}\label{lemma:extractstermmeas}
  $\mathit{extract}_{\rightarrow,\text{term}} :
  (\X,\Xcal)
  \to
  (\T, \Tcal)$
  is measurable.
\end{lemma}
\begin{proof}
  Let $\Tr_0 = \{\trace{}\}$ and $X = \mathbb{V} \times \R_+ \times \Tr_0$.
  We have $\mathit{extract}_{\rightarrow,\text{term}} = \pi_1|_X \cup ()|_{X^c}$, where $()$ here denotes the constant function producing $()$.
  Because $\id$, $()$, and $X$ are measurable, the result follows by Lemma~\ref{lemma:decompmeas}.
\end{proof}

\begin{lemma}
  $\mathit{extract}_{\rightarrow,\text{weight}} :
  (\X,\Xcal)
  \to
  (\R_+, \Bor_+)$
  is measurable.
\end{lemma}
\begin{proof}
  Analogous to the proof of Lemma~\ref{lemma:extractstermmeas}.
\end{proof}

\begin{definition}
  $\leq_{s,\text{term}}$ is the least partial order on $\T$ such that $() \leq_{s,\text{term}} \term$.
\end{definition}

\begin{definition}
  $\leq_{s,\text{weight}}$ is the least partial order on $\R_+$ such that $0 \leq_{s,\text{weight}} w$.
\end{definition}

With $\mathit{step}_{\rightarrow} = \ \rightarrow \cup \ \id|_{\X \setminus \X_s}$, we have the below lemmas.

\begin{lemma}\label{lemma:extractstochterm}
  $
  \mathit{extract}_{\rightarrow,\text{term}}(x)
  \leq_{s,\text{term}}
  \mathit{extract}_{\rightarrow,\text{term}}(\mathit{step}_{\rightarrow}(x))$
\end{lemma}
\begin{proof}
  Consider first $x \in \X_s$.
  Clearly, $\mathit{extract}_{\rightarrow,\text{term}}(x) = ()$ by definition and the result holds immediately.
  Now consider $x \not\in \X_s$.
  By definition, $\mathit{step}_{\rightarrow}(x) = x$ and the result holds.
\end{proof}

\begin{lemma}
  $
  \mathit{extract}_{\rightarrow,\text{weight}}(x)
  \leq_{s,\text{weight}}
  \mathit{extract}_{\rightarrow,\text{weight}}(\mathit{step}_{\rightarrow}(x))$
\end{lemma}
\begin{proof}
  Analogous to Lemma~\ref{lemma:extractstochterm}.
\end{proof}

Applying Lemma~\ref{lemma:final} twice, we get the measurable functions $\mathit{final}_{\rightarrow,\text{term}}$ and $\mathit{final}_{\rightarrow,\text{weight}}$.
\begin{definition}
  $\mathit{final}_{\rightarrow,\text{term}} =
    \mathit{final}_{\rightarrow,\mathit{extract}_{\rightarrow,\text{term}}} :
    (\X, \Xcal)
    \to (\T,\Tcal)$
\end{definition}
\begin{definition}
  $\mathit{final}_{\rightarrow,\text{weight}} =
    \mathit{final}_{\rightarrow,\mathit{extract}_{\rightarrow,\text{weight}}} :
    (\X, \Xcal)
    \to (\R_+,\Bor_+)$
\end{definition}

\begin{lemma}
  $
    r_\term(s) = \mathit{final}_{\rightarrow,\text{term}}(\term,1,s)
  $
\end{lemma}
\begin{proof}
  By construction.
\end{proof}

\begin{lemma}
  $
    f_\term(s) = \mathit{final}_{\rightarrow,\text{weight}}(\term,1,s)
  $
\end{lemma}
\begin{proof}
  By construction.
\end{proof}

\repeatlemma{lemma:rfmeas}{\lemmarfmeas}
\begin{proof}
  Because $\mathit{final}_{\rightarrow,\text{term}}$ is measurable.
\end{proof}

\subsection{The Measurable Functions $r_{\term,n}$ and $f_{\term,n}$}%
\label{sec:rfnmeasurable}
In this section, we prove that $r_{\term,n}$ and $f_{\term,n}$ are measurable.
We follow the proof strategy from Borgström et al.~\cite{borgstrom2015lambda}.

We start with the following definition.
\newcommand{\Ncal}{\mathcal{P} (\mathbb{N}_0)}
\begin{definition}
  $\mathbb{Y} = \X \times \mathbb{N}_0$ and
  $\mathcal{Y} = \mathbb{Y} \otimes \Ncal$.
\end{definition}
Also, we require the following simple lemma.
\begin{lemma}\label{lemma:n0separable}
  $(\mathbb{N}_0,\dr|_{\mathbb{N}_0})$ is a separable metric space
\end{lemma}
\begin{proof}
  $\mathbb{N}_0$ is a countable dense subset of itself.
\end{proof}
We can now define sets corresponding to the domains for the rules $\textsc{Stoch-Fin}$ and $\textsc{Resample-Fin}$, and a domain for the relation $\hookrightarrow$ as a whole.
\begin{definition}
  \begin{equation}
    \begin{aligned}
      \mathbb{Y}_\textsc{Stoch-Fin} &=
      (\X_\textsc{Det}
      \cup \X_\textsc{Sample}
      \cup \X_\textsc{Weight}) \times \mathbb{N}_0 \\
      \mathbb{Y}_\textsc{Resample-Fin} &=
      \X_\textsc{Resample} \times \mathbb{N} \\
      \mathbb{Y}_s &=
      \mathbb{Y}_\textsc{Stoch-Fin}
      \cup \mathbb{Y}_\textsc{Resample-Fin}
    \end{aligned}
  \end{equation}
\end{definition}
\begin{lemma}
  $\mathbb{Y}_\textsc{Stoch-Fin},
  \mathbb{Y}_\textsc{Resample-Fin},
  $ and $\mathbb{Y}_s$ are all $\mathcal{Y}$-measurable.
\end{lemma}
\begin{proof}
  $\mathbb{Y}_\textsc{Stoch-Fin}$ and $\mathbb{Y}_\textsc{Resample-Fin}$ are Cartesian products of measurable sets, hence measurable.
  $\mathbb{Y}_s$ is a finite union of measurable sets, hence measurable.
\end{proof}
The following functions correspond to the rules $\textsc{Stoch-Fin}$ and $\textsc{Resample-Fin}$.
\begin{definition}
  \begin{equation}
    \begin{aligned}
      \mathit{step}_\textsc{Stoch-Fin} &=
      (\rightarrow \circ \ \pi_1, \pi_2)|_{\mathbb{Y}_\textsc{Stoch-Fin}} \\
      \mathit{step}_\textsc{Resample-Fin} &=
      (\mathit{step}_{\textsc{Resample}} \circ \pi_1,
      (n \mapsto n-1) \circ \pi_2)|_{\mathbb{Y}_\textsc{Resample-Fin}}
    \end{aligned}
  \end{equation}
\end{definition}

\begin{lemma}\label{lemma:yunion}
  $\hookrightarrow = \mathit{step}_\textsc{Stoch-Fin} \cup \mathit{step}_\textsc{Resample-Fin}$.
\end{lemma}
\begin{proof}
  By inspection.
\end{proof}

\begin{lemma}
  $\hookrightarrow$ is a function.
\end{lemma}
\begin{proof}
The domains of the functions $\mathit{step}_\textsc{Stoch-Fin}$ and $\mathit{step}_\textsc{Resample-Fin}$ are disjoint.
It follows that $\hookrightarrow$ is a function.
\end{proof}

\begin{definition}
  \begin{equation}
    \mathcal{Y}_\textsc{Stoch-Fin} =
    \mathcal{Y}|_{\mathbb{Y}_\textsc{Stoch-Fin}},
    \quad
    \mathcal{Y}_\textsc{Resample-Fin} =
    \mathcal{Y}|_{\mathbb{Y}_\textsc{Resample-Fin}},
    \quad
    \mathcal{Y}_s =
    \mathcal{Y}|_{\mathbb{Y}_s}
  \end{equation}
\end{definition}

\begin{lemma}
    $\mathcal{Y}_\textsc{Stoch-Fin}, \mathcal{Y}_\textsc{Resample-Fin}, $ and $\mathcal{Y}_s$ are $\sigma$-algebras.
\end{lemma}
\begin{proof}
  By Lemma~\ref{lemma:measrest}.
\end{proof}

\begin{lemma}
  Let $d_\mathbb{Y} = d_\X + d_\R|_{\mathbb{N}_0}$.
  Then
  \begin{equation}
    \begin{gathered}
      \Bor_{(\mathbb{Y},d_\mathbb{Y})}
      = \mathcal{Y},
      \quad
      \Bor_{(\mathbb{Y}_\textsc{Stoch-Fin},d_\mathbb{Y})}
      = \mathcal{Y}_\textsc{Stoch-Fin},
      \\
      \Bor_{(\mathbb{Y}_\textsc{Resample-Fin},d_\mathbb{Y})}
      = \mathcal{Y}_\textsc{Resample-Fin},
      \quad
      \Bor_{(\mathbb{Y}_s,d_\mathbb{Y})}
      = \mathcal{Y}_s
    \end{gathered}
  \end{equation}
  Furthermore, $(\mathbb{Y},d_\mathbb{Y})$ is a separable metric space.
\end{lemma}
\begin{proof}
  Follows from Lemmas~\ref{lemma:xsep}, \ref{lemma:n0separable}, and \ref{lemma:boreqmeas}.
\end{proof}

\begin{lemma}\label{lemma:stochfinmeas}
  $\mathit{step}_\textsc{Stoch-Fin} : (\mathbb{Y}_\textsc{Stoch-Fin}, \mathcal{Y}_\textsc{Stoch-Fin}) \to (\mathbb{Y},\mathcal{Y})$ is measurable.
\end{lemma}
\begin{proof}
  The projections $\pi_1$ and $\pi_2$ are clearly continuous and hence measurable.
  Furthermore, from Lemma~\ref{lemma:arrowsmeas}, $\rightarrow$ is measurable.
  Because restrictions of measurable functions to measurable sets are measurable, and because compositions of measurable functions are measurable, the result follows.
\end{proof}

\begin{lemma}
  $\mathit{step}_\textsc{Resample-Fin} : (\mathbb{Y}_\textsc{Resample-Fin}, \mathcal{Y}_\textsc{Resample-Fin}) \to (\mathbb{Y},\mathcal{Y})$ is measurable.
\end{lemma}
\begin{proof}
  From Lemma~\ref{lemma:stepresampmeas}, it holds that $\mathit{step}_\textsc{Resample}$ is measurable.
  Clearly, $(n \mapsto n - 1) : \mathcal{P}(\mathbb{N}) \to \mathcal{P}(\mathbb{N}_0)$ is measurable (in fact, \emph{every} function $f : (X,\mathcal{P}(X)) \to (Y,\mathcal{P}(Y))$ is measurable by the definition of $\mathcal{P}$).
  Now, by using the same argument as in the proof of Lemma~\ref{lemma:stochfinmeas}, the result follows.
\end{proof}

\begin{lemma}
  $\hookrightarrow : (\mathbb{Y}_s, \mathcal{Y}_s) \to (\mathbb{Y},\mathcal{Y})$ is measurable.
\end{lemma}
\begin{proof}
  Follows from Lemma~\ref{lemma:yunion} and Lemma~\ref{lemma:decompmeas}.
\end{proof}
With the measurability of $\hookrightarrow$ in place, we can define extract functions analogously to Appendix~\ref{sec:rfmeasurable}.
\begin{definition}
  \begin{equation}
    \mathit{extract}_{\hookrightarrow,\text{term}}(\term,w,s,n) =
    \begin{cases}
      \term & \text{if } \term \in \mathbb{V}, s = \trace{} \\
      \term & \text{if } \term = \econt[\ttt{resample}], s = \trace{}, n = 0 \\
      () & \text{otherwise}
    \end{cases}
  \end{equation}
\end{definition}
\begin{definition}
  \begin{equation}
    \mathit{extract}_{\hookrightarrow,\text{weight}}(\term,w,s,n) =
    \begin{cases}
      w & \text{if } \term \in \mathbb{V}, s = \trace{} \\
      w & \text{if } \term = \econt[\ttt{resample}], s = \trace{}, n = 0 \\
      0 & \text{otherwise}
    \end{cases}
  \end{equation}
\end{definition}

\begin{lemma}\label{lemma:extracthookstermmeas}
  $\mathit{extract}_{\hookrightarrow,\text{term}} :
  (\mathbb{Y},\mathcal{Y})
  \to
  (\T, \Tcal)$
  is measurable.
\end{lemma}
\begin{proof}
  Let $Y =
  ((\mathbb{V} \times \R_+ \times \Tr_0) \times \mathcal{P}(\mathbb{N}_0))
  \cup
  ((\T_\textsc{Resample} \times \R_+ \times \Tr_0) \times \{0\})
  $.
  Clearly, $Y \in \mathcal{Y}$.
  We have $\mathit{extract}_{\hookrightarrow,\text{term}} = (\pi_1 \circ \pi_1)|_Y \cup ()|_{X^c}$, where $()$ here denotes the constant function producing $()$.
  Because $\id$, $()$, and $Y$ are measurable, the result follows by Lemma~\ref{lemma:decompmeas}.
\end{proof}

\begin{lemma}
  $\mathit{extract}_{\hookrightarrow,\text{weight}} : (\mathbb{Y},\mathcal{Y}) \to (\R_+, \Bor_+)$ is measurable.
\end{lemma}
\begin{proof}
  Analogous to the proof of Lemma~\ref{lemma:extracthookstermmeas}.
\end{proof}

\begin{definition}
  $\leq_{\hookrightarrow,\text{term}}$ is the least partial order on $\T$ such that $() \leq_{\hookrightarrow,\text{term}} \term$.
\end{definition}
\begin{definition}
  $\leq_{\hookrightarrow,\text{weight}}$ is the least partial order on $\R_+$ such that $0 \leq_{\hookrightarrow,\text{term}} w$.
\end{definition}

With $\mathit{step}_{\hookrightarrow} = \  \hookrightarrow \cup \ \id|_{\mathbb{Y} \setminus \mathbb{Y}_s}$, we have the below lemmas.
\begin{lemma}\label{lemma:extracthookstochterm}
  $\mathit{extract}_{\hookrightarrow,\text{term}}(a) \leq_{\hookrightarrow,\text{term}} \mathit{extract}_{\hookrightarrow,\text{term}} (\mathit{step}_{\hookrightarrow}(a))$
\end{lemma}
\begin{proof}
  Consider first $y \in \mathbb{Y}_s$.
  Clearly, $\mathit{extract}_{\rightarrow,\text{term}}(y) = ()$ by definition and the result holds immediately.
  Now consider $y \not\in \mathbb{Y}_s$.
  By definition, $\mathit{step}_{\rightarrow}(x) = x$ and the result holds.
\end{proof}

\begin{lemma}
  $\mathit{extract}_{\hookrightarrow,\text{weight}}(a)
  \leq_{\hookrightarrow,\text{weight}}
  \mathit{extract}_{\hookrightarrow,\text{weight}}
  (\mathit{step}_{\hookrightarrow}(a))$
\end{lemma}
\begin{proof}
  Analogous to Lemma~\ref{lemma:extracthookstochterm}.
\end{proof}

Applying Lemma~\ref{lemma:final} twice, we get the measurable functions $\mathit{final}_{\hookrightarrow,\text{term}}$ and $\mathit{final}_{\hookrightarrow,\text{weight}}$.
\begin{definition}
  $
  \mathit{final}_{\hookrightarrow,\text{term}} =
  \mathit{final}_{\hookrightarrow,\mathit{extract}_{\hookrightarrow,\text{term}}} :
  (\mathbb{Y},\mathcal{Y})
  \to (\T,\Tcal)
  $
\end{definition}
\begin{definition}
  $
  \mathit{final}_{\hookrightarrow,\text{weight}} =
  \mathit{final}_{\hookrightarrow,\mathit{extract}_{\hookrightarrow,\text{weight}}} :
  (\mathbb{Y},\mathcal{Y})
  \to (\R_+,\Bor_+)
  $
\end{definition}

\begin{lemma}
  $
  r_{\term,n}(s) =
  \mathit{final}_{\hookrightarrow,\text{term}}((\term,1,s),n)
  $
\end{lemma}
\begin{proof}
  By construction.
\end{proof}

\begin{lemma}
  $
  f_{\term,n}(s) =
  \mathit{final}_{\hookrightarrow,\text{weight}}((\term,1,s),n)
  $
\end{lemma}
\begin{proof}
  By construction.
\end{proof}

\repeatlemma{lemma:rfnmeas}{\lemmarfnmeas}
\begin{proof}
  Because $\mathit{final}_{\hookrightarrow,\text{term}}$ is measurable.
\end{proof}

\subsection{The Kernels $k_{\term,n}$ and the Measures $\langle \term \rangle_n$}%
\label{sec:kernelproof}
In the below, assume a fixed $\term$ and $n \in \mathbb{N}$.

\repeatlemma{lemma:prependmeas}{\lemmaprependmeas}
\begin{proof}
  We prove this by showing that $\x{prepend}_s$ is continuous.
  Pick an arbitrary $s' \in \Tr$ and $\varepsilon > 0$.
  For any $s'' \in \Tr$, it holds that
  \begin{equation}
    d_\Tr(s \ast s', s \ast s'')
    = d_\Tr(s', s'')
  \end{equation}
  Hence, we can select $\delta = \varepsilon$, and the function is continuous.
\end{proof}

\begin{lemma}\label{lemma:finiteaux}
  If $p: \Tr \to \{0,1\}$ is such that $p(s)=1$ implies that $p(s')=0$ for all proper prefixes and extensions $s'$ of $s$, then
  \begin{equation}
    \int_{{[0,1]}^k} p \ d\mu_\Tr \le 1-\sum_{i<k}
    \int_{{[0,1]}^i} p \ d\mu_\Tr.
  \end{equation}
  It follows that $\int_\Tr p \ d\mu_\Tr \le 1$.
\end{lemma}
\begin{proof}
  First, note that
  \begin{equation}
    \int_{{[0,1]}^k} p \ d\mu_\Tr \leq
    \int_{{[0,1]}^k} 1 \ d\mu_\Tr = 1.
  \end{equation}
  Second, using that proper extensions $s'$ of $s$ have $p(s') = 0$ if $p(s) = 1$, we have
  \begin{equation}
    \begin{aligned}
      \int_{{{[0,1]}^k}} p \ d\mu_\Tr
      &=
      1 - \int_{p^{-1}(\{0\}) \cap {{[0,1]}^k}} 1 \ d\mu_\Tr
      \\ &\leq
      1 - \sum_{i<k}\int_{(p^{-1}(\{1\}) \cap {{[0,1]}^i}) \times {[0,1]}^{k-i}}
      1 \ d\mu_\Tr
      \\ &=
      1 - \sum_{i<k}\int_{p^{-1}(\{1\}) \cap {{[0,1]}^i}} 1 \ d\mu_\Tr
      \\ &=
      1 - \sum_{i<k}\int_{{[0,1]}^i} p \ d\mu_\Tr.
    \end{aligned}
  \end{equation}
  Thus, we have the first part of the lemma.
  For the second part,
  \begin{equation}
    \begin{aligned}
      \int_\Tr p \ d\mu_\Tr &=
      \lim_{k\to\infty}
      \left(
      \sum_{i<k+1} \int_{{[0,1]}^i} p \ d\mu_\Tr
      \right)
      \\ &\leq
      \lim_{k\to\infty}
      \left(
      \sum_{i<k} \int_{{[0,1]}^i} p +
      1 - \sum_{i<k} \int_{{[0,1]}^i} p \ d\mu_\Tr\right)
      = 1.
    \end{aligned}
  \end{equation}
\end{proof}

\repeatlemma{lemma:decompose}{\lemmadecompose}
\begin{proof}
Assume that $f_{\term,n}(s) > 0$.  Then, we have evaluated exactly $|s|$ calls to $\ttt{sample}_D$ during evaluation before reaching the $n+1$th \ttt{resample}.
Consider now the density $f_{\term,n-1}$.
Clearly, by the definition of $\hookrightarrow$ and $\rightarrow$, there is exactly one $\underline{s} \prec s$ such that $f_{\term,n-1} > 0$.
For any $s'$ such that $\underline{s} \neq s' \prec s$, the trace either depletes and gets stuck on a $\ttt{sample}_D$ before reaching the $n$th $\ttt{resample}$, or the trace will not be empty when reaching the $n$th resample.
The result follows.

Now, assume that $f_{\term,n}(s) = 0$, and consider again the density $f_{\term,n-1}$ and an arbitrary $\underline{s} \ast \overline{s} = s$.
If $f_{\term,n-1}(\underline{s}) = 0$, the result follows immediately.
Therefore, assume that $f_{\term,n-1}(\underline{s}) > 0$.
Because $f_{\term,n}(s) = 0$, the cause for getting the weight 0 must then have occurred in between the $n$th and $n+1$th resample.
In other words, it must hold that $f_{r_{\term,n-1}(\underline{s}),1}(\overline{s}) = 0$, and the result follows.

For the last part, if $f_{\term,n}(s) > 0$, it clearly holds that $f_{r_{\term,n-1}(\underline{s}),1}(\overline{s}) > 0$, and $p_{r_{\term,n-1}(\underline{s}),1}(\overline{s}) = 1$ by definition.
\end{proof}

\repeatlemma{lemma:finkern}{\lemmafinitekernel}
\begin{proof}[Partial]
  We need to show that $k_{\term,n}(s,\cdot)$ is a measure, and that $k_{\term,n}(\cdot,S)$ is a measurable function.
  We show only the former here.
  To show that the $k_{\term,n}$ are sub-probability kernels, we also need to prove that $\sup_{s \in \Tr} k_{\term,n}(s,\Tr) \le 1$.
  First, we check the measure properties:
  \begin{enumerate}
    \item Clearly, $k_{\term,n}(s,S) \geq 0$ for all $S \in \Tr$.
    \item Also, $k_{\term,n}(s,\emptyset) = 0$.
    \item Assume ${\{S_n\}}_n$ is a pairwise disjoint sequence of sets in $\Trcal$.
      Then
      \begin{equation}
        \begin{aligned}
          k_{\term,n}(s,\bigcup_n S_n)
          &=
          \int_{prepend_s^{-1}(\bigcup_n S_n)}
          p_{r_{\term,n-1}(s),1}(s')
          \ d\mu_\Tr(s')
          \\ &=
          \sum_n
          \int_{prepend_s^{-1}(S_n)}
          p_{r_{\term,n-1}(s),1}(s')
          \ d\mu_\Tr(s')
          \\ &= \sum_n k_{\term,n}(s,S_n).
        \end{aligned}
      \end{equation}
      by properties of the inverse image.
  \end{enumerate}
  It follows that $k_{\term,n}(s,\cdot)$ is a measure.
  Second, note that if $p_{r_{\term,n-1}(s),1}(s') = 1$, then $p_{r_{\term,n-1}(s),1}(s'') = 0$ for all proper prefixes and extensions $s''$ of $s'$ (consequence of Lemma~\ref{lemma:decompose}).
  From Lemma~\ref{lemma:finiteaux}, we then have, for any $s \in \Tr$,
  \begin{equation}
    \begin{aligned}
      k_{\term,n}(s,\Tr) &=
      \int_{prepend_s^{-1}(\Tr)}
      p_{r_{\term,n-1}(s),1}(s')
      \ d\mu_\Tr(s')
      \\ &=
      \int_{\Tr}
      p_{r_{\term,n-1}(s),1}(s')
      \ d\mu_\Tr(s')
      \leq
      1.
    \end{aligned}
  \end{equation}
  Then $\sup_{s \in \Tr} k_{\term,n}(s,\Tr) \leq 1$, so the kernel is a sub-probability kernel.
\end{proof}

\repeatlemma{lemma:propdensity}{\lemmapropdensity}
\begin{proof}[Sketch]
  \begin{equation}
    \begin{aligned}
      \langle \term \rangle_n(S) &=
      \int_\Tr
      k_{\term,n}(s, S)
      f_{\term,n-1}(s)d\mu_\Tr(s)
      \\ &=
      \int_\Tr
      \left(
      \int_{\x{prepend}_s^{-1}(S)}
      p_{r_{\term,n-1}(s),1}(s')
      \ d\mu_\Tr(s')
      \right)
      f_{\term,n-1}(s)d\mu_\Tr(s)
      \\ &=
      \int_\Tr
      \left(\int_{\x{prepend}_s^{-1}(S)}
      f_{\term,n-1}(s)
      p_{r_{\term,n-1}(s),1}(s')
      \ d\mu_\Tr(s') \right)
      d\mu_\Tr(s)
      \\ &=
      \int_{\{\underline{s} \mid \underline{s} \ast \overline{s} \in S \}}
      \left(\int_{\x{prepend}_s^{-1}(S)}
        f_{\term,n-1}(\underline{s})
        p_{r_{\term,n-1}(\underline{s}),1}
        (\overline{s})
      \ d\mu_\Tr(\overline{s}) \right)
      d\mu_\Tr(\underline{s})
      \\ &=
      \int_S
      f_{\term,n-1}(\underline{s})
      p_{r_{\term,n-1}(\underline{s}),1}(\overline{s})
      d\mu_\Tr(s)
    \end{aligned}
  \end{equation}
  In the last step, $\underline{s} \ast \overline{s} = s$ is the unique decomposition given by Lemma~\ref{lemma:decompose}.
  Because there is only one such decomposition for which the integrand is non-zero, we can replace the double integral with a single integral over $\Tr$ (this needs to be made more precise).
  If there is no such unique decomposition for a certain $s$, then, also by Lemma~\ref{lemma:decompose}, the integrand for this $s$ is 0 in any case, and can be ignored.
\end{proof}

\begin{lemma}\label{lemma:kernelmeas}
  Let $(\mathbb{A},\mathcal{A},\mu)$ be a measure space, $\mu$ finite, $(\mathbb{A}',\mathcal{A}')$ a measurable space, and $k : \mathbb{A} \times \mathcal{A}' \to \R_+$ a finite kernel.
  Then
  $
    \mu'(A') = \int_\mathbb{A} k(a,A')d\mu(a)
  $
  is a finite measure on $\mathcal{A}'$.
\end{lemma}
\begin{proof}
  From linearity of the Lebesgue integral, it follows that $\mu'$ is a measure.
  Also, let
  \begin{equation}
    \sup_{a \in \mathbb{A}} k(a,\mathbb{A}') = c < \infty.
  \end{equation}
  Clearly, since $\mu$ is finite,
  \begin{equation}
    \mu'(\mathbb{A}') = \int_\mathbb{A}
    k(a,\mathbb{A}')d\mu(a) < c \cdot \mu(\mathbb{A}).
  \end{equation}
  It follows that $\mu'$ is finite.
\end{proof}
A result analogous to Lemma~\ref{lemma:kernelmeas} holds for probability kernels and measures.
That is, integrating a probability kernel over a probability measure results in a probability measure.

\repeatlemma{lemma:propfin}{\lemmapropfin}
\begin{proof}
  The proof that $\langle \term \rangle_0$ is a sub-probability measure follows from Lemma~\ref{lemma:finiteaux}, similarly to the second part of Lemma~\ref{lemma:finkern}.
  That $\langle \term \rangle_n$ is finite given finite $\llangle \term \rrangle_{n-1}$ is a direct consequence of Lemma~\ref{lemma:finkern} and Lemma~\ref{lemma:kernelmeas}.
\end{proof}

\repeatlemma{lemma:weights}{\lemmaweights}
\begin{proof}
  Consider first the case $n > 0$, and note that $f_{\langle \term \rangle_n}(s) > 0$ and Lemma~\ref{lemma:decompose} implies ${p_{r_{\term,n-1}(\underline{s}),1}(\overline{s})} = 1$.
  Now,
  \begin{equation}
    \begin{aligned}
      w_n(s)
      &= \frac{f_{\llangle \term \rrangle_n}(s)}{f_{\langle \term \rangle_n}(s)}
      =
      \frac
      {f_{\term,n}(s)}
      {f_{\term,n-1}
      (\underline{s})p_{r_{\term,n-1}(\underline{s}),1}(\overline{s})}
      \\&
      =
      \frac{f_{\term,n-1}(\underline{s})f_{r_{\term,n-1}(\underline{s}),1}(\overline{s})}
      {f_{\term,n-1}(\underline{s})p_{r_{\term,n-1}(\underline{s}),1}(\overline{s})}
      =
      f_{r_{\term,n-1}(\underline{s}),1}(\overline{s}).
    \end{aligned}
  \end{equation}

  Next, $f_{\langle \term \rangle_0}(s) > 0$ directly implies $p_{\term,0}(s) = 1$, and
  \begin{equation}
    \begin{aligned}
      w_n(s)
      &= \frac{f_{\llangle \term \rrangle_0}(s)}{f_{\langle \term \rangle_0}(s)}
      =
      \frac
      {f_{\term,0}(s)}
      {p_{\term,0}(s)}
      =
      f_{\term,0}(s)
    \end{aligned}
  \end{equation}
\end{proof}

\fi

\end{document}